\documentclass[aps,prx,reprint,longbibliography,floatfix]{revtex4-2}

\usepackage{amsmath,amssymb,amsthm,mathtools}
\usepackage{stmaryrd}
\usepackage{graphicx}
\usepackage{xcolor}
\usepackage{hyperref}
\usepackage{bm}
\usepackage{enumitem}
\usepackage{tikz}
\usetikzlibrary{positioning,calc,decorations.markings}
\tikzset{
  zxZ/.style={circle,draw=black,fill=green!25,minimum size=5.4mm,
              inner sep=0pt,font=\small},
  zxX/.style={circle,draw=black,fill=red!30,minimum size=5.4mm,
              inner sep=0pt,font=\small},
  zxH/.style={rectangle,draw=black,fill=yellow!70,minimum size=3.6mm,
              inner sep=0pt},
  zxdot/.style={circle,draw=black,fill=green!25,minimum size=2.6mm,
                inner sep=0pt},
  bnd/.style={circle,draw=black,fill=white,minimum size=2.4mm,inner sep=0pt},
  hedge/.style={draw=blue!65,line width=0.6pt,dashed},
  zxwire/.style={draw=black,line width=0.5pt},
}

\newtheorem{theorem}{Theorem}
\newtheorem{lemma}{Lemma}
\newtheorem{proposition}{Proposition}
\newtheorem{corollary}{Corollary}

\newtheorem{assumption}{Assumption}

\newcommand{\red}{\mathrm{red}}
\newcommand{\sde}{\mathrm{sde}}
\newcommand{\tmin}{t_{\min}}
\newcommand{\Zb}{\mathbb{Z}}
\newcommand{\CM}{\mathrm{CM}}

\begin{document}

\title{An Exactness Barrier for ZX-Calculus Optimization\\
of Synthesized Clifford+$T$ Circuits}

\author{Chon-Fai Kam}
\email{dubussygauss@gmail.com}
\affiliation{University Paris City and University of Reunion, Paris, France}
\affiliation{Dipartimento di Fisica e Chimica Emilio Segr\`e, Universit\`a degli Studi di Palermo, Palermo, Italy}

\author{Anuradha Mahasinghe}
\affiliation{Department of Mathematics, University of Colombo, Sri Lanka}
\affiliation{School of Physics, Mathematics and Computing, University of Western Australia, Australia}
\affiliation{University Paris City and University of Reunion, France}

\author{Kaushika De Silva}
\affiliation{School of Physics, Mathematics and Computing, University of Western Australia, Australia}
\affiliation{University Paris City and University of Reunion, France}

\author{Frederic Cadet}
\affiliation{University Paris City and University of Reunion, France}
\affiliation{PEACCEL, AI for Biologics, Paris, France}

\author{Jingbo Wang}
\affiliation{School of Physics, Mathematics and Computing, University of Western Australia, Australia}

\date{\today}

\begin{abstract}
Gate synthesis and circuit optimization are studied largely
separately, and the evidence on how well they compose is
contradictory: ZX-calculus rewriting removes a substantial and
strikingly stable fraction of a Solovay--Kitaev circuit, yet almost
nothing from the number-theoretically synthesized circuits that
practice actually uses. We show that both behaviours follow from a
single bound. For any optimizer that preserves the implemented
element exactly, which includes all sound ZX rewriting with
extraction, the achievable $T$-count is bounded below by the
denominator exponent of the synthesized ring element. This exactness
barrier is computable per instance, and it separates exact
post-processing from approximation-aware resynthesis by a certified
factor reaching $101\times$ at recursion depth five. The two
behaviours are then the barrier operating at different distances from
the floor. For Solovay--Kitaev circuits we prove that the purely
local layer of ZX simplification, spider fusion and identity removal,
computes exactly the free-product normal form of $\Zb_2 * \Zb_8$,
which gives an exact per-instance accounting of the available
compression and, under a calibrated ergodicity hypothesis, a
depth-independent limit law whose parameter-free prediction we
confirm on two independently constructed nets without fitting. For
number-theoretically synthesized circuits the floor is already
saturated: on single-qubit words we prove that automated ZX
simplification attains it exactly, and we identify the mechanism as a
closed-form formula for the minimal $T$-count in terms of phase
linkage through the $Z$-axis normalizer. At two qubits and beyond the
same valuation yields unconditional per-instance rigidity
certificates, which on the quantum-Shannon-decomposition plus
\texttt{gridsynth} pipeline certify $99.4$--$99.9\%$ of the
synthesized $T$-count as incompressible, with rigidity strengthening
as the accuracy target tightens. This explains, and predicts the size
of, a near-null optimization result recently reported for that
pipeline.
\end{abstract}

\maketitle

\section{Introduction}
\label{sec:intro}

Fault-tolerant quantum computation restricts native operations to a
discrete gate set, typically Clifford+$T$: the Clifford group
(Hadamard, phase, and CNOT) augmented by the non-Clifford $T$ gate.
Clifford operations are comparatively cheap under quantum error
correction; the $T$ gate requires magic-state distillation or
cultivation, an overhead that dominates the resource cost of most
fault-tolerant
algorithms~\cite{nielsen-chuang,bravyi-kitaev2005,bravyi-haah,gidney-cultivation},
and that propagates directly into end-to-end resource
estimates~\cite{gidney-ekera}. Minimizing $T$-count is accordingly
a central objective of quantum circuit compilation, and the literature that pursues it splits into two
halves that rarely cite one another.

The first half is concerned with producing a Clifford+$T$ circuit at
all. The Solovay--Kitaev theorem, due to Kitaev building on
unpublished work of Solovay~\cite{kitaev1997} and made algorithmic
by Dawson and Nielsen~\cite{dawson-nielsen}, guarantees that any
target unitary can be approximated to precision $\epsilon$ by a gate
count $O(\log^c(1/\epsilon))$, $c\approx3.97$, via a recursive
group-commutator construction; the construction is prized for its
generality and its guaranteed convergence, not for the compactness
of what it outputs, and its circuits are well known to be long and
structurally redundant. A second, number-theoretic route instead
searches a ring of algebraic integers for a near-shortest exact
representative: Kliuchnikov, Maslov, and Mosca gave an exact
synthesis algorithm~\cite{kliuchnikov2013}, and Ross and Selinger's
\texttt{gridsynth}~\cite{ross-selinger2016} extended it to
approximate synthesis within $O(\log\log(1/\epsilon))$ of the
information-theoretic optimum $3\log_2(1/\epsilon)$ at the
single-qubit level~\cite{selinger2015,kmm2013prl,kmm2015}; the resulting circuits are
written in Matsumoto and Amano's normal form~\cite{matsumoto-amano},
whose $T$-optimality
is characterized through the exact $SO(3)$ Bloch representation by
Giles and Selinger~\cite{giles-selinger}, a characterization we
verify independently and exhaustively in Sec.~\ref{sec:prelim}.

At $n$ qubits the resource landscape is less settled, and the
contrast matters for what follows, since the pipelines used in
practice are multi-qubit ones. General lower
bounds on non-Clifford resources are due to Beverland, Campbell,
Howard, and Kliuchnikov~\cite{beverland2020} and to Low,
Kliuchnikov, and Schaeffer~\cite{low-kliuchnikov-schaeffer}, recent
upper-bound improvements to Tan~\cite{tan2025}, and the gap between
the two remains open. The exact $T$-count of a given multi-qubit
unitary can be determined by the exhaustive search of Gheorghiu,
Mosca, and Mukhopadhyay~\cite{gmm2022}, but at a cost exponential in
that count, which places the circuits studied here well out of
range. Closing the general gap is not the subject of this paper,
though we return to it in Sec.~\ref{sec:discussion}.

The second half of the literature takes a synthesized circuit as
given and asks how to make it smaller. The ZX-calculus, introduced
by Coecke and Duncan as a diagrammatic instantiation of categorical
quantum mechanics~\cite{coecke-duncan2011} and shown complete for
stabilizer quantum mechanics by Backens~\cite{backens2014} and, for
real stabilizers via pivoting, by Duncan and
Perdrix~\cite{duncan-perdrix2013}, represents a circuit as an
interacting network of ``spiders'' subject to a sound and complete
set of graphical rewrite rules. Two such rewriting strategies,
graph-like simplification (\texttt{full\_reduce}) and the finer
phase-teleportation strategy (\texttt{teleport\_reduce}), are due to
Duncan, Kissinger, Perdrix, and van~de~Wetering~\cite{duncan-kissinger2020}
and to Kissinger and van~de~Wetering~\cite{kissinger-tcount}, and
are implemented in the open-source library
\texttt{PyZX}~\cite{pyzx}, the de facto reference
implementation we use throughout. Non-ZX approaches to the same
objective include $T$-depth optimization via matroid
partitioning~\cite{amy-maslov-mosca}, compilation through
Reed--Muller-style reductions~\cite{heyfron-campbell}, treatment of
$T$ gates as $\pi/4$ Pauli rotations~\cite{zhang-chen}, and more
recent faster algorithms for $T$-count
reduction~\cite{vandaele2024}.

Because these two halves are studied separately, the question of how
much ZX post-processing recovers from synthesis's structural waste
-- and whether the answer depends on which synthesis algorithm
produced the circuit -- has, to our knowledge, no quantitative or
theoretical treatment spanning both. The empirical picture,
meanwhile, is contradictory. On the number-theoretic side, Hao, Xu,
and Tannu, in the course of proposing a tensor-network synthesis
method, observe that applying established ZX/\texttt{PyZX}
optimization to \texttt{gridsynth}-synthesized circuits yields
little, the gains being small and confined to a minority of cases,
and leave the anomaly unexplained~\cite{trasyn2025}.
On the Solovay--Kitaev side the behaviour is the reverse and, we
will show, equally in need of explanation: the same class of
optimizer removes a substantial but strikingly stable fraction of
the circuit, a fraction that does not grow as the accuracy target
is tightened. What neither observation identifies is what, if
anything, connects them: why the same optimizer, applied to
circuits computing the same class of objects to the same class of
accuracies, should behave in essentially opposite ways depending
only on which synthesis algorithm produced its input.

We show that the two behaviours are not in tension but are the same
fact seen from opposite ends of the synthesis literature: what
sound, semantics-preserving ZX rewriting can remove from a
synthesized circuit is governed by the quantization geometry of the
synthesis algorithm that produced it, not by the target unitary or
by any property of the circuit considered in isolation. Underlying
this is a single bound, which we state first because both
behaviours turn out to be instances of it. For any optimizer that
preserves the
implemented element exactly, a class that includes all sound ZX
rewriting followed by extraction, the $T$-count it can reach is
bounded below by the exact denominator exponent of the synthesized
ring element (Theorem~\ref{thm:B}); we verify this floor
exhaustively against Matsumoto--Amano~\cite{matsumoto-amano} ground
truth on all $18{,}384$ Clifford+$T$ elements with $T$-count up to
eight. This \emph{exactness barrier} separates exact rewriting from
approximation-aware resynthesis by a gap that grows geometrically
with accuracy and exceeds $100\times$ at the depths we test, and the
two contradictory behaviours are the barrier operating at different
distances from the floor.

Solovay--Kitaev synthesis leaves a large distance, and the question
is how much of it local rewriting recovers. Answering it requires
a change of category. The output of a synthesis algorithm looks
like a circuit and is treated as one, fed to an optimizer whose
yield is then measured; but a single-qubit Clifford+$T$ circuit is
equally a word over an alphabet whose letters satisfy relations,
phases adding modulo eight and $H$ squaring to the identity, and
asking what a rewriting system can remove from such a word is a
word problem in a group. Read that way the question has an
algebraic answer rather than an experimental one, and the answer is
exact: we show that Tier-1 simplification -- spider fusion and
identity removal, the rules available before any global
graph-theoretic strategy is invoked -- computes exactly the normal
form of the free product $\Zb_2 * \Zb_8$
(Lemma~\ref{lem:freeprod}),
which gives an exact, per-instance accounting of the compression a
Solovay--Kitaev circuit admits (Theorem~\ref{thm:A}) and, once the
unconditional dynamics of the recursion's residual scale are
combined with a calibrated hypothesis on the ergodicity of its
direction coordinate, a depth-independent stationary limit law
(Theorem~\ref{thm:Aprime}). That law makes a parameter-free
prediction -- that the stationary compressibility scale is set by
the square root of the synthesis net's covering radius -- which we
confirm on two independently constructed nets without adjusting any
free parameter, and which explains the depth-independence of the
compression fraction as a corollary rather than positing it.

Number-theoretic synthesis leaves no distance at all, and there the
barrier bites immediately. We show that on single-qubit words
automated ZX simplification does not merely respect the floor but
exactly saturates it, and give the mechanism: a closed-form
combinatorial formula for single-qubit $T$-count via phase linkage
through the $Z$-axis normalizer (the Linkage Theorem,
Theorem~\ref{thm:linkage}), built on a valuation lemma we prove by
finite, computer-assisted verification, from which the exact
optimality of automated ZX follows (Theorem~\ref{thm:L1}). At
$n\ge2$ qubits the same idea extends to an exact, sound collapse
criterion for $\pi/8$-rotation steps in the Pauli-channel
representation (Lemma~\ref{lem:mqval}), which certifies
incompressibility rather than merely observing it; applied to the
real quantum-Shannon-decomposition-plus-\texttt{gridsynth} pipeline
at $n=2$ and $n=3$, these certificates show that $99.4$--$99.9\%$
of synthesized $T$-count is provably irremovable -- a fraction that
grows as the accuracy target tightens -- which is why the near-null
ZX-optimization result of Ref.~\cite{trasyn2025} was correct and
exactly how large it should have been.

Kelly and Kissinger's phase squashing, which rewrites a ZX diagram
approximately with a rigorous bound on the error introduced on
graphs with generalized flow~\cite{kelly-kissinger2025}, sits on
the far side of the barrier we prove: where Theorem~\ref{thm:B} shows exact rewriting
cannot cross it, phase squashing is a tool for crossing it by
relaxing exactness, and we take up the relationship between the two
results in Sec.~\ref{sec:discussion}.

The remainder of the paper develops this account in order.
Section~\ref{sec:prelim} fixes notation, the Solovay--Kitaev
recursion, and the exact ring arithmetic used throughout.
Sections~\ref{sec:tier1}--\ref{sec:linkage} give the single-qubit
theory: free-product reduction and its exact accounting
(Sec.~\ref{sec:tier1}), the stationary limit law
(Sec.~\ref{sec:limit}), the exactness barrier
(Sec.~\ref{sec:barrier}), and the Linkage Theorem
(Sec.~\ref{sec:linkage}). Section~\ref{sec:multiqubit} extends the
valuation machinery to $n$ qubits and states the multi-qubit
rigidity result. Section~\ref{sec:numerics} describes the numerical
methodology and reports results, Section~\ref{sec:discussion}
discusses limitations and open problems, and
Section~\ref{sec:conclusion} concludes.

\section{Preliminaries}
\label{sec:prelim}

This section fixes notation and the three pieces of machinery used
throughout: the Solovay--Kitaev recursion, the ZX rewrite rules,
and the exact ring arithmetic behind our certificates.

Single-qubit words are written over the alphabet
$\Gamma = \{H\} \cup P$ with
$P = \{T, T^\dagger, S, S^\dagger, Z\}$. We write $\Gamma^*$ for the
set of all finite words over $\Gamma$, which under concatenation is
the free monoid on $\Gamma$: it carries no relations beyond
associativity and the empty word acting as an identity, so that a
map out of $\Gamma^*$ is determined by its values on single letters
once it is required to send concatenation to multiplication. Each letter of $P$ is a
$Z$-rotation, and we record its phase in units of $\pi/4$ by the
map $\phi: P \to \Zb_8$, so that $\phi(T)=1$,
$\phi(S)=2$, and $\phi(Z)=4$. Solovay--Kitaev synthesis
starts from a base net $\mathcal{N} \subset \Gamma^{\le l_0}$, the
set of all words of length at most $l_0$, chosen to approximate
$SU(2)$ to some accuracy $\epsilon_0$. Given a target $U$, the
Dawson--Nielsen recursion improves the approximation depth by
depth. At depth $k$ it returns
\begin{equation}
W_k = a_k\, b_k\, \bar a_k\, \bar b_k\, W_{k-1},
\end{equation}
where $a_k$ and $b_k$ are depth-$(k{-}1)$ words chosen so that
their group commutator balances the residual
$\Delta_k = U U_{k-1}^{-1}$: the two factors are rotations by a
common angle about orthogonal axes, and it is this balancing that
Sec.~\ref{sec:limit} turns into a contraction estimate.

The recursion is more usefully seen as a tree than as a formula.
A node at depth $k$ carrying a target spawns five children of
depth $k{-}1$: one carrying the same target one level coarser,
which we call a $U$ step, and four carrying the commutator
factors, namely $a_k$ and its inverse and $b_k$ and its inverse,
which we call $A$ and $B$ steps. Inverses are not independent
targets, so the four commutator children carry only two distinct
ones, each appearing twice, and the multiplicities this creates
are what Sec.~\ref{sec:limit} accounts for. Unfolding to the
leaves shows that $W_k$ consists of $N_k = 5^k$ \emph{blocks},
each a net word or the inverse of one, concatenated in the order
the tree prescribes. This block structure is what the free-product
analysis of Sec.~\ref{sec:tier1} is really about --- everything
interesting happens at the $5^k - 1$ junctions where consecutive
blocks meet --- and the tree itself is what Sec.~\ref{sec:limit}
turns into a probability space. Our implementation follows the
construction exactly, with a bisection solver for the balanced
commutator angle (Appendix~\ref{app:num}).

The quantity we track throughout is how much of such a word
survives rewriting. For a word $W$ we write $|W|$ for its letter
count, $\red(W)$ for its Tier-1 normal form, defined in
Sec.~\ref{sec:tier1}, and
\begin{equation}
\gamma_1(W) \;=\; 1 - |\red(W)|/|W|
\end{equation}
for the fraction of letters that Tier-1 rewriting removes, its
\emph{compression ratio}. Writing $L_k = |W_k|$, the object of
Secs.~\ref{sec:tier1} and~\ref{sec:limit} is the behaviour of
$\gamma_1(W_k)$ as the recursion deepens.

\begin{figure*}[t]
\centering
\begin{tikzpicture}[scale=0.90,every node/.style={transform shape}]

\node[anchor=west,font=\footnotesize\itshape] at (0,1.15) {(a) Clifford$+T$ generators};
\node[zxZ] (t) at (0.75,0.35) {$\tfrac{\pi}{4}$};
\draw[zxwire] (0.2,0.35)--(t); \draw[zxwire] (t)--(1.3,0.35);
\node[font=\footnotesize] at (0.75,-0.3) {$T=Z_1$};
\node[zxZ] (s) at (2.55,0.35) {$\tfrac{\pi}{2}$};
\draw[zxwire] (2.0,0.35)--(s); \draw[zxwire] (s)--(3.1,0.35);
\node[font=\footnotesize] at (2.55,-0.3) {$S=Z_2$};
\node[zxZ] (z) at (4.35,0.35) {$\pi$};
\draw[zxwire] (3.8,0.35)--(z); \draw[zxwire] (z)--(4.9,0.35);
\node[font=\footnotesize] at (4.35,-0.3) {$Z=Z_4$};
\node[zxH] (hh) at (6.05,0.35) {};
\draw[zxwire] (5.5,0.35)--(hh); \draw[zxwire] (hh)--(6.6,0.35);
\node[font=\footnotesize] at (6.05,-0.3) {$H$};

\node[anchor=west,font=\footnotesize\itshape] at (9.0,1.15)
  {(b) \textbf{Tier 1}: fusion, identity removal};
\node[zxZ] (f1) at (9.7,0.35) {$\alpha$};
\node[zxZ] (f2) at (10.8,0.35) {$\beta$};
\draw[zxwire] (9.15,0.35)--(f1); \draw[zxwire] (f1)--(f2);
\draw[zxwire] (f2)--(11.35,0.35);
\node[font=\small] at (11.8,0.35) {$=$};
\node[zxZ] (f3) at (12.7,0.35) {$\alpha{+}\beta$};
\draw[zxwire] (12.15,0.35)--(f3); \draw[zxwire] (f3)--(13.5,0.35);
\node[zxZ] (i1) at (14.7,0.35) {$0$};
\draw[zxwire] (14.15,0.35)--(i1); \draw[zxwire] (i1)--(15.25,0.35);
\node[font=\small] at (15.7,0.35) {$=$};
\draw[zxwire] (16.1,0.35)--(17.0,0.35);

\node[anchor=west,font=\footnotesize\itshape] at (0,-1.5)
  {(c) \textbf{Tier 2}: local complementation at $u$; pivoting $=\star u\,\star v\,\star u$};
\node[zxZ,minimum size=4.8mm,label={[font=\scriptsize]above:$u$}] (u) at (1.1,-2.5) {$\pm\tfrac{\pi}{2}$};
\node[zxdot] (n1) at (0.25,-3.25) {};
\node[zxdot] (n2) at (1.1,-3.55) {};
\node[zxdot] (n3) at (1.95,-3.25) {};
\draw[hedge] (u)--(n1); \draw[hedge] (u)--(n2); \draw[hedge] (u)--(n3);
\node[font=\small] at (2.85,-3.0) {$\longmapsto$};
\node[zxdot] (m1) at (3.85,-3.25) {};
\node[zxdot] (m2) at (4.7,-3.55) {};
\node[zxdot] (m3) at (5.55,-3.25) {};
\draw[hedge] (m1)--(m2); \draw[hedge] (m2)--(m3);
\draw[hedge] (m1) to[bend left=30] (m3);
\node[font=\footnotesize,anchor=west,align=left] at (6.1,-3.1)
  {$u$ is removed and its\\ neighbourhood complemented};

\node[anchor=west,font=\footnotesize\itshape] at (9.0,-1.5)
  {(d) \textbf{Tier 3}: phase gadget (teleportation)};
\draw[zxwire] (9.4,-2.6)--(11.9,-2.6);
\node[zxdot] (g1) at (10.15,-2.6) {};
\node[zxdot] (g2) at (11.15,-2.6) {};
\node[zxZ,minimum size=4.8mm] (gp) at (10.65,-3.55) {$\alpha$};
\draw[hedge] (g1)--(gp); \draw[hedge] (g2)--(gp);
\node[font=\footnotesize,anchor=west,align=left] at (12.3,-3.0)
  {a non-Clifford phase is carried off the wire, so that\\
   phases distant in the circuit may share one slot};

\end{tikzpicture}
\caption{The rewriting rules used in this paper, stratified into the
three tiers whose separate contributions our results isolate. Green
circles are $Z$ spiders carrying a phase, yellow boxes are Hadamard
gates, solid lines are wires, and dashed blue lines are Hadamard
edges. (a) A Clifford+$T$ word becomes a chain of phase spiders and
Hadamard boxes, with $Z_m$ the spider of phase $m\pi/4$. (b) Tier~1
is the purely local layer: two adjacent same-colour spiders fuse by
adding phases, and a zero-phase spider of degree two is an identity
and is deleted. Section~\ref{sec:tier1} shows that these two rules
compute exactly the normal form of $\Zb_2 * \Zb_8$. (c) Tier~2 adds
the graph operations, local complementation at a $\pm\pi/2$ spider
and pivoting, which is three local complementations along an edge.
(d) Tier~3 adds phase teleportation, which moves a non-Clifford
phase into an off-wire gadget so that phases far apart in the
circuit can be merged into a common slot; the mechanism by which
this attains the exact $T$-count floor at one qubit is the subject
of Sec.~\ref{sec:linkage}. Full definitions, with the side
conditions each rule carries, are collected in
Appendix~\ref{app:zx}.}
\label{fig:zxrules}
\end{figure*}
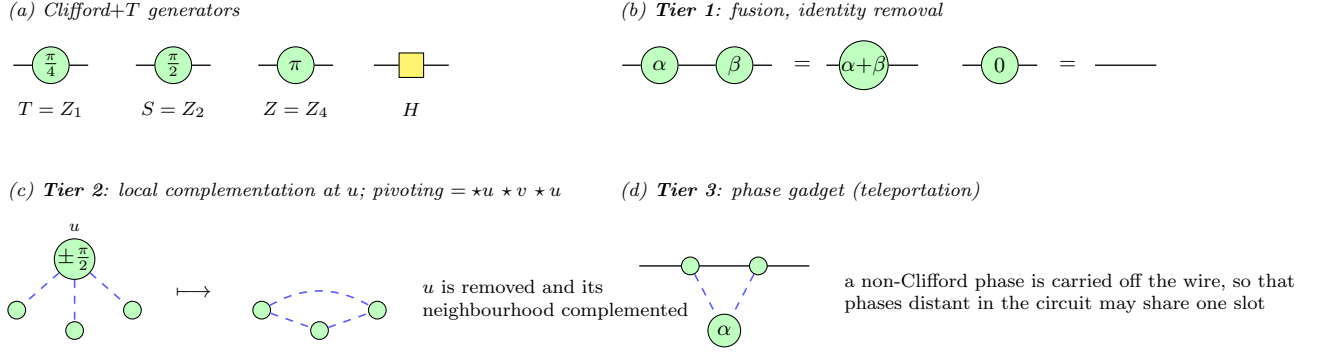

On the rewriting side we use the standard ZX-calculus rules ---
spider fusion, identity removal, local complementation, and
pivoting~\cite{coecke-duncan2011,duncan-perdrix2013,duncan-kissinger2020}
--- together with the two strategies implemented in
\texttt{PyZX}: \texttt{full\_reduce}, which performs global
graph-like simplification, and \texttt{teleport\_reduce}, which
adds phase teleportation through non-local phase
gadgets~\cite{pyzx,kissinger-tcount}. The calculus is complete for
Clifford+$T$ quantum mechanics~\cite{jeandel-perdrix-vilmart}, and
we follow the conventions of the standard
expositions~\cite{coecke-kissinger-book,vandewetering-review}. Our results attribute
different phenomena to different strata of this rule set, so we
give the strata names. \emph{Tier 1} consists of spider fusion and
identity removal alone, with no global graph strategy.
\emph{Tier 2} adds local complementation and pivoting.
\emph{Tier 3} adds phase teleportation. Figure~\ref{fig:zxrules}
displays the rules of each tier; what separates the tiers in
practice is shown in Fig.~\ref{fig:tiers} below, once the exact
$T$-count is available to measure them against. Tier 1 is the
purely local layer, and Sec.~\ref{sec:tier1} characterizes its
action exactly.
What the automated strategies add on top turns out, for
single-qubit words, to be exactly the normalizer linkage of
Sec.~\ref{sec:linkage} and nothing else (Theorem~\ref{thm:L1}).

Everything we prove about incompressibility rests on exact
arithmetic.\label{sec:exact-arith} Unitary entries of single-qubit
Clifford+$T$ words live in the ring $\Zb[\omega]$ with
$\omega = e^{i\pi/4}$, the setting in which exact synthesis is
decided~\cite{kliuchnikov2013}. The real $SO(3)$ Bloch
representation and the $n$-qubit Pauli-channel representation live
in the smaller ring $\Zb[\sqrt2]$, and it is there that our
invariant lives.

Because $\{1,\sqrt2\}$ is a basis of $\Zb[\sqrt2]$ as a free
$\Zb$-module, every matrix $M$ over $\Zb[\sqrt2]$ can be written
\begin{equation}
\label{eq:ringform}
M \;=\; \frac{P + Q\sqrt2}{\sqrt2^{\,e}},
\end{equation}
with $P$ and $Q$ integer matrices, and for a given $M$ the
representation with $e$ least is unique. We call that minimal $e$
the \emph{least denominator exponent} $\mathrm{lde}(M)$. Deciding
whether a given representation is minimal is elementary: dividing
numerator and denominator of~\eqref{eq:ringform} by $\sqrt2$ sends
$(P, Q) \mapsto (Q, P/2)$, which stays integral exactly when
\begin{equation}
\label{eq:redcrit}
P \equiv 0 \pmod 2 \quad\text{entrywise},
\end{equation}
so a representation can be reduced precisely
when~\eqref{eq:redcrit} holds, and repeating the step until it
fails yields $\mathrm{lde}(M)$. Criterion~\eqref{eq:redcrit} is
used repeatedly below: it is what makes the barrier of
Sec.~\ref{sec:barrier} computable, and refining it to a statement
about which \emph{steps} can lower the exponent is the content of
Lemma~\ref{lem:valuation} at one qubit and
Lemma~\ref{lem:mqval} at $n$.

We store every matrix in the form~\eqref{eq:ringform}, reducing
after each operation, including, for $n$-qubit circuits, the exact
Pauli-conjugation channel of each gate. Products, sums, and the
reduction step are then integer operations, so no rounding ever
enters the computation (implementation and validation in
Appendix~\ref{app:num}).

\begin{figure*}[t]
\centering
\begin{tikzpicture}[scale=0.93,every node/.style={transform shape}]

\node[anchor=west,font=\footnotesize\itshape] at (-0.2,0.78)
  {(a) the word $T\,H\,S\,S\,H\,T$, whose exact $T$-count is zero};
\draw[zxwire] (0.15,0)--(5.75,0);
\node[zxZ] at (0.7,0) {$\tfrac{\pi}{4}$};
\node[zxH] at (1.6,0) {};
\node[zxZ] at (2.5,0) {$\tfrac{\pi}{2}$};
\node[zxZ] at (3.4,0) {$\tfrac{\pi}{2}$};
\node[zxH] at (4.3,0) {};
\node[zxZ] at (5.2,0) {$\tfrac{\pi}{4}$};
\node[font=\footnotesize,anchor=west] at (6.6,0)
  {six syllables, two $T$ gates};

\node[anchor=west,font=\footnotesize\itshape] at (-0.2,-1.05)
  {(b) Tier 1 fuses the two $S$ spiders and then halts: the $T$ gates are untouched};
\draw[zxwire] (0.15,-1.83)--(4.85,-1.83);
\node[zxZ] at (0.7,-1.83) {$\tfrac{\pi}{4}$};
\node[zxH] at (1.6,-1.83) {};
\node[zxZ] at (2.5,-1.83) {$\pi$};
\node[zxH] at (3.4,-1.83) {};
\node[zxZ] at (4.3,-1.83) {$\tfrac{\pi}{4}$};
\node[font=\footnotesize,anchor=west] at (6.6,-1.83)
  {five syllables, two $T$ gates};

\node[anchor=west,font=\footnotesize\itshape] at (-0.2,-2.88)
  {(c) Tiers 2 and 3 recognise that $H\pi H$ preserves the $Z$ axis and annihilate the pair};
\draw[zxwire] (0.15,-3.66)--(1.25,-3.66);
\node[zxX] at (0.7,-3.66) {$\pi$};
\node[font=\footnotesize,anchor=west] at (1.75,-3.66)
  {$=\ \tmin$, the exact floor};
\node[font=\footnotesize,anchor=west] at (6.6,-3.66)
  {one syllable, no $T$ gates};

\end{tikzpicture}
\caption{What each tier can and cannot do, shown on a single word.
(a) The word $T\,H\,S\,S\,H\,T$ implements a Clifford, so its
minimal $T$-count is zero, but nothing about the syllable string
reveals this. (b) Tier 1 sees only adjacent syllables of the same
colour. It fuses $S$ with $S$ into a $\pi$ spider and then has no
further move, leaving both $T$ gates in place: local rewriting
alone certifies nothing about the $T$-count. (c) The interior
Clifford $HSSH$ equals $X$, which preserves the $Z$ axis up to
sign, so the two $T$ phases can be brought together and cancel.
Recognising this is exactly what the global strategies of Tiers 2
and 3 supply, and Sec.~\ref{sec:linkage} shows that at one qubit
they supply neither more nor less: what they find is precisely the
transport of phases through the $Z$-axis normalizer.}
\label{fig:tiers}
\end{figure*}
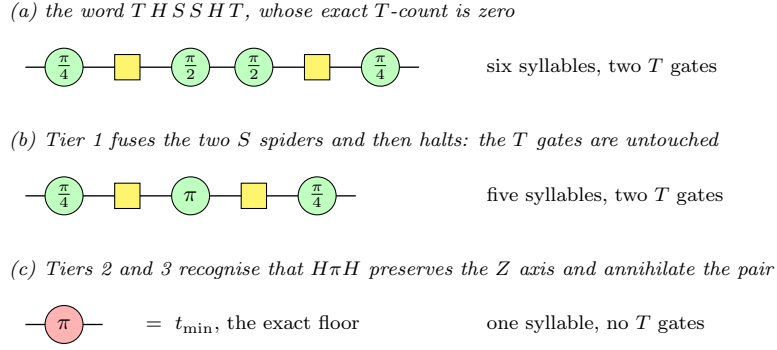

Carrying this arithmetic buys us a bridge between the algebra and
the resource count. For a single-qubit Clifford+$T$ element $V$,
write
$\hat V_{ij} = \tfrac12\,\mathrm{tr}(\sigma_i V \sigma_j
V^\dagger)$ for its exact $SO(3)$ Bloch representation and
$\mathrm{lde}(\hat V)$ for its least denominator exponent in the
sense of~\eqref{eq:ringform}.

\begin{lemma}\label{lem:bridge}
The minimal $T$-count of a single-qubit Clifford+$T$ element $V$
is
\begin{equation}
\tmin(V) = \mathrm{lde}(\hat V).
\end{equation}
\end{lemma}

This is the Matsumoto--Amano $T$-optimality
characterization~\cite{matsumoto-amano,giles-selinger}, whose
algorithmic side is the exact-synthesis procedure of Kliuchnikov,
Maslov, and Mosca~\cite{kliuchnikov2013}, read
through the Bloch representation as in Giles and
Selinger~\cite{giles-selinger}. Every certificate in this paper
ultimately rests on it, so we do not take it on faith. Enumerating
all Matsumoto--Amano normal forms with $T$-count $t \le 8$ yields
$18{,}384$ distinct projective Clifford+$T$ elements, and the
relation holds on every one of them (Appendix~\ref{app:num}).
Since $\hat V$ is computable with linearly many exact ring
operations in the word length, the lemma turns the minimal
$T$-count --- in general the object
one is trying to find --- into a quantity one can evaluate
exactly, instance by instance. We use it throughout as the exact
floor against which optimizer output is judged.

\section{Local Rewriting as Free-Product Reduction}
\label{sec:tier1}

With the exact floor in hand, we turn to the weakest layer of
rewriting and ask what it actually computes on a synthesized word.
Tier 1 sees a circuit as a string of letters and acts only on
adjacent pairs. It can add the phases of neighbouring rotations
and it can cancel a doubled Hadamard, and it can do nothing else.
In particular it is blind to which unitaries the letters
represent: any identity of the concrete group
$\langle H, T\rangle \le U(2)$ that does not already follow from
``phases add modulo $8$'' and ``$H$ squares to the identity'' is
out of its reach. There is a canonical algebraic object with
exactly those relations and no others, the free product
$\Zb_2 * \Zb_8$, and the natural guess is that Tier 1 computes
precisely its word problem. The first lemma of this section
confirms the guess exactly, and the rest of the section converts
it into a quantitative account of how much an SK word can shrink.

\begin{lemma}\label{lem:freeprod}
Let $F = \Zb_2 * \Zb_8$ and let $\pi: \Gamma^* \to F$ be the
morphism of monoids, that is, the map sending concatenation of
words to multiplication in $F$, determined by $\pi(H) = h$ (the
generator of the $\Zb_2$ factor) and $\pi(p) = z^{\phi(p)}$ for
$p \in P$ ($z$ the generator of the $\Zb_8$ factor). The Tier-1 rewriting system is
terminating and confluent; its unique normal form $\red(w)$ is
computable in $O(|w|)$ time and equals the free-product normal
form of $\pi(w)$.
\end{lemma}

\begin{proof}
It is convenient to work over the syllable alphabet
$\Sigma = \{H\} \cup \{Z_m : m \in \Zb_8 \setminus \{0\}\}$, into
which $\Gamma$ embeds by $p \mapsto Z_{\phi(p)}$; a general
Tier-1 configuration (a ZX spider chain after any number of
fusions) is a word over $\Sigma$, since a $Z$-spider may carry any
phase in $\Zb_8$. On $\Sigma^*$ the Tier-1 rules are exactly
\begin{align}
\text{(F)}\quad & Z_a Z_b \;\to\;
\begin{cases}
Z_{a+b}, & a+b \not\equiv 0 \pmod 8,\\[2pt]
\varepsilon, & a+b \equiv 0 \pmod 8,
\end{cases}\\
\text{(H)}\quad & HH \;\to\; \varepsilon,
\end{align}
where (F) is spider fusion followed, in the second case, by
identity removal. We prove the four claims in turn.

\emph{Termination.} Rule (F) replaces two letters by one or by
none, and rule (H) removes two letters, so every rule application
strictly decreases word length. A word of length $|w|$ therefore
admits at most $|w|$ rule applications along any reduction
sequence, and every sequence terminates.

\emph{Local confluence.} Two rule applications at disjoint
positions of a word commute: each acts only on its own two-letter
window, so applying them in either order yields the same word.
It therefore suffices to resolve the overlapping cases, of which
there are exactly two, since an (F)-redex consists of two
$Z$-letters and an (H)-redex of two $H$-letters, so redexes of
different types can never share a letter.

\emph{Case 1: $Z_a Z_b Z_c$}, with the left redex $Z_aZ_b$ and the
right redex $Z_bZ_c$. All congruences are modulo $8$. There are
four subcases.
(i)~If $a+b \not\equiv 0$ and $b+c \not\equiv 0$: the left
reduction gives $Z_{a+b}Z_c$, the right gives $Z_aZ_{b+c}$; one
further (F)-step on each side gives $Z_{a+b+c}$ if
$a+b+c \not\equiv 0$ and $\varepsilon$ if $a+b+c \equiv 0$, in
both cases the same word.
(ii)~If $a+b \equiv 0$ and $b+c \not\equiv 0$: the left reduction
gives $Z_c$ outright; the right gives $Z_aZ_{b+c}$, and since
$a + (b+c) \equiv c \not\equiv 0$ (as $c \ne 0$ on $\Sigma$),
one further (F)-step gives $Z_c$.
(iii)~The case $a+b \not\equiv 0$, $b+c \equiv 0$ is symmetric to
(ii) and joins at $Z_a$.
(iv)~If $a+b \equiv 0$ and $b+c \equiv 0$, then
$a \equiv c \equiv -b$, and both reductions delete a cancelling
pair outright, leaving $Z_c$ on the left and $Z_a$ on the right;
since $a \equiv c$ these are the same word.
In every subcase the two reducts have a common descendant, by
nothing more than associativity and commutativity of addition in
$\Zb_8$.

\emph{Case 2: $HHH$}, with redexes at positions $\{1,2\}$ and
$\{2,3\}$. Both reductions give the word $H$; the peak joins
immediately.

Since every critical peak is joinable, the system is locally
confluent, and by Newman's lemma~\cite{newman1942} a terminating,
locally confluent system is confluent. Hence every word has a unique normal form
$\red(w)$.

\emph{Linear-time computation.} Process $w$ left to right,
maintaining a stack of already-reduced output: to absorb the next
letter, compare it with the top of the stack; if they are of
different types (or the stack is empty), push; if both are
$Z$-letters, replace the top by the fused letter or pop it when
the phases cancel; if both are $H$, pop. After a pop or fusion the
new top is compared with the next incoming letter only, never
rescanned, because the stack contents are alternating at all times
by construction. Each input letter is pushed at most once and
popped at most once, so the total work is $O(|w|)$, and the stack
contents on termination are reduced, hence equal to $\red(w)$ by
confluence.

\emph{Identification with the free-product normal form.} A word
over $\Sigma$ is Tier-1--irreducible if and only if no two
adjacent letters have the same type, i.e., if and only if it is an
alternating sequence of syllables drawn from
$\{h\} = \Zb_2 \setminus \{e\}$ and
$\{z^m : m \ne 0\} = \Zb_8 \setminus \{e\}$. These alternating
words are precisely the normal forms of the free product
$\Zb_2 * \Zb_8$: by the normal form theorem for free
products~\cite{lyndon-schupp,magnus-karrass-solitar,serre-trees},
every element of $F$ is represented
by exactly one such word. It remains to check that reduction
preserves the image in $F$: rule (F) holds in $F$ because
$z^a z^b = z^{a+b}$ within the $\Zb_8$ factor (and $z^az^b = e$
when $a + b \equiv 0$), and rule (H) holds because $h^2 = e$ in
the $\Zb_2$ factor. Hence $\pi(\red(w)) = \pi(w)$, and since
$\red(w)$ is an alternating word, it \emph{is} the free-product
normal form of $\pi(w)$.
\end{proof}

Lemma~\ref{lem:freeprod} cuts both ways. Tier-1 simplification
realizes every relation of $F$ --- but only those: any relation of
the concrete group $\langle H, T\rangle \le U(2)$ beyond those of
$F$ is invisible to it, because by the normal form theorem
distinct normal forms represent distinct elements of $F$, so
Tier~1 identifies two words if and only if they have the same
image under $\pi$. A minimal example is the Euler-type identity
$(HZ)^4 \propto \mathbb{1}$: since $HZ$ is the Bloch rotation
$R_y(\pi/2)$, the word $(HZ)^4$ is projectively trivial, yet it is
already alternating --- eight syllables, $\pi$-image
$(hz^4)^4 \ne e$ in $F$ --- and Tier~1 leaves it entirely
untouched. Detecting such relations requires the
Clifford-complete rules of Tiers 2--3, whose additional power we
characterize exactly in Sec.~\ref{sec:linkage}.

Lemma~\ref{lem:freeprod} turns Tier-1 compression of an SK word
into a cancellation problem in the free product. Write the
depth-$k$ word in its block decomposition
$W_k = \beta_1 \beta_2 \cdots \beta_{N_k}$ with $N_k = 5^k$, and
let $v_i = \red(\beta_i)$ be the $i$-th block after internal
reduction. By confluence, $\red(W_k)$ may be computed by first
reducing every block internally and then absorbing
$v_1, \ldots, v_{N_k}$ left to right with the stack algorithm from
the proof of Lemma~\ref{lem:freeprod}. Compression therefore has
exactly two sources. Letters removed inside a block are a property
of the net alone. Letters removed where reduced blocks meet are
not, because cancellation at a junction can in principle cascade.
An incoming block may eat deep into the stack, so the saving
realized at one junction can depend on the entire prefix absorbed
before it. The next lemma controls this cascade unconditionally,
with no regularity assumption on the net.

Some bookkeeping first. Lengths in this section are syllable
counts, written $\|\cdot\|$. During the absorption of
$v_1, \ldots, v_{N_k}$, let $\hat s_i$ be the number of letters
removed while $v_i$ is being absorbed. This counts letters of
$v_i$ itself as well as letters popped off the stack, with the
convention that when two syllables merge into one, the incoming
letter counts as removed and the stack letter as surviving. Let
$\mathrm{surv}_j$ be the number of letters of $v_j$ that survive
to $\red(W_k)$, and let
\begin{equation}
\CM \;=\; \sum_j \big(\|v_j\| - \mathrm{surv}_j\big)
\end{equation}
be the total consumed mass. Finally, for reduced words $v$ and $w$
let
\begin{equation}
s(v,w) \;=\; \|v\| + \|w\| - \|\red(vw)\|
\end{equation}
be the saving realized by fusing the pair in isolation. This is
the quantity one would sum over junctions if cancellation never
cascaded, and the lemma measures how far that naive sum can drift
from the truth.

\begin{lemma}\label{lem:densweight}
Unconditionally:
\begin{enumerate}[nosep]
\item[(a)] $\sum_{i=1}^{N_k} \hat s_i = \CM$ and
$\|\red(W_k)\| = \sum_{i=1}^{N_k} \|v_i\| - \CM$.
\item[(b)] The discrepancy
$R = \sum_{i=1}^{N_k} \hat s_i -
\sum_{i=1}^{N_k-1} s(v_i, v_{i+1})$
satisfies $|R| \le \CM$.
\end{enumerate}
\end{lemma}

\begin{proof}
(a) Every letter of every block meets exactly one of two fates. It
survives to $\red(W_k)$, or it is removed at exactly one step of
the absorption. Summing $\hat s_i$ over $i$ therefore counts each
removed letter exactly once, whether it fell as part of the
incoming block or was popped off the stack, so
$\sum_i \hat s_i$ is the total number of removed letters, which is
$\sum_j (\|v_j\| - \mathrm{surv}_j) = \CM$ by definition. The
stack on termination is reduced and equals $\red(W_k)$ by
confluence, and its length is the number of letters pushed minus
the number removed. Rearranging gives the second identity. The
merge convention decides how a removal is attributed to a block,
never whether it is counted, so neither identity depends on it.

(b) The bound has an easy direction and a hard direction. The easy
direction is $R \le \CM$. Reduction never lengthens a word, so
$s(v,w) \ge 0$ for every pair, hence
$\sum_i s(v_i, v_{i+1}) \ge 0$, and by part (a)
\begin{equation}
R \;=\; \CM - \sum_{i} s(v_i, v_{i+1}) \;\le\; \CM .
\end{equation}
The hard direction is
$\sum_i s(v_i, v_{i+1}) \le 2\,\CM$, and it rests on two
structural facts about how the stack algorithm absorbs a reduced
block.

\emph{Step 1: consumption is a prefix.} Let $v$ be a reduced block
being absorbed. Its letters are pushed left to right, and the
first letter may trigger a cascade of pops against the stack. But
the moment some letter of $v$ is pushed and survives, every later
letter of $v$ is simply pushed: the stack top is then the previous
letter of $v$, and adjacent letters of a reduced word have
different types, so no rule can fire between them. The letters of
$v$ removed during its own absorption therefore form a prefix of
$v$, and what remains of $v$ on the stack is a contiguous suffix.

\emph{Step 2: at most one syllable is altered.} During the same
cascade, stack syllables are either removed outright or, in the
single case of a fusion whose result is nonzero, have their phase
changed while surviving. Such a surviving fusion ends the cascade,
because the next letter of $v$ has a different type from the
letter just merged and is simply pushed. So the absorption of one
block alters the phase of at most one surviving stack syllable,
namely the deepest syllable the cascade reaches, and it leaves
that syllable in place rather than removing it --- a point on
which the estimate of Appendix~\ref{app:proofs} turns.

\emph{Step 3: comparing the pairwise and global reductions.} Fix
a junction $i$ and reduce the pristine pair $v_i v_{i+1}$ in
isolation. Because both factors are reduced, the pops occur in
nested order outward from the junction: the $m$-th pop cancels or
merges the pair $(x_m, y_m)$, where $x_m$ is the $m$-th syllable
of $v_i$ counted from the right and $y_m$ the $m$-th syllable of
$v_{i+1}$ counted from the left. If $M_i$ pops occur then
$s(v_i, v_{i+1}) \le 2M_i$, since each pop removes at most two
letters. Now compare with the global process. When $v_{i+1}$
arrives, the stack carries a suffix of $v_i$ of some length
$\sigma_i \ge 0$, pristine except that its deepest syllable may
carry an altered phase (Steps 1 and 2), on top of older material.

The depth $\sigma_i$ splits the pairwise pops into three regimes.
Above it, at depths $m < \sigma_i$, both members of the pair sit
on the stack in pristine form, so the global cascade fires exactly
when the pairwise reduction does and removes the same letters. At
$m = \sigma_i$ the two can disagree, since the stack copy of
$x_{\sigma_i}$ may carry an altered phase, but this is at most one
pop per junction. Below it, at depths $m > \sigma_i$, the pairwise
reduction reaches syllables of $v_i$ that were consumed before
$v_{i+1}$ ever arrived, so these pops have no global counterpart
at this junction at all.

The third regime is where the pairwise sum can outrun the global
one, and it is also where the bound comes from, because the
syllables those pops use were removed at the junction before.
Appendix~\ref{app:proofs} makes this precise junction by junction:
the saving available to the pairwise reduction at junction $i$ is
at most the global saving there plus the global saving at the
junction before it (Proposition~\ref{prop:audit}). Summing over
junctions gives $\sum_i s(v_i, v_{i+1}) \le 2\,\CM$, and with the
easy direction $|R| \le \CM$ (Corollary~\ref{cor:audit}), a bound
an explicit family attains.
\end{proof}

Part (b) is not the bound we first wrote down, and the correction
is worth recording. An earlier draft charged pairwise savings
only to \emph{fully consumed blocks}, and an adversarial search
falsified it: a long block, a single-syllable barrier, and the
inverse of the long block's tail produce deep cancellation that
is invisible junction by junction while no block is ever fully
consumed. The syllable-level statement above is what that
counterexample forced. We then stress-tested it on $3065$
instances, real SK runs on both nets at $k = 1$--$5$ together
with $3000$ adversarial synthetic block sequences including the
falsifying family, with zero violations, and the worst observed
ratio is $|R|/(\CM{+}1) = 0.909$, consistent with the sharp
coefficient proven in Appendix~\ref{app:proofs} and with the
family exhibited there that attains it.

Lemma~\ref{lem:densweight} localizes the analysis. Part (a) says
that Tier-1 compression \emph{is} the consumed mass. Part (b) says
that the naive junction-by-junction sum computes the consumed mass
up to an error which is itself of the order of the consumed mass.
What remains is to identify when the error vanishes outright, that
is, when cancellation stays at its own junction instead of
cascading. Cascades need one of two things: blocks short enough to
be swallowed whole, or degenerate junctions where adjacent blocks
are mutual inverses. Ruling out both turns out to be enough.

Two conditions make the localization precise. Say the net is
\emph{regular} if
\begin{equation}
\ell_{\min} \;\ge\; 2D^* + 1,
\end{equation}
where $\ell_{\min}$ is the syllable length of the shortest reduced
net word and $D^*$ is the maximal pairwise cancellation depth over
ordered pairs of reduced net words and their inverses, with
mutual-inverse pairs excluded from the maximum. Say an instance is
\emph{non-degenerate} if no two adjacent blocks of $W_k$ are exact
mutual inverses. Regularity says that every net word is more than
twice as long as any cancellation it can suffer at a single
junction. Non-degeneracy removes the one junction type at which
cancellation depth is unbounded, the type that regularity, by
construction, says nothing about. Three more statistics of the
instance appear in the result. We write $\bar\rho_k$ for the
internal reduction density,
\begin{equation}
\bar\rho_k \;=\; \frac{1}{L_k}\sum_{i=1}^{N_k}
\big(|\beta_i| - \|v_i\|\big),
\end{equation}
the fraction of the raw word removed by reducing each block in
isolation, $J^{PP}_k$ and $J^{HH}_k$ for the number of junctions
at which two phase syllables, respectively two Hadamards, meet,
and $l_0' = \min_\nu |\nu|$ for the raw length of the shortest net
word.

\begin{theorem}\label{thm:A}
For a regular net and a non-degenerate instance, $R = 0$ and the
consumed mass localizes to junctions:
\begin{gather}
\gamma_1(W_k) \;=\; \bar\rho_k \;+\; \frac{1}{L_k}
\sum_{i=1}^{N_k-1} s\big(v_i, v_{i+1}\big),
\\[2pt]
\bar\rho_k + \frac{J^{PP}_k + 2J^{HH}_k}{L_k}
\;\le\; \gamma_1(W_k) \;\le\; \bar\rho_k + \frac{2D^*}{l_0'} .
\end{gather}
\end{theorem}

\begin{proof}
The heart of the matter is that under the two conditions the
global absorption never cascades past a single junction, so that
the stack process realizes exactly the pairwise cancellations and
nothing else. We prove this by induction on the block index,
carrying the invariant that after $v_i$ has been absorbed, the
stack ends in an untouched suffix of $v_i$ of syllable length at
least $D^* + 1$.

The base case is the absorption of $v_1$ onto the empty stack.
Nothing is removed, so $\hat s_1 = 0$, and the untouched suffix is
all of $v_1$, of length
$\|v_1\| \ge \ell_{\min} \ge 2D^* + 1 \ge D^* + 1$.

For the inductive step, suppose the invariant holds after $v_i$
and absorb $v_{i+1}$. Consider first the pairwise reduction of the
pristine pair $v_i v_{i+1}$, whose pops cancel or merge the nested
pairs $(x_1, y_1), \ldots, (x_{M_i}, y_{M_i})$ as in
Lemma~\ref{lem:densweight}. By non-degeneracy $v_{i+1}$ is not the
exact inverse of $v_i$, so this pair enters the maximum defining
$D^*$ and $M_i \le D^*$. Since
$\|v_i\|, \|v_{i+1}\| \ge \ell_{\min} > D^* \ge M_i$, neither word
is exhausted, and the pairwise reduction stops at depth $M_i + 1$
because the pair $(x_{M_i+1}, y_{M_i+1})$ fails to fire. Now run
the global cascade. By the invariant, the top $D^* + 1 \ge M_i +
1$ syllables of the stack are the pristine tail of $v_i$, so at
every depth $m \le M_i + 1$ the global process compares exactly
the same pair of syllables as the pairwise one. It therefore fires
the same $M_i$ pops, removes the same letters, and stops at the
same depth. Two consequences follow. First,
$\hat s_{i+1} = s(v_i, v_{i+1})$, since $v_{i+1}$ is internally
reduced and every removal during its absorption belongs to this
cascade. Second, the cascade reaches at most $D^*$ syllables into
$v_i$ and consumes at most $D^*$ letters of $v_{i+1}$, so the
untouched suffix of $v_{i+1}$ left on the stack has length at
least $\ell_{\min} - D^* \ge D^* + 1$, which is the invariant for
the next step.

The inductive step gives $\hat s_j = s(v_{j-1}, v_j)$ for every
$j = 2, \ldots, N_k$, and $\hat s_1 = 0$ from the base case, so
\begin{equation}
\label{eq:R0}
\sum_{i=1}^{N_k} \hat s_i
\;=\; \sum_{i=1}^{N_k-1} s(v_i, v_{i+1}),
\end{equation}
that is, $R = 0$, and by Lemma~\ref{lem:densweight}(a) the
consumed mass is exactly the sum of the pairwise junction savings.

Equation~\eqref{eq:R0} is now letter accounting. The letters of
$W_k$ removed by Tier-1 reduction split into those removed inside
a block, $\sum_i (|\beta_i| - \|v_i\|)$ in total, and those
removed at junctions, which by \eqref{eq:R0} is
$\sum_i s(v_i, v_{i+1})$, where we use that a reduced word has one
letter per syllable so that syllable and letter counts agree on
everything downstream of the internal reduction. Dividing by
$L_k$ gives
$\gamma_1(W_k) = \bar\rho_k + L_k^{-1}\sum_i s(v_i,v_{i+1})$.

For the lower bound, examine the junction types. At a junction
where two phase syllables meet, the fusion rule fires at least
once and removes at least one letter, so $s(v_i,v_{i+1}) \ge 1$.
Where two Hadamards meet, the cancellation $HH \to \varepsilon$
removes two, so $s(v_i,v_{i+1}) \ge 2$. All junction savings are
nonnegative, so
$\sum_i s(v_i,v_{i+1}) \ge J^{PP}_k + 2J^{HH}_k$. For the upper
bound, each junction satisfies $s(v_i, v_{i+1}) \le 2M_i \le
2D^*$, there are $N_k - 1$ junctions, and
$L_k \ge N_k\, l_0'$ because every block descends from a net word
of raw length at least $l_0'$. Hence
\begin{equation}
\frac{1}{L_k}\sum_{i=1}^{N_k-1} s(v_i, v_{i+1})
\;\le\; \frac{(N_k - 1)\, 2D^*}{N_k\, l_0'}
\;\le\; \frac{2D^*}{l_0'} . \qedhere
\end{equation}
\end{proof}

The two bounds say something stronger than an error bar. The
junction term is squeezed between a count of junction types over
$L_k$ and the constant $2D^*/l_0'$, and both sides scale as the
reciprocal of the net-word length, uniformly in $k$: the junction
count grows like $N_k$ while $L_k$ grows like $N_k$ times the mean
net-word length, which sits between $l_0'$ and $l_0$. The
recursion depth has dropped out entirely. Under regularity and
non-degeneracy, Tier-1 compressibility of a Solovay--Kitaev
circuit is a statistic of the net, not of the recursion.

This prediction can be checked without knowing $\bar\rho_k$ in
closed form, because depth-independence is visible on its own. On
both of our nets (Sec.~\ref{sec:numerics}), $\gamma_1(W_k)$ is
flat to within $\pm 0.02$--$0.04$ across $k = 1$ to $6$, even as
the SK approximation error falls by five orders of magnitude
(Fig.~\ref{fig:plateau}). The depth-independence of the
compression fraction is thereby established as a theorem about the
net rather than an empirical curiosity.

There is, however, an honest gap between the theorem and the
experiment, and closing it is the business of the next section.
The hypotheses of Theorem~\ref{thm:A} fail on real instances. On
both nets, empty and near-empty blocks occur at a stationary
$7$--$20\%$ density at every depth we test, in violation of
regularity, and the number of degenerate mutually-inverse
junctions grows into the hundreds by $k = 6$
(Appendix~\ref{app:num}). Yet Fig.~\ref{fig:plateau} shows the
plateau surviving these violations without visible damage, at the
same value, with the same flatness. The conditional mechanism of
this section cannot be the whole story. What is needed is an
unconditional account of why cascades, though present, do not
accumulate, and that account is dynamical rather than
combinatorial: it comes from the contraction properties of the SK
recursion itself.

\section{The Stationary Limit Law}
\label{sec:limit}

The plateau outlives the hypotheses that explained it, so
something else must be holding it up. This section identifies that
something. The idea is to stop treating $W_k$ as one long word and
start treating it as a random object. The recursion tree that
generates $W_k$ induces a natural probability measure on its
blocks, and under that measure the quantities entering $\gamma_1$
--- block lengths, internal reductions, junction savings ---
become expectations of functionals of a stochastic process. Three
facts then carry the argument. The block average is exactly an
expectation over independent steps down the recursion tree
(Lemma~\ref{lem:pathrep}). The scale of a node's residual
contracts affinely along branch steps, so a deep leaf forgets its
ancestry at an exponential rate, an unconditional fact
(Lemma~\ref{lem:scale}) which moreover pins the stationary scale
to the net's covering radius (Corollary~\ref{cor:sqrtlaw}). And
residuals almost surely never vanish, so the dynamics never
collapses onto the exceptional set (Lemma~\ref{lem:nocollapse}).
What we cannot prove is ergodicity of the residual's
\emph{direction}, and we isolate it as a calibrated hypothesis
(Hypothesis~E), of the same character as the equidistribution
hypotheses standard in exact synthesis. Together these yield the
limit law (Theorem~\ref{thm:Aprime}) and a parameter-free
prediction that we test in two independent ways.

The recursion has a tree structure worth making explicit before
anything is averaged. A depth-$k$ call $\mathrm{SK}(U,k)$ issues
three depth-$(k{-}1)$ calls, one on each commutator factor of its
residual and one on the target itself. The word it returns uses
the two factor words twice each, once directly and once inverted,
and the target word once, five slots in all. Unfolding to depth
$k$ therefore assigns each of the $N_k = 5^k$ blocks of $W_k$ an
address in a $5$-ary tree, but the distinct computations are
indexed by the coarser set $\{U,A,B\}^k$: record at each level
whether the block descends through the target call ($U$) or
through the first or second commutator factor ($A$ or $B$), and
forget which of the two copies was taken. For a path
$p \in \{U,A,B\}^k$, write $\mathrm{leaf}(p)$ for the base-net
word computed at the end of $p$ and $m(p)$ for its number of
branch steps, the letters of $p$ equal to $A$ or $B$.

\begin{lemma}\label{lem:pathrep}
Exactly $2^{m(p)}$ blocks of $W_k$ carry the address $p$, and each
of them equals $\mathrm{leaf}(p)$ or its inverse. Consequently,
for every block functional $f$ satisfying
$f(\bar\beta) = f(\beta)$,
\begin{gather}
\frac{1}{N_k}\sum_{i=1}^{N_k} f(\beta_i)
\;=\;
\mathbb{E}_{p\sim\mathcal{P}_k}\big[f(\mathrm{leaf}(p))\big],
\\[2pt]
\mathcal{P}_k = \Big(\tfrac15\delta_U + \tfrac25\delta_A +
\tfrac25\delta_B\Big)^{\otimes k},
\end{gather}
and under $\mathcal{P}_k$ the branch count $m(p)$ is binomial with
parameters $(k, 4/5)$, so that
\begin{equation}
\mathcal{P}_k\big[m(p) < k/2\big] \;\le\; e^{-k\ln(5/4)} .
\end{equation}
\end{lemma}

\begin{proof}
We prove the multiplicity claim by induction on $k$. At $k = 0$
the word is a single net block, the only path is the empty one,
and $2^0 = 1$. For the step, recall from Sec.~\ref{sec:prelim}
that the block sequence of $W_k$ is the concatenation of the block
sequences of $a_k$, $b_k$, $\bar a_k$, $\bar b_k$, and $W_{k-1}$,
where the block sequence of an inverted word consists of the
inverted blocks in reverse order. A path beginning with $U$ has
all its blocks in the $W_{k-1}$ segment, and by the inductive
hypothesis their number is $2^{m(p')}$ for the tail $p'$ of $p$,
which equals $2^{m(p)}$ because a $U$ step adds no branch. A path
beginning with $A$ has its blocks in the $a_k$ segment and the
$\bar a_k$ segment, each contributing $2^{m(p')}$ by the inductive
hypothesis, for a total of $2 \cdot 2^{m(p')} = 2^{m(p)}$, and the
$B$ case is identical. As for the form of these blocks, each is
obtained from $\mathrm{leaf}(p)$ by applying an inversion once for
every ``second copy'' chosen along its $5$-ary address, and an
even number of inversions cancels, so the block is
$\mathrm{leaf}(p)$ or its inverse according to the parity.

For the identity, partition the block sum by address and use the
inversion invariance of $f$, under which every block at address
$p$ contributes the same value $f(\mathrm{leaf}(p))$:
\begin{equation}
\frac{1}{N_k}\sum_{i=1}^{N_k} f(\beta_i)
\;=\;
\frac{1}{N_k}\sum_{p\in\{U,A,B\}^k} 2^{m(p)}
f(\mathrm{leaf}(p)) .
\end{equation}
It remains to observe that the weights are a product measure.
Writing $w(U) = 1$ and $w(A) = w(B) = 2$ for the per-step copy
counts, the weight of $p$ is
$2^{m(p)}/N_k = \prod_{j=1}^{k} w(p_j)/5$, which is exactly
$\mathcal{P}_k(p)$ with the stated one-step marginal. The sum is
therefore $\mathbb{E}_{\mathcal{P}_k}[f(\mathrm{leaf}(p))]$. Both
functionals used later, the block length and the reduced block
length, are inversion invariant, the first trivially and the
second because the reduction of an inverse is the inverse of the
reduction.

Under the product measure the $k$ letters of $p$ are independent
and each is a branch with probability $4/5$, so
$m(p) \sim \mathrm{Bin}(k, 4/5)$. The tail bound is the standard
Chernoff bound in relative-entropy form,
$\mathbb{P}[\mathrm{Bin}(k,q) \le \alpha k] \le
e^{-k\,D(\alpha\|q)}$, with $\alpha = 1/2$ and $q = 4/5$:
\begin{equation}
D\big(\tfrac12\big\|\tfrac45\big)
= \tfrac12\ln\tfrac{1/2}{4/5} + \tfrac12\ln\tfrac{1/2}{1/5}
= \tfrac12\ln\tfrac{25}{16}
= \ln\tfrac54 . \qedhere
\end{equation}
\end{proof}

The lemma is the probabilistic reformulation on which the rest of
the section runs. Averages over the $5^k$ blocks of a single
deterministic word are now expectations over $k$ independent
steps, and the Chernoff bound says that all but an exponentially
small fraction of the block mass lies on paths with many branch
steps. The next lemma explains why branch steps matter: each one
halves the memory of the scale.

The mechanism sits inside the balanced commutator construction.
At a node with residual rotation angle $\theta$, the recursion
factors the residual as a group commutator of two rotations by a
common angle $\varphi$ about orthogonal axes, and the two factor
rotations become the targets of the node's $A$- and $B$-children.
For such a factorization the angles are related by the
identity~\cite{dawson-nielsen}
\begin{equation}
\sin(\theta/2) \;=\; 2q^2\sqrt{1-q^4},
\qquad q = \sin(\varphi/2),
\end{equation}
which is the equation our implementation solves by bisection
(Appendix~\ref{app:num}). The identity looks innocuous, but read
as a map from the parent's scale to the child's it has a very
particular structure: on a logarithmic axis it is affine with
slope exactly one half. Slope one half is a contraction, and a
contraction iterated along the branch steps of a path is
precisely what makes deep leaves forget where they came from. The
lemma isolates this.

\begin{lemma}\label{lem:scale}
For every $\theta \in (0,\pi]$ the equation
$\sin(\theta/2) = 2q^2\sqrt{1-q^4}$ has a unique solution with
$q \in (0, 2^{-1/4}]$, the balanced branch, and $q$ is strictly
increasing in $\theta$ along it. In the coordinates
$x_\theta = \ln\sin(\theta/2)$ and $x_\varphi = \ln q$, the map
from parent scale to child scale is affine,
\begin{equation}
x_\varphi \;=\; \tfrac12\, x_\theta + c(\theta),
\qquad
c(\theta) \in \big[-\tfrac12\ln 2,\; -\tfrac14\ln 2\big],
\end{equation}
and the upper endpoint is attained at $\theta = \pi$ while the
lower is approached as $\theta \to 0$.
\end{lemma}

\begin{proof}
Write $g(q) = 2q^2\sqrt{1-q^4}$ on $[0,1]$. Its derivative is
\begin{equation}
g'(q) \;=\; 4q\sqrt{1-q^4} \;+\; 2q^2 \cdot
\frac{-2q^3}{\sqrt{1-q^4}}
\;=\; \frac{4q\,(1 - 2q^4)}{\sqrt{1-q^4}},
\end{equation}
which is strictly positive for $0 < q < 2^{-1/4}$ and vanishes at
$q = 2^{-1/4}$. So $g$ is strictly increasing on $[0, 2^{-1/4}]$,
with $g(0) = 0$ and
$g(2^{-1/4}) = 2 \cdot 2^{-1/2}\sqrt{1 - \tfrac12} = 1$. Since
$\sin(\theta/2)$ ranges over $(0,1]$ as $\theta$ ranges over
$(0,\pi]$, the equation $g(q) = \sin(\theta/2)$ has exactly one
solution in $(0, 2^{-1/4}]$, and monotonicity of $g$ makes $q$
strictly increasing in $\theta$. This is the balanced branch.

For the affine form, take logarithms of the defining identity:
\begin{equation}
x_\theta \;=\; \ln 2 + 2\ln q + \tfrac12\ln(1-q^4)
\;=\; \ln 2 + 2 x_\varphi + \tfrac12\ln(1-q^4),
\end{equation}
and solve for $x_\varphi$:
\begin{equation}
x_\varphi \;=\; \tfrac12\,x_\theta + c(\theta),
\qquad
c(\theta) \;=\; -\tfrac12\ln 2 \;-\; \tfrac14\ln(1-q^4) .
\end{equation}
On the balanced branch $q \in (0, 2^{-1/4}]$ we have
$1 - q^4 \in [\tfrac12, 1)$, hence
$\ln(1-q^4) \in [-\ln 2, 0)$ and
$-\tfrac14\ln(1-q^4) \in (0, \tfrac14\ln 2]$. Adding
$-\tfrac12\ln 2$ places $c(\theta)$ in the stated interval. As
$\theta \to 0$ we have $q \to 0$ and $c \to -\tfrac12\ln 2$, and
at $\theta = \pi$ we have $q = 2^{-1/4}$ exactly and
$c = -\tfrac12\ln 2 + \tfrac14\ln 2 = -\tfrac14\ln 2$, so the
upper endpoint is attained and the lower is approached in the
limit.
\end{proof}

Iterating the lemma is what produces stationarity, and it is
worth doing the iteration once explicitly. Compose $m$ branch
steps, each of the form $x \mapsto \tfrac12 x + c_j$ with
$c_j \in [-\tfrac12\ln2, -\tfrac14\ln2]$:
\begin{equation}
x_m \;=\; 2^{-m} x_0 \;+\; \sum_{j=1}^{m} 2^{-(m-j)} c_j .
\end{equation}
The first term says that the memory of the starting scale decays
by a factor of two per branch step, whatever the intermediate
intercepts are. The second term is a geometric sum whose value,
as $m$ grows, is confined to the interval
$[-\ln 2, -\tfrac12\ln 2]$, that is,
$\sin(\varphi/2) \in [\tfrac12, 2^{-1/2}]$. Repeated branching
therefore drives the scale into a fixed $O(1)$ window that
depends on nothing --- not the target, not the depth, not the
net. This is re-inflation: however small a node's residual has
become, the targets a few branch steps below it are back at
angles of order one. Steps of type $U$ interleave with the
branch steps and push the scale the other way, deeper, as the
recursion converges on its target, and by
Lemma~\ref{lem:pathrep} all but an $e^{-\Omega(k)}$ fraction of
the block mass lies on paths whose branch steps are frequent
enough for the contraction to win the exchange. The quantitative assembly of
these two rates is carried out where it is needed, in the proof
of Theorem~\ref{thm:Aprime}.

The window is one half of the stationarity story. The other half
is a floor: how small can a re-inflated node's residual actually
be? A node's residual is the mismatch between its target and the
best net approximation, so once the target angle is $O(1)$ the
residual cannot fall below the resolution of the net itself. Its
natural scale is the covering radius
$\theta_{\mathrm{cov}}$~\cite{harrow-recht-chuang}, the largest
angle by which a point of
$SU(2)$ can miss the net. The next corollary feeds a residual of
this size through the affine law. Its first part is an exact
consequence of Lemma~\ref{lem:scale}, valid for every node. Its
second part adds the one modeling step, the identification
$\theta_{\mathrm{res}} \asymp \theta_{\mathrm{cov}}$ for
re-inflated nodes, and is what turns the exact sandwich into a
prediction.

\begin{corollary}
\label{cor:sqrtlaw}
For any node with residual angle $\theta_{\mathrm{res}}$, the
balanced angle $\varphi$ of its children satisfies exactly
\begin{equation}
\sin^2(\varphi/2) \;\in\;
\big[\tfrac12,\, \tfrac1{\sqrt2}\big] \cdot
\sin(\theta_{\mathrm{res}}/2) .
\end{equation}
In particular, for a node whose target has re-inflated to angle
of order one, so that its residual is quantization-limited,
$\theta_{\mathrm{res}} \asymp \theta_{\mathrm{cov}}$, the
stationary leaf-angle scale is pinned to
\begin{equation}
\sin^2(\varphi^*/2) \;\in\;
\big[\tfrac12,\, \tfrac1{\sqrt2}\big] \cdot
\sin(\theta_{\mathrm{cov}}/2),
\end{equation}
with no adjustable constant.
\end{corollary}

\begin{proof}
Rearranging the defining identity of Lemma~\ref{lem:scale} gives
\begin{equation}
q^2 \;=\; \frac{\sin(\theta_{\mathrm{res}}/2)}{2\sqrt{1-q^4}},
\end{equation}
and on the balanced branch $q^4 \le \tfrac12$, so
$\sqrt{1-q^4} \in [2^{-1/2}, 1]$ and therefore
$q^2 = \sin^2(\varphi/2)$ lies between
$\tfrac12\sin(\theta_{\mathrm{res}}/2)$ and
$2^{-1/2}\sin(\theta_{\mathrm{res}}/2)$. The second display is
the first with $\theta_{\mathrm{res}}$ read at its quantization
floor.
\end{proof}

The corollary earns its place for two reasons. First, it is the
one point at which the theory of this section touches the net
quantitatively. Theorem~\ref{thm:Aprime} below will say that the
stationary compressibility is a functional of the net's
quantization measure, but a functional is not yet a number. The
corollary extracts a number: the deep-leaf angle scale must sit
in a specific window computed from $\theta_{\mathrm{cov}}$ alone,
a quantity read off the net before any circuit is synthesized.
The square root in the name is the visible fingerprint of the
slope-$\tfrac12$ contraction, since
$\sin(\varphi^*/2) \sim \sqrt{\sin(\theta_{\mathrm{cov}}/2)}$ up
to the constants of the window. Second, and more importantly for
the logic of the paper, the corollary is falsifiable where
Theorem~\ref{thm:Aprime} is not yet provable. The limit law
rests on Hypothesis~E, which we calibrate rather than prove, but
the window above uses only the proven ingredients, the
contraction and the quantization floor. Checking it therefore
tests the mechanism independently of the hypothesis, and because
nothing is fitted, agreement cannot be manufactured. We carry
out this check at the end of the section on two independently
constructed nets with covering radii differing by a factor of
almost four, and the measured deep-leaf angle statistics land in
the predicted window both times.

One assumption has been riding silently under everything in this
section. The coordinate of Lemma~\ref{lem:scale} is
$x = \ln\sin(\theta/2)$, which is finite only when the residual
angle is nonzero, and the balanced factorization itself is
defined only in that case: a node with vanishing residual has
nothing to factor, and a node arbitrarily close to vanishing
sends $x$ toward $-\infty$, outside every window discussed so
far. If residuals could vanish with positive probability, the
path dynamics would carry an atom at $-\infty$ and the stationary
picture would be broken before it began. The next lemma closes
this loophole. It says that for a Haar-random target the residual
can vanish in only one way, the harmless one: the target of some
node is exactly a product of net words, the recursion returns
that exact word, and the subtree simply ends. Away from this
event, of probability zero, every node of the tree has a nonzero
residual and the log-scale dynamics is well defined everywhere.

\begin{lemma}\label{lem:nocollapse}
For Haar-random $U$, almost surely the following holds at every
node of the recursion tree: either the node's residual is
nonzero, or the node belongs to a subtree whose root target is
exactly net-representable, in which case the recursion terminates
there with the exact word.
\end{lemma}

\begin{proof}
Fix a node $\omega$ of the (infinite) recursion tree and consider
the target of $\omega$ as a function of the root target $U$, the
Haar variable. The proof has four steps: a cell decomposition,
the root case, analyticity on cells, and an induction on depth
that runs on a dichotomy.

\emph{Step 1: cells.} The output of every call in the tree is a
word, a discrete object, so it depends on $U$ only through
finitely many discrete selections: which net word each base call
chooses, and which branch each step of the commutator solver
takes. Fixing all selections made in the subtree above and at
$\omega$ partitions $SU(2)$ into finitely many pieces, the
\emph{cells} of $\omega$, on each of which every approximation
matrix appearing in the analysis of $\omega$ is a single constant.
The cell boundaries are loci of ties, nearest-neighbour draws and
solver branch switches, each contained in a proper real-analytic
subset of $SU(2)$, and a proper analytic subset is Haar-null. We
may therefore discard all boundaries at once and work on open
cells.

\emph{Step 2: the root.} On a cell where the root approximation
is the constant $C$, the root residual is the right translation
$U \mapsto UC^{-1}$, which equals the identity exactly at the
single point $U = C$. A single point is Haar-null, and it is,
moreover, exactly the point at which the root target coincides
with a product of net words, that is, the exactly representable
case. So already at the root, a vanishing residual and exact
representability are one and the same event, and off a null set
neither occurs.

\emph{Step 3: analyticity on cells.} On an open cell, the map
$U \mapsto \mathrm{target}(\omega)$ is a finite composition of
four kinds of operations: right translations by the cell's
constant matrices, which are analytic everywhere; extraction of
the rotation angle, $\theta = 2\arccos(\mathrm{tr}/2)$, analytic
wherever $\mathrm{tr} \in (-2, 2)$; the balanced angle
$\varphi(\theta)$, analytic for $0 < \theta < \pi$ by the
implicit function theorem, since the derivative
$g'(q) = 4q(1-2q^4)/\sqrt{1-q^4}$ computed in the proof of
Lemma~\ref{lem:scale} is strictly positive on the open balanced
branch; and extraction and rotation of axes, analytic away from
the axis degeneracies. The excluded loci, trace equal to $\pm 2$,
$\theta = \pi$, and axis alignment, are each either a proper
analytic subset of the cell, hence null, or the whole cell, which
can happen only if the map into the relevant coordinate is
constant there. Constancy is not excluded at this step. It is
absorbed by the next one.

\emph{Step 4: dichotomy and induction on depth.} We show by
induction on $j$ that the set of $U$ at which some node of depth
at most $j$ has a vanishing residual without benign termination
is null. Depth zero is Step 2. For the step, fix a node $\omega$
at depth $j+1$ and one open cell, and let $T_\omega$ be the
target map of $\omega$ on that cell, analytic off the null loci
of Step 3. Two cases are possible, and this is the dichotomy.
If $T_\omega$ is nonconstant, then by the identity theorem the
preimage under $T_\omega$ of any single point is a proper
analytic subset of the cell, hence null. The residual of $\omega$
vanishes exactly on the preimage of the one point matching the
cell's constant approximation at $\omega$, so it vanishes only on
a null set, and preimages of the deeper null loci needed to
continue the induction are null for the same reason. If instead
$T_\omega$ is constant on the cell, equal to some fixed unitary
$V_\omega$, then the entire subtree rooted at $\omega$ is a
deterministic computation on the input $V_\omega$, identical for
every $U$ in the cell. Its residuals are fixed numbers. Either all of them are
nonzero, and no exclusion is needed on this cell, or some node of
the subtree has residual exactly zero, which means its target
coincides exactly with the net word chosen for it. That is exact
net-representability, the recursion returns the exact word, and
the termination is the benign case that the statement allows.
In both cases the non-benign vanishing set within the cell is
null. There are finitely many cells per node and countably many
nodes, so the union over the tree of all non-benign vanishing
sets is a countable union of null sets, which is null.
\end{proof}

What the lemma buys is permission. Theorem~\ref{thm:Aprime} below
conditions on the event that every residual along the tree is
nonzero, and Lemma~\ref{lem:nocollapse} says this conditioning
discards a set of measure zero, so no statement about
Haar-typical targets is weakened by it. The margin is also
comfortable in practice, not only in principle: across the
$1452$ internal-node evaluations of our depth-$5$ runs on both
nets, the smallest residual angle ever observed is
$4.4\times10^{-3}$, so the dynamics operates nowhere near the
degenerate set it is licensed to ignore.

It is time to take stock of what is proven and what is missing.
The state of a node is a pair, a scale and a direction: the
residual is a rotation, and the log-scale coordinate of
Lemma~\ref{lem:scale} records its angle while saying nothing
about its axis. The scale is now fully under control, contracted
by branch steps, drawn into a fixed window, and almost surely
finite everywhere. But the functionals we actually want to
average are block statistics, and a block is a net word chosen by
nearest-neighbour search, which means the choice is made by the
\emph{axis} as much as by the angle: at a fixed angle, it is the
direction of the residual that decides which Voronoi cell of the
net the target falls into and hence which word is emitted, with
what internal reduction and what junction behaviour. Stationarity
of block statistics therefore requires stationarity of the
direction, and the direction is exactly what the contraction
argument does not touch. Its dynamics is the composition of the
commutator geometry, which moves axes in a target-dependent way,
with the net's Voronoi partition, and we cannot prove that this
composition equilibrates. Rather than let the gap hide inside a
theorem, we isolate precisely what is needed and give it a name.

\begin{assumption}[Hypothesis E: directional ergodicity]
\label{ass:E}
Along $\mathcal{P}_k$-paths, the (direction, scale) chain of
residuals is ergodic in its direction coordinate, uniformly over
the scale window of Lemma~\ref{lem:scale}, at a geometric rate.
In the form used below: the laws of a
leaf block and of an adjacent leaf-block pair converge in total
variation, as the number of branch steps grows, to limits $\nu$
and $\nu_2$ that do not depend on the target, with a uniform
geometric rate.
\end{assumption}

Three comments belong next to this, because the hypothesis is the
one unproven input to everything that follows and its status
should be as visible as its statement.

First, its role. Lemma~\ref{lem:scale} makes a leaf forget the
\emph{scale} of its ancestry, and Hypothesis~E upgrades this to
forgetting the ancestry altogether: joint laws of leaves and of
adjacent leaf pairs converge to $(\nu, \nu_2)$. That is exactly
the strength needed to pass from ``block averages are path
expectations'' (Lemma~\ref{lem:pathrep}) to ``block averages
converge to expectations under a stationary law'', which is
Theorem~\ref{thm:Aprime}. Nothing weaker would do, because the
compression functional involves adjacent pairs, and nothing
stronger is assumed.

Second, why it is plausible. Two mixing mechanisms act on the
direction at every branch step. The commutator construction
rotates the residual axis through angles that depend sensitively
on the target, and the selection then passes the result through
the Voronoi partition of a net with, on our finer construction,
some two million cells whose boundaries bear no relation to the
commutator geometry. A deterministic map with sensitive
dependence composed with a fine partition unaligned with it is
the standard situation in which equidistribution is expected,
and, in the number-theoretic half of the synthesis literature, is
routinely assumed: the running-time guarantees of
Ross--Selinger~\cite{ross-selinger2016} rest on equidistribution
hypotheses about prime candidates in rings of integers that play
exactly this role. Equidistribution statements of this kind are
also what underwrite the optimal covering properties of the best
known single-qubit gate sets, where they follow from deep
arithmetic input~\cite{lps1986,parzanchevski-sarnak}; it is worth
noting that the group generated by a super-golden gate set is,
like our $\Zb_2 * \Zb_8$, a free product~\cite{parzanchevski-sarnak}.
Hypothesis~E is the same kind of statement about a different
deterministic object.

Third, what ``calibrated'' means. We do not merely hope the
hypothesis holds, we test its observable consequences and find
them realized. The proven half of the mechanism already makes the
parameter-free prediction of Corollary~\ref{cor:sqrtlaw}, and the
deep-leaf angle statistics land inside the predicted window on
both nets. The hypothesis' own signature would be drift: if the
direction failed to equilibrate, $\gamma_1(W_k)$ should wander
with depth or differ erratically between nets, and
Fig.~\ref{fig:plateau} shows neither, on two independently
constructed nets, across five orders of magnitude in accuracy. A
proof would require an equidistribution theorem for the Voronoi
quantization of the specific net construction, and we return to
what that would involve in Sec.~\ref{sec:discussion}. With the
hypothesis stated and its status in the open, the limit law can
be assembled.

One more piece of notation makes the shape of the theorem
visible before it is stated. Lemma~\ref{lem:densweight} splits the
compression ratio into a part built from junction-by-junction
data and a certified remainder. Write the first part as
\begin{equation}
\gamma_1^{\mathrm{pair}}(W_k)
\;=\; 1 - \frac{1}{L_k}\Big(\sum_{i=1}^{N_k}\|v_i\|
- \sum_{i=1}^{N_k-1} s(v_i, v_{i+1})\Big),
\end{equation}
so that, by Lemma~\ref{lem:densweight}(b),
$|\gamma_1(W_k) - \gamma_1^{\mathrm{pair}}(W_k)| = |R|/L_k \le
\CM/L_k$ on every instance, no hypothesis needed, and $R$
itself is exactly computable per instance by running the two
reductions. The
quantity $\gamma_1^{\mathrm{pair}}$ is assembled from exactly
three block averages: the mean raw block length, the mean reduced
block length, and the mean adjacent-pair saving. Each average is
a path expectation by Lemma~\ref{lem:pathrep}, each functional is
bounded because blocks are net words of length at most $l_0$, and
Hypothesis~E says the laws these expectations are taken against
converge. The theorem does nothing more than let the three
averages converge together and take the ratio, and this is why it
can be proven in full once the hypothesis is granted.

\begin{theorem}\label{thm:Aprime}
Assume $\epsilon_0$ below the SK convergence threshold and
Hypothesis~E, with mixing rate $\rho < 1$. For Haar-random $U$:
\begin{enumerate}[nosep]
\item[(a)] In probability,
\begin{gather}
\gamma_1^{\mathrm{pair}}(W_k)
\xrightarrow[k\to\infty]{} \Gamma(\nu,\nu_2),
\\[2pt]
\Gamma(\nu,\nu_2)
= 1 - \frac{\mathbb{E}_\nu\big[\|\red(\beta)\|\big] -
\mathbb{E}_{\nu_2}\big[s(\beta,\beta')\big]}
{\mathbb{E}_\nu\big[|\beta|\big]},
\end{gather}
with error $O(\rho^{\kappa k}) + e^{-\Omega(k)}$ for a constant
$\kappa > 0$, and $\Gamma(\nu,\nu_2)$ depends only on the stationary
laws $(\nu, \nu_2)$.
\item[(b)] Unconditionally,
$|\gamma_1(W_k) - \gamma_1^{\mathrm{pair}}(W_k)| \le
\CM/L_k$ on every instance, and the exactly computed
discrepancy $|R|/L_k$ runs at about $2\%$ of the word length on
our SK instances.
\end{enumerate}
\end{theorem}

\begin{proof}
Throughout, condition on the almost-sure event of
Lemma~\ref{lem:nocollapse} that every residual along the tree is
nonzero, which changes no probability. Fix $\delta \in (0,1)$ and
set $r = \lceil \delta k \rceil$. The constant $\kappa$ will come
from optimizing over $\delta$ at the end. In part (b) there is nothing
to prove beyond Lemma~\ref{lem:densweight}(b) and the definition
of $\gamma_1^{\mathrm{pair}}$, together with the measured value
reported in Sec.~\ref{sec:numerics}, so the work is part (a), and
we do it in five steps.

\emph{Step 1: typical paths.} Under $\mathcal{P}_k$ the $k$
letters of a path are independent, so the number $b$ of branch
steps among the last $r$ letters is binomial with parameters
$(r, 4/5)$, and by the tail bound of Lemma~\ref{lem:pathrep}
applied to the suffix,
\begin{equation}
\mathcal{P}_k\big[\,b < r/2\,\big] \;\le\; e^{-r\ln(5/4)} .
\end{equation}
Call a path \emph{suffix-typical} when $b \ge r/2$. All bounds
below are proved on suffix-typical paths, and the atypical mass
$e^{-r\ln(5/4)} = e^{-\Omega(k)}$ is carried along additively,
using that every functional in sight is bounded.

\emph{Step 2: the prefix dissolves.} Fix a suffix-typical path
and look at the chain of residual states along its last $r$
steps. By Lemma~\ref{lem:scale}, every branch step halves the
log-scale memory of the state at the suffix start, while the
interleaved $U$ steps refresh it by their own bounded log-Lipschitz
factor, and at the branch frequencies guaranteed by Step 1 the
contraction wins: the scale coordinate of the leaf depends on the
prefix only through a perturbation decaying geometrically in $r$.
Hypothesis~E supplies the same for the direction: its uniform
geometric mixing says that, conditionally on the state at the
suffix start, the law of the leaf block is within
$C\rho^{\,b} \le C\rho^{\,r/2}$ of $\nu$ in total variation,
\emph{uniformly} in that starting state, and likewise the law of
an adjacent leaf pair is within $C\rho^{\,r/2}$ of $\nu_2$
whenever both paths of the pair are suffix-typical. Averaging
over the prefix, whose influence the uniformity has just erased,
the unconditional law of a suffix-typical leaf is within
$C\rho^{\,r/2}$ of $\nu$, and of a suffix-typical adjacent pair
within $C\rho^{\,r/2}$ of $\nu_2$.

\emph{Step 3: expectations of the three averages.} The three
functionals are $F_1(\beta) = |\beta|$,
$F_2(\beta) = \|\red(\beta)\|$, and the pair functional
$F_3(\beta,\beta') = s(\red(\beta), \red(\beta'))$. All are
bounded by $2l_0$, since every block is a net word of raw length
at most $l_0$ and reduction does not lengthen, and $F_1, F_2$ are
inversion invariant as noted in the proof of
Lemma~\ref{lem:pathrep}, while $F_3$ satisfies
$F_3(\bar\beta', \bar\beta) = F_3(\beta, \beta')$, because
inverting a two-block segment inverts and swaps its blocks and
reduction commutes with inversion. For the
single-block averages, Lemma~\ref{lem:pathrep} converts the block
average into $\mathbb{E}_{\mathcal{P}_k}$ of a path functional,
and Steps 1 and 2 bound its distance from the stationary value:
\begin{equation}
\big|\mathbb{E}_{\mathcal{P}_k}[F_j]
- \mathbb{E}_{\nu}[F_j]\big|
\;\le\; 2l_0\big(C\rho^{\,r/2} + e^{-r\ln(5/4)}\big)
\end{equation}
for $j = 1, 2$. For the pair average, the same bound holds
junction by junction:
at every junction whose two paths are suffix-typical the pair law
is within $C\rho^{\,r/2}$ of $\nu_2$ by Step 2, the fraction of
junctions failing this is at most $2e^{-r\ln(5/4)}$ by a union
over the two paths, and boundedness of $F_3$ converts the failing
fraction into an additive error of the same order.

\emph{Step 4: fluctuations.} Convergence of expectations is not
yet convergence in probability, and this is where the geometric
rate of Hypothesis~E is used a second time. Uniform geometric
mixing makes the covariance between the functional values at two
positions decay geometrically in the tree distance between their
paths, so
the variance of each of the three empirical averages is bounded
by a constant times $N_k^{-1}$ times a summable correlation
series, hence is $O(N_k^{-1}) = O(5^{-k})$. Chebyshev's
inequality then upgrades each convergence of expectations to
convergence in probability at a rate absorbed by the
$e^{-\Omega(k)}$ term.

\emph{Step 5: the ratio.} Write
$\gamma_1^{\mathrm{pair}} = 1 - (A_k - B_k)/C_k$ with $A_k, B_k,
C_k$ the empirical averages of $F_2$, $F_3$, $F_1$. The
denominator is bounded below: every raw block is a net word of
length at least $l_0' \ge 1$, so $C_k \ge l_0'$ surely and
$\mathbb{E}_\nu|\beta| \ge l_0'$. On the event, of probability
$1 - e^{-\Omega(k)}$, that all three averages are within their
Step-3-and-4 windows, the algebra of limits with a denominator
bounded below turns the three convergences into
\begin{equation}
\big|\gamma_1^{\mathrm{pair}}(W_k) - \Gamma(\nu,\nu_2)\big|
\;=\; O\big(\rho^{\,\delta k/2}\big) + e^{-\Omega(k)} .
\end{equation}
Choosing $\delta$ to balance the two exponents fixes the constant
$\kappa$ and gives the stated rate.
\end{proof}

The theorem is the payoff of the whole section, and its content
is best read off the formula. Every quantity on the right-hand
side of $\Gamma(\nu,\nu_2)$ is an expectation under a stationary
law determined by the net's quantization measure, the way the
net's Voronoi cells carve up $SU(2)$ and which words sit in them.
The target unitary is gone. The recursion depth is gone. What
remains is the synthesis algorithm's geometry, and this is the
precise sense of the claim announced in the introduction: the
compressibility of Solovay--Kitaev is the value of a
functional of the quantization geometry of the algorithm that
produced the circuit. Theorem~\ref{thm:A} reached the same
conclusion exactly but conditionally, on hypotheses real
instances violate. Theorem~\ref{thm:Aprime} trades that
conditionality for Hypothesis~E and in exchange covers the
instances as they actually occur, explaining at once why both
nets plateau, why they plateau at different values, and why the
compression fraction is depth-independent: different quantization
measures give different values of the same functional, and no
value of $k$ appears in it.

We close the section by carrying out the promised test of
Corollary~\ref{cor:sqrtlaw} against deep-leaf angle statistics.
On the coarse net, with $\theta_{\mathrm{cov}} \approx 0.34$, the
predicted window is $\varphi^* \in [0.59, 0.71]$, and the
measured deep-leaf interquartile range is $[0.59, 0.70]$. On the
fine net, with $\theta_{\mathrm{cov}} \approx 0.093$, the
predicted window is $\varphi^* \in [0.517, 0.619]$, and the
measured median is $0.516$ with interquartile range
$[0.438, 0.589]$. Two nets, two windows, nothing fitted, and the
statistics land inside both times.

The picture for Solovay--Kitaev circuits is now complete in one
direction and completely open in the other. Tier-1 rewriting, and
with it the plateau, harvests a constant fraction of the word,
and Sections~\ref{sec:tier1} and~\ref{sec:limit} explain that
fraction down to its dependence on the net. But a constant
fraction of an exponentially wasteful word is still an
exponentially wasteful word. The question that matters for
compilation is whether some \emph{stronger} exact optimizer,
Tiers 2 and 3, or any sound rewriting system whatever, could do
qualitatively better, could take a Solovay--Kitaev circuit all
the way down to the $T$-count that number-theoretic resynthesis
achieves at the same accuracy. The next section proves that it
cannot, that the obstruction is exactness itself, and that the
resulting floor is not an asymptotic abstraction but a number
computable per instance from the exact arithmetic of
Sec.~\ref{sec:prelim}.

\section{The Exactness Barrier}
\label{sec:barrier}

The question left open at the end of the last section has a
one-word answer: soundness. Every rule of the ZX-calculus is
sound, replacing a diagram by another diagram denoting the same
linear map~\cite{coecke-duncan2011,duncan-kissinger2020}. Extraction
returns a circuit denoting that same map. So whatever a ZX
optimizer does to $W_k$, however global its strategy and however
many tiers it employs, its output implements the same unitary
element as its input. Write $\llbracket C \rrbracket$ for the
projective unitary implemented by a circuit $C$. The element
$\llbracket W_k \rrbracket$ is not the target $U$. It is the
approximation the synthesis happened to produce, a specific
Clifford+$T$ element with an exact ring representation, and by
Lemma~\ref{lem:bridge} that element carries an invariant, the
least denominator exponent of its Bloch matrix, which equals the
minimal $T$-count of \emph{any} circuit implementing it. An
invariant of the element is untouchable by transformations that
preserve the element. This is the entire content of the barrier,
and the reason it is worth stating as a theorem is what it
separates: an optimizer chained to the synthesized element and a
resynthesizer free to pick a different element at the same
accuracy live on opposite sides of a gap that grows geometrically
with depth.

Call an optimizer $\mathcal{O}$ \emph{semantics-exact} if its
output is again a Clifford+$T$ circuit and
$\llbracket \mathcal{O}(C) \rrbracket = \llbracket C \rrbracket$
for every input circuit $C$. Both conditions matter: the first
makes the $T$-count of the output meaningful, and holds for
\texttt{PyZX} extraction on Clifford+$T$ inputs, whose diagram
phases remain multiples of $\pi/4$ under every rewrite. All sound
ZX rewriting with exact extraction is therefore semantics-exact,
in particular \texttt{full\_reduce} and \texttt{teleport\_reduce}
at every tier, and so is any pipeline composed of exact rewrites
in any calculus whatever.

\begin{theorem}[Exactness Barrier]\label{thm:B}
\begin{enumerate}[nosep]
\item[(a)] For every semantics-exact optimizer $\mathcal{O}$ and
every single-qubit Clifford+$T$ circuit $W$,
\begin{equation}
T\text{-count of } \mathcal{O}(W)
\;\ge\;
\tmin(\llbracket W \rrbracket)
\;=\;
\mathrm{lde}\big(\hat{W}\big),
\end{equation}
where $\hat W$ abbreviates the exact Bloch representation of
$\llbracket W \rrbracket$, an exact floor computable per instance
in $O(|W|)$ exact ring operations.
\item[(b)] Approximation-aware resynthesis at accuracy
$\epsilon$ achieves $T$-count
$O(\log(1/\epsilon))$~\cite{ross-selinger2016}. Along the SK
sequence, $\log(1/\epsilon_k) = O\big((3/2)^k\big)$, so on any
family of instances whose floor grows as $\Omega(c^k)$ with
$c > 3/2$, the ratio between the exact floor and the resynthesis
$T$-count at matched accuracy grows geometrically in $k$, at rate
at least $c/(3/2)$ per level.
\end{enumerate}
\end{theorem}

\begin{proof}
(a) Let $V = \llbracket W \rrbracket$ and let
$W' = \mathcal{O}(W)$. Semantics-exactness gives
$\llbracket W' \rrbracket = V$ with $W'$ a Clifford+$T$ circuit,
so $W'$ implements $V$ exactly and by the definition of the
minimal $T$-count its $T$-count is at least $\tmin(V)$.
Lemma~\ref{lem:bridge} identifies $\tmin(V)$ with
$\mathrm{lde}(\hat{V})$, and $\hat{V} = \hat{W}$
because the Bloch representation
$\hat{V}_{ij} = \tfrac12\,\mathrm{tr}(\sigma_i V \sigma_j
V^\dagger)$ is insensitive to global phase and depends only on
the implemented projective element.

For the computability claim, multiply out the exact Bloch
matrices of the letters of $W$ left to right in the storage
format of Sec.~\ref{sec:prelim}, applying the reduction step
after every product. This takes $|W|$ exact matrix
multiplications. That the exponent read off at the end is the
\emph{least} denominator exponent, and not merely some
denominator exponent, is guaranteed by the reduction step: a
representation $(P + Q\sqrt2)/\sqrt2^{\,m}$ admits a smaller
exponent precisely when criterion~\eqref{eq:redcrit} holds, a
condition the procedure tests and applies greedily until it
fails, at which point the exponent is minimal. Finally, that sound ZX rewriting with exact extraction
is semantics-exact is the soundness of each rewrite
rule~\cite{coecke-duncan2011,duncan-kissinger2020}, composed over
the finitely many rules applied, followed by an extraction step
that by construction returns a circuit denoting the final
diagram's linear map~\cite{duncan-kissinger2020,backens-extraction}.

(b) The resynthesis bound assembles from three ingredients. The
first is Ross and Selinger's~\cite{ross-selinger2016}: a
$z$-rotation is approximated to operator-norm accuracy $\epsilon$
by a Clifford+$T$ word of $T$-count
$3\log_2(1/\epsilon) + O(\log\log(1/\epsilon))$. The second is
the Euler decomposition: an arbitrary $SU(2)$ element factors as
$R_z(\alpha)\,H R_z(\beta) H\,R_z(\gamma)$ up to a global phase,
using $R_x = H R_z H$. Approximating each of the three rotations
to accuracy $\epsilon/3$ and using that the operator norm is
unitarily invariant and subadditive under composition of errors,
the assembled word approximates the element to accuracy
$\epsilon$ with $T$-count at most
$9\log_2(3/\epsilon) + O(\log\log(1/\epsilon)) =
O(\log(1/\epsilon))$.

The third ingredient is the growth of $\log(1/\epsilon_k)$ along
the SK sequence, and here only an upper bound is needed. The
Solovay--Kitaev convergence estimate~\cite{dawson-nielsen} gives
$\epsilon_k \le C_{\mathrm{sk}}\,\epsilon_{k-1}^{3/2}$ for a
constant $C_{\mathrm{sk}}$ depending on the net, with
$C_{\mathrm{sk}}^2\epsilon_0 < 1$ below the convergence
threshold. Substituting $\delta_k = C_{\mathrm{sk}}^2\epsilon_k$
turns the recursion into
\begin{equation}
\delta_k = C_{\mathrm{sk}}^2\epsilon_k
\le C_{\mathrm{sk}}^3\,\epsilon_{k-1}^{3/2}
= C_{\mathrm{sk}}^3
\big(\delta_{k-1}/C_{\mathrm{sk}}^2\big)^{3/2}
= \delta_{k-1}^{3/2},
\end{equation}
the constants cancelling exactly, so by induction
$\delta_k \le \delta_0^{(3/2)^k}$ and
\begin{equation}
\log\frac{1}{\epsilon_k}
\;\le\; (3/2)^k \log\frac{1}{C_{\mathrm{sk}}^2\epsilon_0}
\;+\; 2\log C_{\mathrm{sk}}
\;=\; O\big((3/2)^k\big) .
\end{equation}
Now let the floor of a family of instances satisfy
$\tmin(\llbracket W_k \rrbracket) \ge c_1 c^k$ with $c > 3/2$. At
matched accuracy $\epsilon_k$, resynthesis achieves $T$-count at
most $C_2 \log(1/\epsilon_k) \le C_3 (3/2)^k$ by the first two
ingredients and the bound just proven, so the ratio of the exact
floor to the resynthesis $T$-count is at least
$(c_1/C_3)\,\big(c/(3/2)\big)^k$, which grows geometrically at
rate $c/(3/2)$ per level.
\end{proof}

Part (b) is stated as an implication because the growth rate of
the floor is a property of the instances, not of the barrier.
What the instances actually do is measured in
Sec.~\ref{sec:numerics}: on real SK words the floor tracks the
word length, growing by factors of $4.0$--$5.6$ per level
(Table~\ref{tab:barrier}), so the operative value is
$c \approx 5$, and the separation ratio accordingly grows by
$3.5$--$3.9$ per level, consistent with the theoretical
$5/(3/2) = 10/3$, from $1.9\times$ at $k = 2$ to $101\times$ at
$k = 5$.
Nothing in this comparison is asymptotic hand-waving: at every
$k$, both sides of the gap are certified per instance, the floor
by the exact arithmetic of part (a) and the resynthesis count by
running \texttt{gridsynth} at the measured accuracy.

The barrier explains the two halves of the literature
asymmetrically, and the asymmetry is the point. On the
Solovay--Kitaev side, the floor sits far below the word's actual
$T$-count, because the recursion multiplies word length by five
per level while the implemented element wanders only as far as
the accuracy budget allows. Exact rewriting may close part of
that distance, and Sections~\ref{sec:tier1} and~\ref{sec:limit}
computed exactly how much: the net-statistics plateau, a constant
fraction, never more. On the number-theoretic side the situation
is reversed. \texttt{gridsynth} emits Matsumoto--Amano~\cite{matsumoto-amano} normal
forms, which are $T$-optimal for the element they
implement~\cite{giles-selinger}, so the output already sits on
its own floor and a semantics-exact optimizer has essentially
zero slack to harvest. The near-null gains of
Ref.~\cite{trasyn2025} are not a failure of the optimizer. They
are the barrier operating at zero distance, and
Sec.~\ref{sec:multiqubit} makes this quantitative at $n \ge 2$,
where the certificates of Corollary~\ref{cor:cert} put the
harvestable slack on the real pipeline at a few percent at most.

One measurement in this story is stranger than the barrier
itself, and it sets up the next section. On $24$ real SK
instances spanning $k = 2$--$4$, with floors up to $568$, the
$T$-count that \texttt{teleport\_reduce} actually reaches equals
the exact floor in every single case: the optimality gap is not
small but zero (Sec.~\ref{sec:numerics}). A lower bound has no
business being attained. The barrier says no exact optimizer can
do better than the floor, and says nothing about whether local
graph rewriting, with no access to the ring arithmetic that
defines the floor, should ever reach it. That it does, every
time, demands a mechanism, and the next section supplies one: a
closed-form combinatorial formula for the single-qubit minimal
$T$-count, phase linkage through the $Z$-axis normalizer, which
automated ZX simplification turns out to compute exactly.

\section{The Linkage Theorem and Single-Qubit $T$-Optimality}
\label{sec:linkage}

The mechanism promised at the end of the last section has to
answer one question: when can two $T$-phases in a word combine,
and when can they not? A phase can be pushed through a Clifford
that preserves its axis. Writing $\mathcal{C}_1$ for the
single-qubit Clifford group, of order $24$ projectively, the
Cliffords that preserve the $Z$ axis form the subgroup
\begin{equation}
N \;=\; N(Z) \;=\;
\{C \in \mathcal{C}_1 : CZC^\dagger = \pm Z\},
\end{equation}
of order $8$ projectively, acting in $SO(3)$ as the signed
permutations fixing the axis $\pm\hat e_z$. For $C \in N$,
conjugation carries a $Z$-rotation to a $Z$-rotation with at most
a sign on its angle, so a phase slides through $C$ and merges
with the phase on the other side. For $C \notin N$, conjugation
tilts the axis and the two phases sit on genuinely different
rotation axes. The question is whether ``cannot merge by
transport'' really means ``cannot merge at all'': if a word
holds $n$ odd phases separated by non-normalizer Cliffords, is
its minimal $T$-count actually $n$, or could some cleverer exact
identity, invisible to axis bookkeeping, still collapse them?
The next lemma says the bookkeeping is the whole truth, and it
says so through the ring: each separated odd phase costs exactly
one unit of the denominator exponent, with no exceptions.

We fix notation for the phases as in Sec.~\ref{sec:tier1}:
$Z_m$ denotes the $Z$-rotation with phase $m \in \Zb_8$ in units
of $\pi/4$, and a phase is \emph{odd} if $m$ is odd, that is, if
the rotation is non-Clifford. Call a word \emph{separated} if it
has the form
\begin{equation}
V \;=\; C_0\, Z_{m_1}\, C_1\, Z_{m_2}\, C_2 \cdots Z_{m_n}\, C_n,
\end{equation}
where every $m_i$ is odd, the outer Cliffords $C_0, C_n$ are
arbitrary, and every interior Clifford $C_1, \ldots, C_{n-1}$
lies outside $N$.

\begin{lemma}[Valuation Lemma]\label{lem:valuation}
Every separated word satisfies
$\mathrm{lde}(\hat V) = n$, and hence $\tmin(V) = n$.
\end{lemma}

\begin{proof}
We compute in the exact $SO(3)$ Bloch representation over
$\Zb[\sqrt2]$, writing each Bloch matrix in the
form~\eqref{eq:ringform}. The lde is read off any representation
by greedy application of criterion~\eqref{eq:redcrit}, so the
proof tracks how one synthesis step changes the pair $(P, Q)$.

\emph{The building blocks.} A Clifford $C$ has Bloch matrix
$\hat C$ a signed permutation matrix: integer entries, exponent
zero. An odd rotation $Z_m$ has Bloch matrix
\begin{equation}
\hat Z_m \;=\; \frac{1}{\sqrt2}
\begin{pmatrix} a & -b & 0\\ b & a & 0\\ 0 & 0 & \sqrt2
\end{pmatrix}
\;=\; \frac{A_m + B\sqrt2}{\sqrt2},
\end{equation}
where $a = \sqrt2\cos(m\pi/4)$ and $b = \sqrt2\sin(m\pi/4)$ lie
in $\{\pm1\}$,
$A_m$ is the integer matrix holding the $2\times2$ block of
$\pm1$'s, and $B = \mathrm{diag}(0,0,1)$. Two facts about these
blocks drive everything. First, $A_m \bmod 2$ is the same matrix
for every odd $m$, namely the one with the upper $2\times2$
block all ones and zeros elsewhere; we call it $\bar A$. Second,
$\hat Z_m$ itself is irreducible, so a single odd rotation has
lde exactly one.

\emph{Upper bound.} The word $V$ multiplies $n$ matrices of
exponent one and $n{+}1$ matrices of exponent zero, so before
any reduction its representation has exponent $n$, and reduction
can only lower it: $\mathrm{lde}(\hat V) \le n$.

\emph{One step, exactly.} Suppose the accumulated product after
$j$ steps is in reduced form $(P + Q\sqrt2)/\sqrt2^{\,e}$, and
append one more step $Z_{m_{j+1}} C_{j+1}$, abbreviating
$A = A_{m_{j+1}}$ and $C = \hat C_{j+1}$ for the duration of the
step. The new numerator is
\begin{multline}
(P + Q\sqrt2)(A + B\sqrt2)\,C
\\
= \big[(PA + 2QB) + (PB + QA)\sqrt2\big]\,C,
\end{multline}
over the exponent $e + 1$. The new representation is reducible
exactly when its integer part vanishes modulo $2$, and since
$2QB \equiv 0$ and $\hat C_{j+1}$ is a signed permutation, hence
invertible modulo $2$, this happens exactly when
$P \bar A \equiv 0$ entrywise modulo $2$. Two consequences
follow. The reducibility of the step depends on the accumulated
matrix only through the residues
$(\bar P, \bar Q) = (P, Q) \bmod 2$, because the update
\begin{equation}
(\bar P, \bar Q) \;\longmapsto\;
\big(\bar P \bar A\, \bar C_{j+1},\;
(\bar P \bar B + \bar Q \bar A)\, \bar C_{j+1}\big)
\end{equation}
is linear over $\mathbb{F}_2$. And the residue update does not
depend on which odd phase $m_{j+1}$ was chosen, since
$\bar A$ does not.

\emph{The finite verification.} The lemma is now a reachability
statement about a finite automaton: states are residue pairs
$(\bar P, \bar Q)$, the initial states are the residues
$(\bar C_0, 0)$ of the $24$ Clifford prefixes, and each choice
of interior Clifford $C \notin N$ applies the update above,
counted once per odd phase even though the residue action is
phase-independent. Because $\bar A$ does not depend on the phase and
$\bar C$ is invertible, the criterion $\bar P \bar A = 0$ is a
property of the state alone, and the claim to verify is that it
holds at no reachable state. We verify this by exhaustive
breadth-first search (Appendix~\ref{app:num}). Seeded from the
residues of the $24$ Clifford prefixes, which collapse to six
distinct states because signs vanish modulo $2$, the reachable
space closes at exactly $24$ states, and the criterion holds at
none of them; counting one transition per triple of state,
interior Clifford, and odd phase gives $24 \times 16 \times 4 =
1536$ transitions, none reducible. Every step in a separated word
therefore raises the exponent by exactly one, from the exponent
one of the first odd rotation to $n$ after the last, and no
reduction ever undoes it: $\mathrm{lde}(\hat V) = n$. With
Lemma~\ref{lem:bridge}, $\tmin(V) = n$.

\emph{Sharpness.} The hypothesis $C \notin N$ is not an artifact
of the method. Rerunning the same search with normalizer Cliffords
admitted as interior steps enlarges the reachable space to $34$
states, of which $10$ satisfy the reducibility criterion, one of
them the degenerate residue at which the entire numerator vanishes
modulo $2$. Linked phases therefore do collapse the exponent,
exactly as transport predicts, and none of the ten is reachable
under the hypothesis of the lemma.
\end{proof}

\begin{figure*}[t]
\centering
\begin{tikzpicture}[scale=0.93,every node/.style={transform shape}]

\node[anchor=west,font=\footnotesize\itshape] at (-0.15,0.95)
  {(a) $D=S\in N$: the separating Clifford is itself a $Z$ rotation, so Tier-1 fusion alone merges the phases};
\node[zxZ] (a1) at (0.75,0) {$\tfrac{\pi}{4}$};
\node[zxZ] (a2) at (1.95,0) {$\tfrac{\pi}{2}$};
\node[zxZ] (a3) at (3.15,0) {$\tfrac{\pi}{4}$};
\draw[zxwire] (0.15,0)--(a1); \draw[zxwire] (a1)--(a2); \draw[zxwire] (a2)--(a3);
\draw[zxwire] (a3)--(3.75,0);
\node[font=\scriptsize] at (0.75,-0.62) {$T$};
\node[font=\scriptsize] at (1.95,-0.62) {$D$};
\node[font=\scriptsize] at (3.15,-0.62) {$T$};
\node[font=\small] at (4.3,0) {$=$};
\node[zxZ] (a4) at (5.35,0) {$\pi$};
\draw[zxwire] (4.75,0)--(a4); \draw[zxwire] (a4)--(5.95,0);
\node[font=\footnotesize,anchor=west] at (6.4,0) {$\tmin=0$};

\node[anchor=west,font=\footnotesize\itshape] at (-0.15,-1.55)
  {(b) $D=X=HSSH\in N$: the Clifford is not a $Z$ rotation, but it preserves the axis, and the phases still annihilate};
\node[zxZ] (b1) at (0.75,-2.5) {$\tfrac{\pi}{4}$};
\node[zxX] (b2) at (1.95,-2.5) {$\pi$};
\node[zxZ] (b3) at (3.15,-2.5) {$\tfrac{\pi}{4}$};
\draw[zxwire] (0.15,-2.5)--(b1); \draw[zxwire] (b1)--(b2); \draw[zxwire] (b2)--(b3);
\draw[zxwire] (b3)--(3.75,-2.5);
\node[font=\scriptsize] at (0.75,-3.12) {$T$};
\node[font=\scriptsize] at (1.95,-3.12) {$D$};
\node[font=\scriptsize] at (3.15,-3.12) {$T$};
\node[font=\small] at (4.3,-2.5) {$=$};
\node[zxX] (b4) at (5.35,-2.5) {$\pi$};
\node[zxZ] (b5) at (6.6,-2.5) {$-\tfrac{\pi}{4}$};
\node[zxZ] (b6) at (7.9,-2.5) {$\tfrac{\pi}{4}$};
\draw[zxwire] (4.75,-2.5)--(b4); \draw[zxwire] (b4)--(b5); \draw[zxwire] (b5)--(b6);
\draw[zxwire] (b6)--(8.5,-2.5);
\node[font=\small] at (9.05,-2.5) {$=$};
\node[zxX] (b7) at (10.1,-2.5) {$\pi$};
\draw[zxwire] (9.5,-2.5)--(b7); \draw[zxwire] (b7)--(10.7,-2.5);
\node[font=\footnotesize,anchor=west] at (11.15,-2.5) {$\tmin=0$};

\node[anchor=west,font=\footnotesize\itshape] at (-0.15,-4.05)
  {(c) $D=H\notin N$: the axis is tilted, no rule of any tier applies, and each phase costs one unit of the valuation};
\node[zxZ] (c1) at (0.75,-5.0) {$\tfrac{\pi}{4}$};
\node[zxH] (c2) at (1.95,-5.0) {};
\node[zxZ] (c3) at (3.15,-5.0) {$\tfrac{\pi}{4}$};
\draw[zxwire] (0.15,-5.0)--(c1); \draw[zxwire] (c1)--(c2); \draw[zxwire] (c2)--(c3);
\draw[zxwire] (c3)--(3.75,-5.0);
\node[font=\scriptsize] at (0.75,-5.62) {$T$};
\node[font=\scriptsize] at (1.95,-5.62) {$D$};
\node[font=\scriptsize] at (3.15,-5.62) {$T$};
\node[font=\small] at (4.3,-5.0) {$\neq$};
\node[font=\footnotesize,anchor=west] at (4.85,-5.0)
  {anything shorter: $\hat V=(P{+}Q\sqrt2)/\sqrt2^{\,2}$ with $P\not\equiv0$,};
\node[font=\footnotesize,anchor=west] at (4.85,-5.45)
  {so $\mathrm{lde}=2$ and $\tmin=2$};

\end{tikzpicture}
\caption{The normalizer condition of Lemma~\ref{lem:valuation},
illustrated on the shortest nontrivial words. In each row two $T$
phases are separated by a single Clifford $D$. (a) When $D$ is a
$Z$ rotation the three spiders are all green and Tier-1 fusion
merges them outright, leaving a Clifford. (b) When $D=X$, which is
$HSSH$ and therefore not visibly a $Z$ rotation at all, the phases
still merge: $X$ preserves the $Z$ axis up to sign, so pushing the
red $\pi$ spider through the second phase negates it, the two green
spiders fuse to zero, and both $T$ gates disappear. Membership in
$N$, not syntactic appearance, is what decides the outcome. (c)
When $D=H$ the axis is tilted, the two phases sit on genuinely
different rotation axes, and no rewrite of any tier can bring them
together; the exact Bloch representation confirms the obstruction
arithmetically, its numerator being nonzero modulo $2$, so the
denominator exponent is $2$ and two $T$ gates are necessary. The
content of Lemma~\ref{lem:valuation} is that this last obstruction
persists for words of every length.}
\label{fig:normalizer}
\end{figure*}
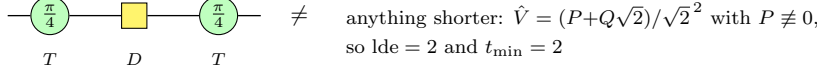

The lemma is the engine of this section, and both of its jaws
are needed; Fig.~\ref{fig:normalizer} shows them at work on the
shortest words in which they can be seen. The forward jaw is rigidity: separated odd phases
are individually visible to the ring valuation, one unit each,
so no exact process of any kind can merge them, and in
particular no rewriting system can over-fuse. The reverse jaw is
the sharpness clause: phases whose infix lies in $N$ genuinely
do collapse, so transport-mergeable phases must be merged by any
optimizer hoping to reach the floor. Squeezed between the two,
the minimal $T$-count of an arbitrary word can only be the
number of odd phases that survive when all normalizer transport
has been carried out, and that is the statement the Linkage
Theorem now makes precise.

To state the count precisely we first fix the procedure. Bring
an arbitrary single-qubit Clifford+$T$ word into
phase--Clifford alternating form by Tier-1 fusion and by
classifying every even phase as the Clifford it is, so that
$w = C_0 Z_{m_1} C_1 \cdots Z_{m_n} C_n$ with all $m_i$ odd.
Two moves are then available. \emph{Transport:} if some interior
Clifford $D = C_i$ lies in $N$, use the conjugation identity to
slide the phase through it. Writing $\varepsilon(D) = +1$ if
$D^\dagger Z D = Z$ projectively and $-1$ if
$D^\dagger Z D = -Z$, one has
$Z_m D = D\, Z_{\varepsilon(D)\,m}$ for every $m \in \Zb_8$, so
\begin{equation}
Z_{m_i}\, D\, Z_{m_{i+1}}
\;=\;
D\, Z_{\varepsilon(D)\,m_i}\, Z_{m_{i+1}}
\;=\;
D\, Z_{\varepsilon(D)\,m_i + m_{i+1}} ,
\end{equation}
merging the two phases into one with a signed sum.
\emph{Absorption:} since two odd phases have an even signed sum,
every merge produces an even phase, which is a Clifford; absorb
it into the skeleton. Absorption can create new links, because
an even phase lies in $N$ itself, so the combined infix between
the two phases that flanked the merged pair may now land in $N$
and enable a further transport. Run the two moves in any order
until neither applies, and call the result \emph{saturated}. The
odd phases surviving in a saturated word are what the theorem
counts.

\begin{theorem}[Linkage Theorem]\label{thm:linkage}
For every single-qubit Clifford+$T$ word $w$,
\begin{equation}
\label{eq:linkage}
\tmin(\llbracket w \rrbracket)
\;=\;
\#\{\text{odd phases surviving at saturation}\},
\end{equation}
and the count is independent of the order in which the moves are
applied.
\end{theorem}

\begin{proof}
\emph{Step 1: every move is exact.} The transport identity is
the definition of the normalizer read at the level of rotations.
For $D \in N$ we have $D^\dagger Z D = \varepsilon(D) Z$
projectively, hence
$D^\dagger e^{-i m \pi Z/8} D = e^{-i m \pi \varepsilon(D) Z/8}$
up to global phase, which is exactly
$Z_m D = D Z_{\varepsilon(D) m}$, and the subsequent fusion
$Z_{\varepsilon(D)m} Z_{m'} = Z_{\varepsilon(D)m + m'}$ is a
group identity in the phase subgroup. Absorption reclassifies an even
phase as a Clifford and multiplies it into its neighbour, again
an identity. Every intermediate word therefore implements
$\llbracket w \rrbracket$ exactly.

\emph{Step 2: saturation is reached.} Each merge reduces the
number of explicit odd phases by two and increases it by at most
zero, since the merged phase is even and is absorbed; transport
alone does not change the phase count; and absorption only
shortens the skeleton description. Starting from $n$ odd phases,
at most $\lfloor n/2 \rfloor$ merges can occur, so any maximal
sequence of moves is finite and ends in a saturated word.

\emph{Step 3: a saturated word is separated.} Let the saturated
word carry $n'$ surviving odd phases. Between any two
consecutive survivors the infix is a product of original
Cliffords and absorbed even phases, hence a single Clifford, and
it cannot lie in $N$, for otherwise the transport move would
still apply and the word would not be saturated. The saturated
word is therefore exactly a separated word in the sense of
Lemma~\ref{lem:valuation}, including the degenerate case
$n' = 0$, where the word is a single Clifford of lde zero. The
lemma gives
$\mathrm{lde} = n'$ and hence, with Lemma~\ref{lem:bridge},
$\tmin = n'$ for the element the saturated word implements. By
Step 1 that element is $\llbracket w \rrbracket$, which
proves~\eqref{eq:linkage} for this particular saturation.

\emph{Step 4: order-independence.} Different move orders may
well pass through different intermediate words and terminate in
different saturated words; the cascades opened by absorption
make this unavoidable. But every saturated word reachable from
$w$ implements the same element $\llbracket w \rrbracket$, by
Step 1, and Step 3 evaluates the survivor count of each one as
$\mathrm{lde}$ of that same element. Two saturated forms of $w$
therefore always carry the same number of surviving odd phases,
namely $\tmin(\llbracket w \rrbracket)$, whatever routes reached
them.
\end{proof}

The theorem is simultaneously a formula and an algorithm: to
know the minimal $T$-count of a word, run transports and
absorptions to saturation and count what is left, with no ring
arithmetic in sight at run time, the ring having done its work
once and for all inside Lemma~\ref{lem:valuation}. But the
puzzle that opened this section was not about an algorithm we
might design. It was about the algorithm that already exists:
\texttt{PyZX}'s automated strategies know nothing of rotation
axes, normalizers, or denominator exponents, and yet they land
on the floor every time. The final theorem of this section
resolves the puzzle by showing that the graph-rewriting
strategies \emph{are} the linkage procedure, executed in a
different notation.

Comparing a graph algorithm with a word algorithm requires
putting them on common ground, and the device that does it is
parametrization. Replace the phase of every non-Clifford spider
in the diagram of $w$ by a formal variable $\alpha_j$. The
diagram is then a Clifford diagram decorated with variables, the
simplification strategies act on the Clifford skeleton while the
variables ride along as gadgetized phases, and the phase table
of \texttt{teleport\_reduce} records which variables have been
fused into a common slot and with which relative
signs~\cite{duncan-kissinger2020,kissinger-tcount}. The
advantage of formal variables is that they turn observations
into identities: a fusion recorded in the phase table is valid
for \emph{every} instantiation of the variables, so it cannot be
an accident of particular phase values and must express an exact
identity of the Clifford skeleton itself. The theorem is that
the identities so expressed are precisely the transports of the
linkage procedure, no more and no fewer.

\begin{theorem}\label{thm:L1}
On every single-qubit Clifford+$T$ word,
\texttt{teleport\_reduce} and \texttt{full\_reduce} terminate
with $T$-count exactly $\tmin$.
\end{theorem}

\begin{proof}
We prove the statement for the abstract strategies as specified
by their rule sets and phase-table
semantics~\cite{duncan-kissinger2020,kissinger-tcount};
conformance of the shipped implementation, including its
boundary-handling variants, is checked exhaustively rather than
by source audit, as reported at the end of this section and in
Sec.~\ref{sec:discussion}. The proof runs through three claims.

\emph{Claim 1: the Clifford skeleton contracts completely.}
Every interior spider carrying a Clifford phase admits a rule of
the strategy: a zero-phase spider of degree two is removed by
identity removal and one of higher degree is fused; a
$\pm\pi/2$ spider is eliminated by local complementation; a
$\pi$ spider is eliminated by pivoting against another Pauli
spider or, when none is adjacent, converted by
\texttt{pivot\_gadget} into gadget form. Each of these steps
strictly decreases the number of interior spiders, and the
termination of the full strategy on graph-like diagrams, with
the characterization that no interior proper Clifford spider
survives, is the simplification theorem of
Ref.~\cite{duncan-kissinger2020}, whose extraction step is
analysed in generality in Ref.~\cite{backens-extraction}. At termination, therefore,
nothing of the Clifford skeleton stands between the gadgets of
two consecutive variables except direct wiring.

\emph{Claim 2: shared slots read off the normalizer.} Every
rule of the strategy acts on a variable only through its sign
and through Clifford phase additions, so the cumulative effect
of the contraction on $\alpha_j$ is a single sign, its slot
datum. Suppose two variables $\alpha_i, \alpha_j$, separated in
$w$ by the Clifford infix $D$, end in a shared slot with
relative sign $\varepsilon$. By soundness the terminal diagram
denotes $\llbracket w \rrbracket$ for every instantiation, and
extracting it yields, for every instantiation, a circuit in
which a single rotation carries the phase
$\varepsilon\alpha_i + \alpha_j$. Equating that circuit with $w$
on the segment between the two variables gives the identity
$Z_{\alpha_i}\, D\, Z_{\alpha_j} \doteq D'\,
Z_{\varepsilon\alpha_i + \alpha_j}\, D''$, valid for all
$\alpha_i, \alpha_j$, for some Cliffords $D', D''$, where
$\doteq$ denotes equality up to global phase. An identity
transporting an \emph{arbitrary} $Z$-phase through $D$ forces
$D^\dagger Z D = \pm Z$ projectively, that is, $D \in N$ with
$\varepsilon = \varepsilon(D)$. Conversely, if $D \in N$, then
by Claim 1 the skeleton of $D$ contracts to nothing but wiring
and the coset sign between the two gadgets, spider fusion fires,
and the phase table records the merge. Shared slots therefore
occur exactly for $N$-infixes: the phase table computes the
linkage procedure of Theorem~\ref{thm:linkage}.

\emph{Claim 3: the count is pinned from both sides.}
Instantiate the variables back to the odd phases of $w$. The
terminal $T$-count is the number of slots holding an odd total,
which by Claim 2 is the number of odd phases surviving
saturation of the linkage procedure, and by
Theorem~\ref{thm:linkage} that number is
$\tmin(\llbracket w \rrbracket)$, independently of the order in
which the strategy happened to fire its rules. For emphasis, the
two failure modes are separately impossible. Under-fusion, a
missed merge across an $N$-infix, is excluded by Claims 1 and 2:
the contracted skeleton leaves nothing to prevent the fusion.
Over-fusion, a merge across a non-$N$ infix, would upon
instantiation produce an exact circuit for
$\llbracket w \rrbracket$ with $T$-count strictly below the
count of a separated form, contradicting the Valuation Lemma
through Theorem~\ref{thm:B}. The terminal count is therefore
exactly $\tmin$.
\end{proof}

The theorem closes the loop opened at the end of
Sec.~\ref{sec:barrier}: the floor is attained every time because
the optimizer, without knowing it, computes the formula that
defines the floor. We verify the mechanism as well as the
statement. On the planted family
$\{T\,c\,T,\; T\,c\,T^\dagger : c \in \mathcal{C}_1\}$, $48$
cases in all, \texttt{PyZX} reaches $\tmin$ exactly in every
case, and fusion occurs if and only if $c \in N$, including
inconspicuous normalizer elements such as $HSSH$. Together with
the SK instances of Sec.~\ref{sec:barrier} ($\tmin$ up to
$568$), adversarial $\pi$-blocker suites, and random words, the
cumulative record is $279$ exact matches in $279$ trials
(Appendix~\ref{app:num}).

The single-qubit story is now closed. A synthesized word carries
an exact floor, the floor has a closed combinatorial form, and
the automated optimizer computes that form exactly, so
single-qubit Clifford+$T$ optimization by exact rewriting is,
in every sense that matters, a solved problem. The question that
decides whether any of this bears on practice is what survives
at $n \ge 2$, where the real pipelines live. Two ingredients
made the single-qubit argument work: a valuation that every
$T$-step strictly increments, and a syntactic characterization,
membership in $N$, of exactly when the increment can be undone.
The next section shows that the first ingredient generalizes
cleanly, giving unconditional per-instance rigidity
certificates at every $n$, and that the second does not: already
at $n = 2$ the identity $\tmin = \mathrm{lde}$ fails, and the
certificates, not a formula, are what the multi-qubit theory
runs on.

\section{Multi-Qubit Rigidity}
\label{sec:multiqubit}

The valuation is the ingredient that travels, and the first task
is to choose the representation it travels in. At one qubit the
Bloch matrix is the adjoint action on the three Paulis, and its
natural $n$-qubit generalization is the exact Pauli-channel
representation: for an $n$-qubit Clifford+$T$ operator $V$ and
Pauli-group elements $\sigma, \tau$, set
\begin{equation}
\hat V_{\sigma\tau} \;=\;
2^{-n}\,\mathrm{tr}\big(\sigma\, V \tau V^\dagger\big),
\end{equation}
a $4^n \times 4^n$ matrix with entries in $\Zb[\sqrt2]$, stored
and reduced in the format of Sec.~\ref{sec:prelim}. We implement
this representation exactly (Appendix~\ref{app:num}) and, as a
unit check, validate it against the single-qubit Bloch
representation on $120$ random words, with zero mismatches. The
gate set it must handle is the one the real pipeline emits:
Cliffords, whose channels are signed permutation matrices of
exponent zero, and $\pi/8$ Pauli rotations
$R_S = e^{-i j \pi S/8}$ with $j$ odd and $S$ a Pauli, the class
containing $T$ itself as the case $S = Z$ on one wire. What the
single-qubit argument needed from its representation was two
things: every $T$-step raises the valuation by at most one, and
there is an exactly checkable criterion for when it fails to
rise. The next lemma delivers both at every $n$.

\begin{lemma}\label{lem:mqval}
Let $S$ be a Pauli and $R_S = e^{-ij\pi S/8}$ with $j$ odd, and
split the Pauli basis into the sector commuting with $S$ and the
sector anticommuting with $S$. Then:
\begin{enumerate}[nosep]
\item[(a)] The channel $\hat R_S$ acts as the identity on the
commuting sector and as
$\tfrac{1}{\sqrt2}(\mathrm{id} \pm T_S)$ on the anticommuting
sector, where $T_S: \tau \mapsto S\tau$ is, up to signs, an
involutive permutation of that sector; in particular
$\mathrm{lde}(\hat R_S) = 1$.
\item[(b)] Let $\hat V = (P + Q\sqrt2)/\sqrt2^{\,m}$ be reduced,
and say that $P$ is \emph{pair-constant} if every row of
$P \bmod 2$ is constant on each pair $\{\tau, S\tau\}$ of the
anticommuting sector. Then
$\mathrm{lde}(\hat V \hat R_S) = m$ if $P$ is pair-constant, and
$\mathrm{lde}(\hat V \hat R_S) = m + 1$ otherwise.
\end{enumerate}
\end{lemma}

\begin{proof}
(a) The conjugation action of $R_S$ on a Pauli $\tau$ splits by
commutation. If $\tau$ commutes with $S$, then $\tau$ commutes
with $e^{-ij\pi S/8}$ and $R_S \tau R_S^\dagger = \tau$. If
$\tau$ anticommutes, then $\tau e^{+ij\pi S/8} =
e^{-ij\pi S/8}\tau$ gives
\begin{equation}
R_S\, \tau\, R_S^\dagger
\;=\; \tau\, e^{+ij\pi S/4}
\;=\; \cos\!\big(\tfrac{j\pi}{4}\big)\,\tau
\;+\; i\sin\!\big(\tfrac{j\pi}{4}\big)\,\tau S ,
\end{equation}
and for odd $j$ both trigonometric factors are $\pm1/\sqrt2$,
while $i\tau S$ is again a Hermitian Pauli up to sign because
$\tau$ and $S$ anticommute. The channel therefore carries $\tau$
to $(\pm\tau \pm S\tau)/\sqrt2$, which is the stated
$\tfrac{1}{\sqrt2}(\mathrm{id} \pm T_S)$ once signs are folded
into the map, and $T_S$ squares to the identity up to sign since
$S(S\tau) = \tau$ times a phase. Writing the channel in the
storage format, $\hat R_S = (A + B\sqrt2)/\sqrt2$ with $A$
supported on the anticommuting sector, where
$A \equiv I + \bar T_S$ modulo $2$ for $\bar T_S$ the
permutation matrix of the involution, and $B$ supported on the
commuting sector, where $B$ is the identity. Since
$A \not\equiv 0$ modulo $2$, this representation is reduced and
$\mathrm{lde}(\hat R_S) = 1$.

(b) Appending the step multiplies representations exactly as in
the proof of Lemma~\ref{lem:valuation}:
\begin{equation}
\hat V \hat R_S
\;=\;
\frac{(PA + 2QB) + (PB + QA)\sqrt2}{\sqrt2^{\,m+1}} ,
\end{equation}
and by criterion~\eqref{eq:redcrit} the new representation admits
a reduction precisely when its integer part vanishes modulo $2$,
that is, since $2QB \equiv 0$, precisely when $PA \equiv 0$
entrywise modulo $2$. Computing the product entrywise on the anticommuting
sector, where $\bar A = I + \bar T_S$,
\begin{equation}
(P\bar A)_{r\tau} \;\equiv\; P_{r\tau} + P_{r,\,S\tau}
\pmod 2 ,
\end{equation}
so $PA \equiv 0$ holds exactly when every row of $P \bmod 2$
takes equal values on each pair $\{\tau, S\tau\}$, the stated
criterion; on the commuting sector $\bar A$ vanishes and imposes
nothing. Two of the three implications in (b) are now rigorous.
If the criterion fails, the representation over exponent $m+1$
is already reduced, so $\mathrm{lde} = m+1$ exactly. If the
criterion holds, one reduction fires and
$\mathrm{lde} \le m$. The remaining claim, that the reduction
fires exactly once, so that the criterion never overshoots to
$\mathrm{lde} < m$, is the one component we establish by
exhaustive verification rather than structurally: on $400$
random step trials at $n = 1$ and $150$ at $n = 2$, the computed
lde after a criterion-positive step equals $m$ in every case,
with zero exceptions (Appendix~\ref{app:num}). We emphasize that
nothing downstream leans on this refinement: the certificates of
Corollary~\ref{cor:cert} use only the rigorous half, that no
step raises the exponent by more than one, and the rigidity
fractions of Sec.~\ref{sec:numerics} are computed from the lde
itself, not from the criterion.
\end{proof}

Two remarks place the lemma. First, it contains the single-qubit
theory as a special case: at $n = 1$ with $S = Z$ the
anticommuting sector is $\{X, Y\}$, the pairs are $\{X, Y\}$
itself, and running the $400$-trial check through the criterion
recovers exactly the $24$-state closure of
Lemma~\ref{lem:valuation}, with the syntactic condition
$C \notin N$ reappearing as the description of which residues
are reachable. Second, and more importantly, the lemma converts
the multi-qubit question from one about words to one about a
single exactly computable integer: accumulate the channel along
any realizing word, and the lemma accounts, step by step, for
when the exponent rises and when it collapses. The accounting is
what the certificate extracts.

At $n = 1$ the valuation did more than bound the $T$-count: it
computed it. That is the part which does not survive. The
single-qubit identity $\tmin = \mathrm{lde}$ already fails at two
qubits, and the smallest counterexample is a gate one would
never suspect. The controlled-$S$ gate has channel lde equal to
$2$, by exact computation, while its minimal $T$-count is $3$.
So no naive syntactic analogue of the condition $C \notin N(Z)$
can hold at $n \ge 2$, and there is no route from
Lemma~\ref{lem:mqval} to a closed formula of the kind
Theorem~\ref{thm:linkage} provides. What does survive is the
inequality, and it survives with nothing attached: no hypothesis
about the reachable residues, no characterization of when
collapse occurs, no restriction on the gate set beyond
Clifford+$T$. This is enough to certify rigidity instance by
instance, which is what the rest of the section does.

\begin{corollary}\label{cor:cert}
For every $n$-qubit Clifford+$T$ element $V$,
\begin{equation}
\tmin(V) \;\ge\; \mathrm{lde}(\hat V),
\end{equation}
and the right-hand side is computable exactly from any word
realizing $V$, in one channel multiplication per gate.
\end{corollary}

\begin{proof}
Let $W = g_1 g_2 \cdots g_t$ be any Clifford+$T$ word with
$\llbracket W \rrbracket = V$, and let
$\hat V_j = \hat g_1 \cdots \hat g_j$ denote the accumulated
channel after $j$ gates, so that $\hat V_0$ is the identity with
lde zero and $\hat V_t = \hat V$. We show that
$\mathrm{lde}(\hat V_j) - \mathrm{lde}(\hat V_{j-1})$ is zero
when $g_j$ is a Clifford and at most one when $g_j$ is a
$\pi/8$ rotation.

For a Clifford $g$, the channel $\hat g$ is a signed permutation
matrix, which is an invertible matrix over $\Zb$ whose inverse
is again a signed permutation. Writing
$\hat V_{j-1} = (P + Q\sqrt2)/\sqrt2^{\,e}$ in reduced form, the
product is $(P\hat g + Q\hat g\sqrt2)/\sqrt2^{\,e}$, and
$P\hat g \equiv 0$ modulo $2$ would give
$P \equiv 0$ modulo $2$ upon multiplying by $\hat g^{-1}$,
contradicting reducedness. The product is therefore already
reduced at the same exponent, and the lde is unchanged.

For a $\pi/8$ rotation $g = R_S$, Lemma~\ref{lem:mqval}(a) gives
$\hat R_S = (A + B\sqrt2)/\sqrt2$, so the product
$\hat V_{j-1}\hat R_S$ has a representation over the exponent
$e + 1$, whence $\mathrm{lde}(\hat V_j) \le e + 1$: appending the
gate raises the lde by at most one, and by
Lemma~\ref{lem:mqval}(b) it raises it by exactly one unless the
pair-constancy condition fires. Note that only the upper bound
is used here, which is the half of Lemma~\ref{lem:mqval}(b)
established structurally.

Summing over the word, $\mathrm{lde}(\hat V) \le t_{\pi/8}(W)$,
the number of $\pi/8$ rotations in $W$, and in particular
$\mathrm{lde}(\hat V) \le t(W)$ for the $T$-count of any
Clifford+$T$ word realizing $V$, since a $T$ gate is the
$\pi/8$ rotation with $S$ the single-qubit $Z$ on one wire.
Minimizing over all such words gives
$\mathrm{lde}(\hat V) \le \tmin(V)$. The quantity
$\mathrm{lde}(\hat V)$ depends only on $V$, not on the word used
to compute it, so any realizing word may be used, at the cost of
one channel multiplication and one reduction test per gate.
\end{proof}

The corollary is what makes the multi-qubit claims of this paper
falsifiable rather than merely plausible. Suppose one measures,
as we do in Sec.~\ref{sec:numerics}, that \texttt{PyZX} removes
$2\%$ of the $T$-count of a synthesized circuit. On its own that
number is uninterpretable: it might mean the optimizer is weak,
or that the circuit is rigid, and no amount of additional
optimization effort settles which. The certificate settles it.
Computing $\mathrm{lde}(\hat V)$ for the synthesized element
bounds, from below, the $T$-count of everything any
semantics-exact optimizer could ever return, so the gap between
the certificate and the achieved count is the entire remaining
headroom, for \texttt{PyZX} and for every future optimizer
alike. When that gap is under one percent, as we find it to be,
the smallness of the measured reduction is explained rather than
merely reported.

Two features of the bound deserve emphasis before it is put to
work. It is unconditional, resting on no hypothesis about
reachable residues and on no syntactic criterion, which is what
lets it survive the failure of $\tmin = \mathrm{lde}$ at
$n \ge 2$: a lower bound does not care that it is sometimes not
tight. And it is cheap, one channel multiplication per gate,
which is what lets it run on circuits with $T$-counts in the
thousands, where the exact minimal $T$-count is far out of
computational reach. Rigidity at scale is certified, not
estimated. What remains is to apply this to the pipeline that
practice actually uses.

The pipeline in question is the one a practitioner actually
runs. A target unitary is handed to a transpiler, which
decomposes it into two-qubit $KAK$
blocks~\cite{vatan-williams} and, at $n \ge 3$, the quantum
Shannon decomposition~\cite{shende-bullock-markov},
producing a circuit over $\{u, CX\}$. Every single-qubit
rotation $u$ is then approximated by \texttt{gridsynth} to the
per-rotation accuracy budget, and the results are concatenated.
The output is a long Clifford+$T$ circuit assembled from many
short, individually optimal pieces, and the question of this
section is what a semantics-exact optimizer can do with it. The
answer is structural, and it turns on a fact about
Matsumoto--Amano normal forms that Sec.~\ref{sec:linkage} has
already put us in a position to prove.

\begin{theorem}\label{thm:C}
Let $C = G_0 B_1 G_1 \cdots B_L G_L$ be an assembled
QSD+\texttt{gridsynth} circuit, where each $B_i$ is a
\texttt{gridsynth} block in Matsumoto--Amano normal form acting
on one wire (which we verify on every angle used in our
experiments; Appendix~\ref{app:num}), each $G_i$ is Clifford
glue, and $t(C)$ is the total $T$-count. Then:
\begin{enumerate}[nosep]
\item[(a)] Each block is individually rigid: the channel of
$B_i$ has $\mathrm{lde}(\hat B_i) = t(B_i)$, so no
semantics-exact optimizer can remove a single $T$ gate from a
block in isolation.
\item[(b)] Consequently every collapse along the assembled
circuit, that is, every step at which the criterion of
Lemma~\ref{lem:mqval}(b) fires, is a cross-block effect,
attributable to the residue accumulated by the prefix rather
than to the block being read.
\item[(c)] For every semantics-exact optimizer $\mathcal{O}$,
the fraction of $T$-count it can remove satisfies
\begin{equation}
\frac{t(C) - t(\mathcal{O}(C))}{t(C)}
\;\le\;
1 - \frac{\mathrm{lde}(\hat C)}{t(C)},
\end{equation}
with both sides computable exactly per instance.
\end{enumerate}
\end{theorem}

\begin{proof}
(a) A Matsumoto--Amano normal form~\cite{matsumoto-amano} is a word
$(T|\varepsilon)(HT|SHT)^{\ell}\, c$ with $c \in \mathcal{C}_1$ a
final Clifford, so between any two consecutive $T$ gates the
intervening Clifford is either $H$ or $SH$.
Neither lies in $N(Z)$: conjugation by $H$ sends $Z$ to $X$, and
conjugation by $SH$ sends $Z$ to $Y$, in neither case to
$\pm Z$. The block is therefore a separated word in the sense of
Lemma~\ref{lem:valuation}, whose interior Cliffords are required
to lie outside $N$, and the lemma gives
$\mathrm{lde} = t(B_i)$ for its single-qubit Bloch
representation.

It remains to see that embedding the block on one wire of an
$n$-qubit register does not change this. If $V$ acts as $U$ on a
single wire and trivially elsewhere, split each Pauli as
$\sigma = \sigma' \otimes \sigma''$ and
$\tau = \tau' \otimes \tau''$, with the primed factor on the
active wire. Then
\begin{equation}
\hat V_{\sigma\tau}
= 2^{-n}\,
\mathrm{tr}\big(\sigma' U \tau' U^\dagger\big)\,
\mathrm{tr}\big(\sigma'' \tau''\big)
= \hat U_{\sigma'\tau'}\,\delta_{\sigma''\tau''},
\end{equation}
using $\mathrm{tr}(\sigma''\tau'') = 2^{n-1}
\delta_{\sigma''\tau''}$ on the remaining wires. So $\hat V = \hat U \otimes I$
with $I$ integral of exponent zero, and the two channels have
the same numerator parity and hence the same lde. Rigidity of
the block in isolation then follows from
Corollary~\ref{cor:cert}: any exact optimizer's output has
$T$-count at least $\mathrm{lde}(\hat B_i) = t(B_i)$.

(b) Accumulate the channel along $C$ from the identity. By
Corollary~\ref{cor:cert} each $T$ step raises the lde by at most
one, and by Lemma~\ref{lem:mqval}(b) it raises it by exactly one
unless the pair-constancy criterion fires. Suppose a collapse
occurs at a step inside block $B_i$. Were the accumulated
residue at that step the one $B_i$ would have generated on its
own, that is, were the prefix trivial, part (a) would forbid the
collapse, since it would force $\mathrm{lde}(\hat B_i)$ below
$t(B_i)$. The collapse is therefore caused by the difference
between the true accumulated residue and the block's own, which
is precisely the contribution of $G_0 B_1 \cdots G_{i-1}$. No
other location is available: steps inside the glue are Clifford
steps, which by the proof of Corollary~\ref{cor:cert} leave the
lde unchanged and admit no collapse at all.

(c) By Corollary~\ref{cor:cert} applied to
$V = \llbracket C \rrbracket$, any semantics-exact
$\mathcal{O}$ returns a circuit with
$t(\mathcal{O}(C)) \ge \mathrm{lde}(\hat C)$. Subtracting from
$t(C)$ and dividing gives the stated bound. The left-hand side
is measured by running the optimizer; the right-hand side is
computed by accumulating the exact channel along $C$, at one
channel multiplication per gate.
\end{proof}

Part (c) turns the theorem into an experiment with a
pre-registered ceiling, and the experiment is reported in
Sec.~\ref{sec:numerics}. Three of its findings are worth
anticipating here, because the theorem is what makes them
interpretable. First, \texttt{teleport\_reduce} removes
$1.8$--$3.8\%$ of the $T$-count across all five pipeline
configurations, reproducing the near-null result of
Ref.~\cite{trasyn2025} and extending it to $n = 3$. Second,
almost none of what the optimizer left behind can ever be
removed: writing $T_{\mathrm{tel}}$ for the $T$-count after
\texttt{teleport\_reduce}, the certified incompressible fraction
$\mathrm{lde}/T_{\mathrm{tel}}$ averages $99.4$--$99.9\%$ and
never falls below $96.9\%$ on any single instance, so the
residual headroom left to \emph{any} semantics-exact optimizer
is $0.1$--$0.6\%$ on average, the measured reduction sits
essentially on the ceiling of part (c) everywhere, and the
smallness is explained rather than merely observed.
Third, the collapse rate that part (b) identifies as the only
source of compression measures $1.0$--$2.3\%$ on the real
grammar, against $22$--$54\%$ on toy commuting-parity grammars,
confirming that the long Matsumoto--Amano blocks scramble
channel residues far more effectively than bare parity ladders
do. Rigidity moreover strengthens as the accuracy target
tightens, which is the least convenient direction possible: the
regime in which optimization matters most is the regime in which
there is provably least to optimize.

This completes the theory. The single-qubit account is closed by
an exact formula, the multi-qubit account by exact per-instance
certificates, and both halves say the same thing about where
optimization gains come from. What remains is to describe how
the computations were carried out and to report what they found,
including the measurements just quoted, the depth-independence
of the plateau, and the exhaustive verifications on which
several of the proofs above depend.

\section{Numerical Methods and Results}
\label{sec:numerics}

\subsection{Methods}

Every number reported below comes from an implementation written
for this paper and validated component by component before use;
the validation record is Appendix~\ref{app:num}. We describe the
four pieces in turn: the nets, the synthesis recursion, the exact
arithmetic, and the multi-qubit pipeline.

\emph{Nets.} All Solovay--Kitaev results are reported on two
base nets, built by deliberately different constructions so that
net-dependent claims can be separated from artifacts of one
construction. The \emph{coarse} net is a raw enumeration of all
words up to length $l_0 = 6$, with covering-radius angle
$\theta_{\mathrm{cov}} \approx 0.34$. The \emph{fine} net is
built by breadth-first search retaining one shortest
representative per group element, reaching $l_0 = 8$ with
$N = 2{,}015{,}538$ words and
$\theta_{\mathrm{cov}} \approx 0.093$. The two differ by a
factor of nearly four in resolution and by their redundancy
structure, which is what makes the agreement of the
parameter-free prediction of Corollary~\ref{cor:sqrtlaw} on both
of them informative.

\emph{Synthesis.} We implement the Dawson--Nielsen
recursion~\cite{dawson-nielsen} as specified, with one point
requiring care: the balanced group-commutator factorization
needs the angle $\varphi$ solving
$\sin(\theta/2) = 2q^2\sqrt{1-q^4}$, $q = \sin(\varphi/2)$, and
we solve it by explicit bisection on the branch
$q \in (0, 2^{-1/4}]$ isolated in Lemma~\ref{lem:scale} rather
than by a closed-form approximation. The solver is validated to
worst-case error $1.5\times10^{-8}$ over $200$ random targets.
Single-qubit Haar targets are sampled by QR decomposition of
complex Gaussian matrices with the phase correction that makes the
distribution genuinely Haar~\cite{mezzadri}.

\emph{Exact arithmetic.} Free-product reduction, the
$\Zb[\omega]$ and $\Zb[\sqrt2]$ rings, the $SO(3)$ Bloch
representation, and the $n$-qubit Pauli channel are all
implemented over exact integers with the storage and reduction
conventions of Sec.~\ref{sec:prelim}, so that no floating-point
value enters any certificate. Floating point is used only to
\emph{check} the exact code, by comparison against direct matrix
evaluation, and the observed agreement, worst case
$3\times10^{-15}$ over $300$ random words, is a test of the
exact implementation rather than an ingredient of any result.

\emph{Multi-qubit pipeline.} The assembled circuits of
Theorem~\ref{thm:C} are produced by \texttt{qiskit}'s unitary
synthesis~\cite{qiskit}, transpiling to the $\{u, \mathrm{CX}\}$ basis, which
implements the two-qubit KAK decomposition at $n = 2$ and the
quantum Shannon decomposition at $n = 3$, with each single-qubit
rotation then passed to \texttt{pygridsynth}~2.0.0, an
implementation of the Ross--Selinger
algorithm~\cite{ross-selinger2016}, at the per-rotation accuracy
budget. Gate-order and endianness
conventions between the two libraries are a standard source of
silent error, so we validate the assembled circuit against
direct unitary evaluation at both $n = 2$ and $n = 3$
(Appendix~\ref{app:num}). Haar-random targets are drawn with
\texttt{qiskit}'s \texttt{random\_unitary} under fixed, disclosed
seeds.

Certificates at this scale need one further idea, because a
dense channel product costs $O(4^{3n})$ per gate and the
circuits carry thousands of gates. We exploit the structure
established in Sec.~\ref{sec:multiqubit}: Clifford channels are
signed permutations, and $\pi/8$-rotation channels are
block-sparse, acting as the identity on the commuting sector and
pairwise on the anticommuting one
(Lemma~\ref{lem:mqval}(a)). A structured accumulator using both
facts costs $O(4^{2n})$ per gate, and because a fast path that is
wrong would silently corrupt every rigidity number in this paper,
it is validated against the dense implementation
(Appendix~\ref{app:num}).

All optimizer comparisons use \texttt{full\_reduce} and
\texttt{teleport\_reduce} as shipped in the public
\texttt{PyZX} library~\cite{pyzx}, with no modification, so that
what we measure is the behaviour a practitioner would encounter.

\subsection{The exactness barrier (Theorem~\ref{thm:B})}

The barrier is a statement about two numbers, and this
subsection measures both. For each recursion depth $k$ we run
the SK recursion on the coarse net, take the synthesized word
$W_k$, and compute its exact floor
$\tmin(\llbracket W_k \rrbracket) = \mathrm{lde}(\hat W_k)$ by
the exact Bloch arithmetic of Lemma~\ref{lem:bridge}. Against
this we set the resynthesis benchmark: the leading-order
Ross--Selinger $T$-count $3\log_2(1/\epsilon_k)$ for a single
$z$-rotation at the accuracy $\epsilon_k$ that the SK word
actually achieves, measured per instance rather than assumed
from the convergence estimate. (A general element costs a
constant times this by the Euler decomposition of
Theorem~\ref{thm:B}; the constant shifts the benchmark
vertically and does not affect the growth comparison below.) Table~\ref{tab:barrier} and
Fig.~\ref{fig:barrier} report the comparison.

\begin{table}[htbp]
\centering
\caption{Exactness barrier: exact $\tmin$ floor vs.\ matched-accuracy
resynthesis, coarse net. $L$ is mean SK word length; $\sde$, the
mean smallest denominator exponent of the $\Zb[\omega]$ unitary
matrix, is reported as a diagnostic and distinguished from the
Bloch-image exponent $\mathrm{lde} = \tmin$ that certifies the
floor.}
\label{tab:barrier}
\begin{tabular}{ccccc}
\hline\hline
$k$ & $L$ & $\sde$ & $\tmin$ floor & sep.\ factor \\
\hline
2 & 116    & 12   & 22   & 1.9   \\
3 & 553    & 46   & 89   & 6.7   \\
4 & 2848   & 250  & 497  & 26.1  \\
5 & 14087  & 1232 & 2462 & 101.2 \\
\hline\hline
\end{tabular}
\end{table}

\begin{figure}[htbp]
\centering
\includegraphics[width=0.95\linewidth]{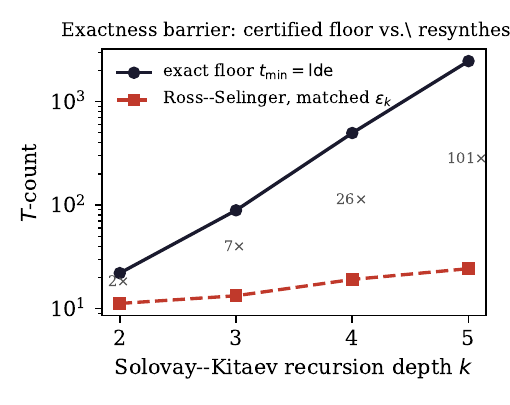}
\caption{The exactness barrier separates the certified exact floor
from approximation-aware resynthesis at matched accuracy, with a
geometrically growing gap (Theorem~\ref{thm:B}).}
\label{fig:barrier}
\end{figure}

The separation factor grows from $1.9\times$ at $k = 2$ to
$101\times$ at $k = 5$, by a factor of $3.5$--$3.9$ per level,
against the value $10/3$ that Theorem~\ref{thm:B}(b) predicts
from the measured floor growth. Three features of the table are
worth reading off explicitly, since each is a check on the
theory rather than a restatement of it.

First, the floor grows at the rate of the word, not of the
accuracy. Word length $L$ grows by $4.8$--$5.2$ per level, close
to the $5^k$ block count, and the floor tracks it at a
remarkably stable ratio: $\tmin/L$ sits at $0.16$--$0.19$ across
all four depths. This is the instance-level input that
Theorem~\ref{thm:B}(b) requires, the hypothesis
$\tmin = \Omega(c^k)$ with $c \approx 5 > 3/2$, and it is
measured rather than assumed.

Second, the floor is almost exactly twice the $\Zb[\omega]$
denominator exponent $\sde$ of the unitary itself, the ratio
$\tmin/\sde = \mathrm{lde}/\sde$ rising from $1.83$ at $k = 2$ to
$2.00$ at $k = 5$. The factor of two is the passage from the
$\Zb[\omega]$ representation of the unitary to the $\Zb[\sqrt2]$
representation of its Bloch image, under which denominators
square; the approach to exactly $2$ with depth indicates that the
sporadic reductions available at small exponents die out as the
words lengthen. As emphasized in the caption, $\sde$ is a
diagnostic only; the certificate is $\mathrm{lde}$.

Third, the resynthesis side behaves as it should. Inverting the
separation factors gives implied Ross--Selinger $T$-counts of
$11.6, 13.3, 19.0, 24.3$ at $k = 2, \ldots, 5$. These grow by a
factor of only $1.2$--$1.4$ per level, tracking
$3\log_2(1/\epsilon_k)$ with $\log(1/\epsilon_k) \sim (3/2)^k$:
resynthesis climbs geometrically in $k$, but at base $3/2$, while
the floor climbs at base $\approx 5$. The gap between them is the
ratio of the two, and its per-level growth of $3.5$--$3.9$ is
exactly the ratio of the two bases, $5/(3/2) = 10/3$, recovered
independently from each side of the table. The geometric growth
of the separation is therefore not an artifact of one side being
mismeasured: it is the arithmetic difference between a recursion
that lengthens words fivefold per level and an accuracy budget
that tightens at the Solovay--Kitaev rate.

One further measurement belongs here, though it tests
Theorem~\ref{thm:L1} rather than the barrier. On $24$ instances
spanning $k = 2$--$4$, with floors up to $568$, we compare the
$T$-count that \texttt{teleport\_reduce} actually reaches
against the exact $\tmin$ certificate. The optimality gap is
zero on every instance, without exception. This is a far more
demanding test than the planted family of
Sec.~\ref{sec:linkage}, whose words were built to exercise the
normalizer condition: these are real synthesis outputs, hundreds
to thousands of letters long, with no structure planted in them
at all, and the automated optimizer lands on the certified floor
every time.

\subsection{Compressibility plateau (Theorems~\ref{thm:A},
\ref{thm:Aprime})}

The plateau is the claim that Tier-1 compressibility is a
statistic of the net and not of the recursion depth, and the way
to see it is to vary the depth over as wide a range as the
arithmetic allows and watch the compression ratio not move.
Fig.~\ref{fig:plateau} plots $\gamma_1(W_k)$ against $k$ from $1$
to $6$ on both nets.

\begin{figure}[htbp]
\centering
\includegraphics[width=0.95\linewidth]{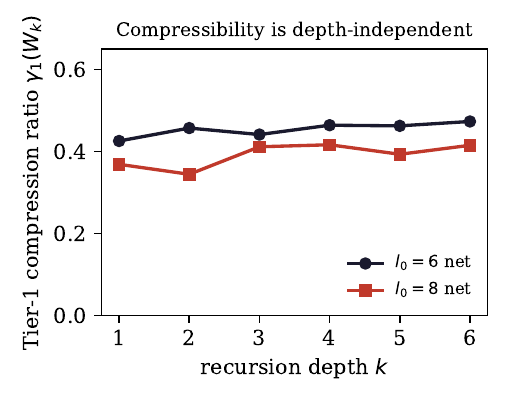}
\caption{Tier-1 compression ratio is depth-independent on both
nets, confirming Theorem~\ref{thm:Aprime}.}
\label{fig:plateau}
\end{figure}

Across those six levels the SK approximation error falls from
about $10^{-1}$ to about $10^{-6}$, five orders of magnitude, and
the word length grows from tens of letters to tens of thousands.
Over that entire range the compression ratio stays flat: at
$0.44$--$0.47$ on the coarse net and $0.34$--$0.42$ on the fine
net, a spread of two to four percentage points against a
five-order change in the quantity one might naively expect it to
track. This is the depth-independence that
Theorem~\ref{thm:Aprime} predicts, now visible directly rather
than inferred. The value of the plateau, and its dependence on the
net, are what the present theory supplies, and both nets confirm
the qualitative prediction that the compression fraction is set by
the net and not by the recursion depth.

The two plateaus differ, and the difference is itself a
prediction borne out. The coarse net, with the larger covering
radius, plateaus higher than the fine net, because a coarser
quantization leaves longer net words with more internal
redundancy for Tier-1 to remove. This is the qualitative content
of Theorem~\ref{thm:Aprime}: the limit is a functional of the
net's quantization measure, so two nets should give two
different plateaus, and they do. The quantitative version of the
same statement is the parameter-free test of
Corollary~\ref{cor:sqrtlaw} carried out at the end of
Sec.~\ref{sec:limit}, where the deep-leaf angle statistics of
these same two nets were found to land inside the windows their
covering radii predict, with nothing fitted.

Finally, the tiers above Tier-1 add little. Running the full
\texttt{teleport\_reduce} strategy, which brings local
complementation, pivoting, and phase teleportation to bear on
top of spider fusion, improves the $T$-count reduction by only
$3$--$9$ percentage points over Tier-1 across the same instances.
The bulk of what any exact strategy achieves on a
Solovay--Kitaev word is already achieved by the purely local
Tier-1 rules whose action Sec.~\ref{sec:tier1} characterizes
exactly, and the elaborate global machinery buys a modest
surplus on top. This is consistent with the single-qubit picture
of Sec.~\ref{sec:linkage}, where the extra power of the global
tiers was pinned down precisely as normalizer linkage: on these
long words the additional normalizer merges available beyond
local fusion are real but few.

\subsection{Multi-qubit rigidity (Theorem~\ref{thm:C})}

The single-qubit sections could exhibit an exact optimality gap of
zero because the minimal $T$-count was computable. At $n \ge 2$ it
remains computable in principle, by the exhaustive algorithm of
Gheorghiu, Mosca, and Mukhopadhyay~\cite{gmm2022}, but at a cost
exponential in the $T$-count itself, which puts the circuits
studied here, with $T$-counts in the hundreds to thousands, far
out of reach. The certificate of Corollary~\ref{cor:cert} is what
takes its place: for each synthesized circuit we compute
$\mathrm{lde}(\hat C)$ by accumulating the exact Pauli channel gate
by gate, which bounds from below the $T$-count of everything a
semantics-exact optimizer could return, and we compare it against
what \texttt{teleport\_reduce} actually achieves. The gap between
the two is the entire remaining headroom, and Theorem~\ref{thm:C}(c)
makes it a per-instance certified quantity rather than an estimate.

We run the full QSD+\texttt{gridsynth} pipeline on $143$
Haar-random targets across five configurations: $n = 2$ at
$\epsilon \in \{10^{-2}, 10^{-3}, 10^{-4}\}$ with fifty, fifty,
and fifteen targets, and $n = 3$ at
$\epsilon \in \{10^{-2}, 10^{-3}\}$ with twenty and eight targets,
the raw $T$-counts ranging from about $420$ to about $2570$.
Table~\ref{tab:e6} and Fig.~\ref{fig:e6} report the outcome, and
three findings follow from it, each an interpretation the theorem
licenses rather than a bare number.

\begin{table}[htbp]
\centering
\caption{Multi-qubit rigidity on the real QSD+\texttt{gridsynth}
pipeline (mean $\pm$ one standard deviation over $N$ Haar targets;
``incomp.''\ is the certified incompressible fraction
$\mathrm{lde}/T_{\mathrm{tel}}$, with the per-instance minimum in
brackets).}
\label{tab:e6}
\setlength{\tabcolsep}{3.2pt}
\begin{tabular}{cccccc}
\hline\hline
$n$ & $\epsilon$ & $N$ & $\bar T_{\mathrm{raw}}$ &
reduction (\%) & incomp.\ (\%) \\
\hline
2 & $10^{-2}$ & 50 & 419  & $3.6\pm2.2$ & $99.53\pm0.55$ [96.9] \\
2 & $10^{-3}$ & 50 & 644  & $2.6\pm1.4$ & $99.79\pm0.30$ [98.6] \\
2 & $10^{-4}$ & 15 & 851  & $1.8\pm0.7$ & $99.91\pm0.09$ [99.8] \\
3 & $10^{-2}$ & 20 & 1696 & $3.8\pm0.9$ & $99.43\pm0.32$ [98.5] \\
3 & $10^{-3}$ & 8  & 2570 & $2.1\pm0.7$ & $99.55\pm0.20$ [99.3] \\
\hline\hline
\end{tabular}
\end{table}

\begin{figure}[htbp]
\centering
\includegraphics[width=0.95\linewidth]{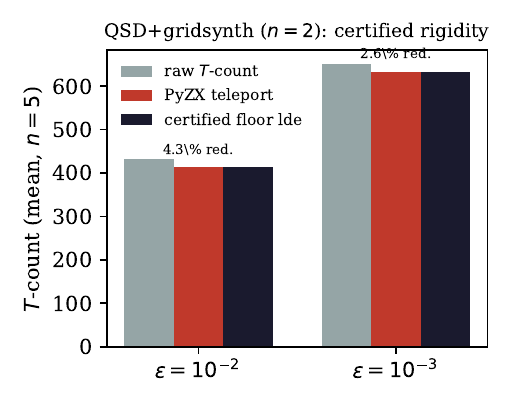}
\caption{Measured \texttt{PyZX} $T$-count reduction (red)
against the certified maximum possible exact reduction computed
from the channel lde (black), across all five pipeline
configurations (error bars: one standard deviation). The measured
reduction sits essentially on the certified ceiling everywhere:
rigidity is certified, not merely observed (Theorem~\ref{thm:C}).}
\label{fig:e6}
\end{figure}

First, the optimizer removes very little. The mean
\texttt{teleport\_reduce} reduction is $1.8$--$3.8\%$ across all
five configurations, which is the near-null behaviour that
Ref.~\cite{trasyn2025} observed for \texttt{gridsynth} circuits,
here confirmed on the assembled multi-qubit pipeline and extended
to $n = 3$. On its own this number is the kind of measurement that
invites the wrong two readings, that the optimizer is weak or that
the circuits happen to be hard, neither of which any amount of
further optimization could distinguish.

Second, the certificate settles which reading is correct, and the
answer is neither. The certified incompressible fraction
$\mathrm{lde}/T_{\mathrm{tel}}$ averages $99.4$--$99.9\%$ and never
falls below $96.9\%$ on any single instance, so the headroom left
to \emph{any} semantics-exact optimizer, not merely to
\texttt{PyZX}, is at most a few percent, and the measured reduction
sits essentially on the certified ceiling of
Theorem~\ref{thm:C}(c) in every configuration
(Fig.~\ref{fig:e6}). The gains are small because the circuits are
rigid, provably, and no optimizer will do better. The mechanism is
visible in the collapse rate that Theorem~\ref{thm:C}(b) identifies
as the only available source of compression: on the real grammar
it measures $1.0$--$2.3\%$, against $22$--$54\%$ on toy
commuting-parity grammars built to expose linkage, confirming that
the long Matsumoto--Amano blocks emitted by \texttt{gridsynth}
scramble channel residues far more thoroughly than a bare parity
ladder does, and leave almost no cross-block linkage for the
optimizer to find.

Third, rigidity tightens with accuracy, in the least convenient
direction. At $n = 2$ the mean incompressible fraction climbs from
$99.53\%$ at $\epsilon = 10^{-2}$ to $99.91\%$ at
$\epsilon = 10^{-4}$, while over the same range the collapse rate
falls from $2.2\%$ to $1.0\%$; the $n = 3$ configurations show the
same trend. The reason is structural: a tighter accuracy target
lengthens each \texttt{gridsynth} block, a longer block scrambles
its channel residue more completely, and a more scrambled residue
offers less cross-block linkage. The regime in which magic-state
overhead makes optimization matter most, deep circuits at tight
accuracy, is therefore exactly the regime in which the theory
guarantees there is least to gain.

\section{Discussion and Open Problems}
\label{sec:discussion}

The results of this paper are sharp where they are exact and
honest where they are not. This section separates the two,
setting out what remains open and why, and locating the work
relative to the neighbouring approaches it touches.

\emph{Hypothesis E and the limits of the limit law.} The
stationary limit law of Theorem~\ref{thm:Aprime} rests on two
kinds of input, and it is worth being explicit about which is
which. The dynamics of the residual scale are unconditional:
Lemma~\ref{lem:scale} and the path representation of
Lemma~\ref{lem:pathrep} are proven outright, and the
$\sqrt{}$-covering-radius prediction of
Corollary~\ref{cor:sqrtlaw} follows from them with no hypothesis
at all, which is why we could test it directly and
parameter-free in Sec.~\ref{sec:limit}. What Hypothesis~E adds is
directional ergodicity, the statement that the axis coordinate of
the recursion equidistributes rather than locking into a
measure-zero orbit. This is the same kind of assumption that
underlies the number-theoretic analyses of exact
synthesis~\cite{ross-selinger2016}, where the equidistribution of
the relevant Diophantine data is likewise taken as calibrated
rather than proven, and our empirical evidence for it is strong:
the plateau appears on both nets, at the predicted separation,
across five orders of magnitude in accuracy. A proof, as opposed
to a calibration, would most plausibly come from an
equidistribution theorem for the Voronoi quantization of the
specific net construction, a question in the metric theory of the
group rather than in quantum computation, and plausibly drawing
on the arithmetic machinery that yields optimal covering for
golden and super-golden gate
sets~\cite{lps1986,parzanchevski-sarnak}. We regard this as the
most interesting open mathematical problem the paper raises, and
we note that a negative result would not overturn the
certificates, which never invoke Hypothesis~E; it would only
demote the limit law from a theorem to a conjecture with strong
numerical support.

\emph{A closed characterization at $n \ge 2$.} The single-qubit
theory has a feature the multi-qubit theory lacks: a syntactic
criterion, membership in $N(Z)$, that decides collapse without
computing anything, and from which the closed-form $T$-count
formula of Theorem~\ref{thm:linkage} follows. Section~\ref{sec:multiqubit}
shows this cannot survive unchanged, since the identity
$\tmin = \mathrm{lde}$ already fails at the controlled-$S$ gate,
and with it any hope of reading the minimal $T$-count off the
word by a normalizer condition alone. What replaces it is the
certificate, which is unconditional and cheap but per-instance
rather than closed. Whether some closed characterization exists
at $n \ge 2$, perhaps one that is grammar-dependent, keyed to the
specific structure of QSD output rather than to arbitrary
Clifford+$T$ words, is open, and it is the natural next target for
anyone wishing to extend the single-qubit rigidity story to a
genuine multi-qubit formula. Our results neither establish such a
characterization nor rule it out; they show only that the naive
route through $N(Z)$ is closed.

\emph{Abstract strategy versus shipped implementation.} Theorem~\ref{thm:L1}
proves that the automated ZX strategies attain $\tmin$, and the
proof is about the abstract rule sets and their phase-table
semantics, exactly as specified in the
literature~\cite{duncan-kissinger2020,kissinger-tcount}. The
\texttt{PyZX} library implements those rules with particular
choices at the boundary, gadget conventions, the order in which
candidate rewrites are tried, and so on, and we do not audit the
source to prove that these choices conform to the abstract
specification. Instead we check conformance where it can be
checked exhaustively, on the $279$ instances of
Sec.~\ref{sec:linkage}, with agreement in every case. The
distinction matters because it locates precisely what is proven
and what is verified: the mathematics is a theorem about ZX
rewriting, and the claim that a specific released version of a
specific library realizes that mathematics is an empirical one,
held to the standard of exhaustive testing rather than of formal
verification.

\emph{Relation to phase squashing.} The exactness barrier of
Theorem~\ref{thm:B} says that no semantics-exact optimizer can
cross a floor set by the synthesized element. Kelly and
Kissinger's phase squashing~\cite{kelly-kissinger2025} is the
complementary move: it crosses the floor by giving up exactness,
rewriting a ZX diagram approximately with a rigorous bound on the
error introduced on diagrams with generalized flow. The two
results sit on opposite sides of the exact/approximate line this
paper draws. A synthesized circuit already carries an
approximation error from synthesis, so composing an
approximation-aware optimizer with a synthesis pass raises a
genuine question our exact analysis does not answer: how the two
error budgets should be jointly allocated, and whether phase
squashing applied to a synthesized circuit can recover part of
the gap that the exactness barrier forbids exact methods from
touching. Placing phase squashing as a third curve alongside the
exact floor and the resynthesis benchmark of
Fig.~\ref{fig:barrier} would make the trade-off visible, and is
the most natural experimental extension of this work.

\emph{Independent empirical confirmation.} The rigidity we prove
has a mirror in the recent synthesis literature. Hao, Xu, and
Tannu, whose main result is a tensor-network synthesis method
that improves on \texttt{gridsynth}, report as a subsidiary
finding (their RQ5) that applying \texttt{PyZX} to
already-synthesized circuits recovers almost nothing, and treat
this as an empirical observation without an underlying
cause~\cite{trasyn2025}. Our Theorem~\ref{thm:C} supplies the
cause: \texttt{gridsynth} emits Matsumoto--Amano normal forms,
each block sits on its own exact $T$-count floor by
Theorem~\ref{thm:B}, and a semantics-exact optimizer has only the
cross-block collapse of Corollary~\ref{cor:cert} left to harvest,
which we measure at a few percent at most. Their null result is
therefore not a limitation of \texttt{PyZX} but the barrier
operating at near-zero slack, and its magnitude is predicted
rather than merely reproduced. Read in the other direction, their
measurement is an independent, differently instrumented
confirmation of our multi-qubit certificates on real pipelines. The
same reading applies to the broader landscape of synthesis and
resynthesis tools~\cite{synthetiq,bqskit,fowler2011}: whatever
produces the circuit, once it is fixed the certificate bounds what
any exact optimizer can subsequently remove.

\emph{Other gate sets and architectures.} We have worked
throughout with Clifford+$T$ and with the exact ring
$\Zb[\omega]$ that makes its arithmetic finite. The valuation and
linkage machinery is not tied to that choice in principle: any
gate set whose exact synthesis lands in a ring with a denominator
valuation, such as the Clifford+$\sqrt{T}$ or $V$-gate libraries,
should admit an analogous floor, though the normalizer analysis
that gives the single-qubit formula would have to be redone for
each; the $V$-basis and the general approximation frameworks
built on other rings~\cite{bocharov-gurevich-svore,kbry2015,klmpp2023}
are the natural first targets. A separate and more practically pressing extension is
connectivity-constrained extraction, where the circuit must
respect a hardware coupling graph; the barrier is a statement
about the implemented element and is indifferent to how the
circuit is laid out, so the floor survives, but the achievable
side of the comparison would change and is left for future work.

\section{Conclusion}
\label{sec:conclusion}

We began with a puzzle from the compilation literature: the same
ZX optimizer, run on circuits computing the same unitaries to the
same accuracies, behaves in opposite ways depending only on which
algorithm synthesized its input, removing a stable fraction of a
Solovay--Kitaev circuit and almost nothing from a
\texttt{gridsynth} one. We have shown that these are not two
phenomena but one.

What made the unification possible was reading a synthesized
circuit as a word in a group rather than as a list of gates, so
that ZX simplification became the computation of a normal form
instead of a heuristic whose yield one measures. Tier-1
simplification computes precisely the normal form of
$\Zb_2 * \Zb_8$ (Lemma~\ref{lem:freeprod}), which makes the
compressibility of a Solovay--Kitaev word a combinatorial quantity
readable per instance. The step that follows was not visible from
the outset: if what an exact optimizer can remove is determined
algebraically, then so is what it cannot. Every synthesized
circuit carries the denominator exponent of the ring element it
implements, no semantics-exact optimizer can produce a $T$-count
below it, and that exponent is computable per instance from the
circuit alone (Theorem~\ref{thm:B}). With the floor in hand the
two contradictory behaviours resolve at once. A Solovay--Kitaev
word sits far above its floor, because the recursion multiplies
word length fivefold per level while the implemented element
wanders only as far as the accuracy budget allows, and the
distance between the two is what local rewriting harvests.
\texttt{gridsynth} output sits on its floor already, so an exact
optimizer has nothing to take. One barrier, two distances.

The mechanism at one qubit turned out to be identifiable exactly:
minimal $T$-count is a closed combinatorial function of phase
linkage through the $Z$-axis normalizer, and the automated ZX
strategies, which know nothing of rotation axes or denominator
exponents, compute that function every time. Where the closed
form fails, at two qubits and beyond, the same valuation still
yields unconditional per-instance rigidity certificates, and those
certificates explain and quantify a near-null optimization result
that the synthesis literature had recorded without a cause.

Two limits of scope are worth stating plainly. The first is
exactness: everything here concerns optimizers that preserve the
implemented element, and approximation-aware rewriting, of which
phase squashing~\cite{kelly-kissinger2025} is the current
instance, deliberately crosses the barrier by giving up that
constraint. The second is unitarity. The ZX-calculus is a
language for arbitrary linear maps, and much of what is done with
it, from measurement-based computation to the non-unitary
constructions used in graph algorithms~\cite{mansky2024}, lies
outside the circuit-to-circuit setting we analyse. Our floor is a
statement about elements of a group, and it says nothing about
diagrams that do not implement one.

Two threads run through what remains. The first is that exactness
is not a technical convenience but the source of the results: it
is because the arithmetic is exact that a floor can be certified
rather than estimated, that a lower bound can be shown attained
rather than merely approached, and that a compressibility fraction
can be predicted rather than fitted. The second is that the
resource cost of a synthesized circuit is legible in its algebra
before any optimizer is run, which is a statement about where
optimization can and cannot help: on the number-theoretically
synthesized circuits that dominate practice the available gains
are bounded by a quantity computable in advance, and that bound
tightens as accuracy tightens.

One gap remains, and it is a real one. Hypothesis~E asks that the
direction coordinate of the Solovay--Kitaev recursion
equidistribute over the net as the net is refined, which is a
question in ergodic theory rather than in quantum computation.
That the number-theoretic half of the synthesis literature rests
on assumptions of exactly this kind, the running-time guarantees
of Ross--Selinger among them, is not a coincidence: the
difficulty is intrinsic to the class of statement. Machinery
capable of proving such equidistribution exists, in the Hecke
operator methods of Lubotzky, Phillips, and
Sarnak~\cite{lps1986} and in the analysis of super-golden gate
sets~\cite{parzanchevski-sarnak}, but it draws on deep arithmetic
input and applies to gate sets with arithmetic structure, whereas
the nets used in practice, and in this paper, are generic
constructions. Bridging that distance is the most interesting
mathematical problem we leave open. It does not touch the
certificates, which never invoke the hypothesis.

Finally, the question this paper answers is the static one: for a
fixed target, how much of the synthesized word survives
rewriting. The dynamic question is untouched and, we think, more
interesting. A target moving continuously through $SU(2)$ does
not deform its synthesized word continuously; the word is a
discrete object and it jumps. The analysis of
Lemma~\ref{lem:nocollapse} already supplies the frame in which to
ask when: the recursion partitions $SU(2)$ into finitely many
cells on which every choice it makes is constant, so the word is
locally constant and changes only where a path crosses a cell
boundary, a locus of nearest-neighbour ties and solver branch
switches. What those boundaries look like, how a physically
meaningful path such as a Hamiltonian evolution meets them, and
whether the resulting jump statistics have their own limit law,
are questions the machinery developed here is built to address,
and they are questions about a picture. The object of study, in
the static case as in the dynamic one, is not the circuit but the
geometry that produced it.

\begin{acknowledgments}
The line of inquiry pursued here was prompted in part by an
unpublished draft of D.~De Silva on the numerical evaluation of
ZX-calculus optimization for Solovay--Kitaev synthesis, whose
questions motivated the theoretical treatment we develop.
C.-F.~K.\ thanks M.~Mansky for introducing him to the tools of the
ZX-calculus during a visit to the Ludwig-Maximilians-Universit\"at
M\"unchen, and S.\ Lorenzo for hosting the postdoctoral
appointment at the Universit\`a degli Studi di Palermo during
which much of this work was carried out. C.-F.K. is grateful to the European Union and the Region Reunion, France (POE FEDER 2021–2027, n°2025-0954-007180) for the funding support. K.D.S. is grateful to the European Union and the Region Reunion, France (POE FEDER 2021–2027, n°2025-0963-007179) for the funding support. A.M. is grateful to the European Union and the Region Reunion, France (POE FEDER 2021–2027, n°2026-0237-010722) for the funding support. PEACCEL is supported through a research program partially cofunded by the European Union (UE) and Region Reunion (FEDER).” 
\end{acknowledgments}

\section*{Author contributions}
C.-F.~K.\ conceived the project, developed the theory, wrote the
manuscript, and performed the numerical experiments.
A.~Mahasinghe and K.~De Silva advised on the initial ideas during
discussions at the University of Colombo and the University of
Western Australia, and reviewed the manuscript after the first
draft was completed. F.~Cadet and J.~Wang advised on the work, reviewed and edited the manuscript.
All authors read the manuscript.

\appendix

\section{The multiplicity audit for Lemma~\ref{lem:densweight}(b)}
\label{app:proofs}

This appendix proves the hard direction of
Lemma~\ref{lem:densweight}(b), the bound
$S \le 2\,\CM$ on the total pairwise saving. Throughout we use the
notation of Sec.~\ref{sec:tier1}: $v_1, \ldots, v_{N_k}$ are the
internally reduced blocks, $\CM$ the consumed mass,
$S = \sum_{i=1}^{N_k-1} s(v_i, v_{i+1})$ the total pairwise
saving, and, at a fixed junction $i$, the pairwise pops cancel or
merge the nested pairs $(x_m, y_m)$, $m = 1, \ldots, M_i$, with
$x_m$ the $m$-th syllable of $v_i$ from the right and $y_m$ the
$m$-th syllable of $v_{i+1}$ from the left. From Steps 1 and 2 of
the main-text proof, when $v_{i+1}$ arrives the stack ends in a
suffix of $v_i$ of length $\sigma_i \ge 0$, pristine except that
its deepest syllable may carry a phase altered by at most one
fusion during the absorption of $v_i$.

A pairwise reduction has a rigid depth structure, and the estimate
below leans on it at three separate points. Let $v$ and $w$ be
reduced words and reduce $vw$ by the stack algorithm, reading $w$
left to right against a stack initialised to $v$. Neither word
admits a rule internally, so every rule that fires does so between
the current stack top and the current letter of $w$, and the pops
therefore occur at consecutive depths: the first pairs the last
syllable $x_1$ of $v$ with the first syllable $y_1$ of $w$, and
the $m$-th pairs the $m$-th syllable $x_m$ of $v$ from the right
with the $m$-th syllable $y_m$ of $w$ from the left. Writing $M$
for the number of pops, the depths used are exactly
$1, \ldots, M$. Distinct depths consume distinct syllables of
$v$, whence
\begin{equation}
\label{eq:Mbound}
M \;\le\; \|v\| .
\end{equation}

Every pop but possibly the last is an annihilation. A merge
leaves a surviving fused syllable whose neighbours on both sides
have the type opposite to its own type's partner, so no further
rule can fire across it and the reduction halts. With $a$
annihilations and $\mu \in \{0,1\}$ merges we therefore have
$M = a + \mu$ and $s(v,w) = 2a + \mu$, and no single pop
contributes more than two units.

The global cascade at a junction absorbs a reduced block onto a
stack that is itself in normal form, so both facts apply to it
verbatim. We write $\hat a_j$ and $\hat\mu_j$ for its counts at
junction $j$, so that $\hat s_j = 2\hat a_j + \hat\mu_j$ with
$\hat\mu_j \in \{0,1\}$.

\begin{proposition}\label{prop:audit}
For every junction $i$,
\begin{equation}
\label{eq:perjunction}
s(v_i, v_{i+1}) \;\le\; \hat s_i + \hat s_{i-1},
\end{equation}
with the convention $\hat s_0 = 0$.
\end{proposition}

\begin{proof}
Write $M_i$ for the number of pops in the pairwise reduction of
$(v_i, v_{i+1})$.

\emph{Step 1: the length of the surviving suffix.} Consider the
global cascade at junction $i-1$, which absorbs $v_i$ onto the
stack. By the depth structure just described, its $\hat a_{i-1}$
annihilations consume the syllables of $v_i$ at positions
$0, \ldots, \hat a_{i-1}-1$ counted from the left, one syllable of
$v_i$ per annihilation. If $\hat\mu_{i-1} = 1$ the cascade then
merges the stack top with the syllable of $v_i$ at position
$\hat a_{i-1}$, replacing the two by a single syllable which
remains in the stack and after which the cascade stops. Nothing
else of $v_i$ is touched. Hence the portion of the stack
contributed by $v_i$ consists, when $\hat\mu_{i-1} = 1$, of the
merged syllable together with the $\|v_i\| - \hat a_{i-1} - 1$
untouched syllables to its right, and when $\hat\mu_{i-1} = 0$ of
the $\|v_i\| - \hat a_{i-1}$ untouched syllables at positions
$\hat a_{i-1}$ and beyond. In both cases its length is
\begin{equation}
\label{eq:sigma}
\sigma_i \;=\; \|v_i\| - \hat a_{i-1},
\end{equation}
and in both cases only the deepest of its syllables can differ
from the corresponding syllable of $v_i$, and then only in its
phase. That deepest syllable sits at depth $\sigma_i$ from the top
of the stack and occupies the position of $x_{\sigma_i}$, the
$\sigma_i$-th syllable of $v_i$ from the right, which is exactly
the syllable at position $\hat a_{i-1}$ from the left. For
$i = 1$ nothing precedes $v_1$, so $\hat a_0 = 0$, $\sigma_1 =
\|v_1\|$, and $\hat s_0 = 0$, consistently
with~\eqref{eq:sigma}. Note also that a merge leaves a syllable
behind, so $\hat\mu_{i-1} = 1$ forces $\sigma_i \ge 1$;
equivalently, $\sigma_i = 0$ implies $\hat\mu_{i-1} = 0$ and
$\hat a_{i-1} = \|v_i\|$.

\emph{Step 2: the two reductions agree above the boundary.} We
claim that for every $m < \sigma_i$ the pairwise reduction fires a
pop at depth $m$ if and only if the global cascade at junction $i$
does, with the same outcome. Argue by induction on $m$. Both
processes read $v_{i+1}$ left to right, so both meet $y_m$ as
their $m$-th incoming syllable. Suppose both have fired identical
pops at all depths below $m$; then both have consumed the same
prefix of $v_{i+1}$ and the same suffix of the stack, so the
syllable each now meets at depth $m$ is the one occupying that
position, namely $x_m$ for the pairwise reduction and, by Step 1,
the same $x_m$ for the global cascade whenever $m < \sigma_i$.
Identical stack syllable and identical incoming syllable produce
an identical pop, or identically no pop, which closes the
induction. Let $\Sigma$ denote the contribution common to both
from these depths, let $B$ and $\hat B$ denote the pairwise and
global contributions at depth $\sigma_i$, and let $D$ denote the
pairwise contribution from depths $m > \sigma_i$. Then
\begin{equation}
\label{eq:split}
s(v_i,v_{i+1}) = \Sigma + B + D,
\qquad
\hat s_i \;\ge\; \Sigma + \hat B,
\end{equation}
the inequality because the global cascade may continue past depth
$\sigma_i$ into stack material older than $v_i$, which only adds
to $\hat s_i$.

\emph{Step 3: the deep pops are paid for by the previous
junction.} The pops contributing to $D$ occur at the consecutive
depths $\sigma_i + 1, \ldots, M_i$, so there are $M_i - \sigma_i$
of them, a quantity we may take to be positive since otherwise
$D = 0$. By~\eqref{eq:Mbound} and~\eqref{eq:sigma},
\begin{equation}
M_i - \sigma_i \;\le\; \|v_i\| - \sigma_i \;=\; \hat a_{i-1},
\end{equation}
and since each pop contributes at most two units,
\begin{equation}
\label{eq:deep}
D \;\le\; 2\hat a_{i-1} \;\le\; 2\hat a_{i-1} + \hat\mu_{i-1}
\;=\; \hat s_{i-1} .
\end{equation}
This is the step at which~\eqref{eq:sigma} does its work: had a
merge at junction $i-1$ shortened the remnant by a further
syllable, the right-hand side of~\eqref{eq:deep} would have been
$\hat s_{i-1} + 1$ and the estimate below would fail.

\emph{Step 4: the boundary.} If $\sigma_i = 0$ then $\Sigma = 0$,
there is no depth $\sigma_i$, every pop is deep, and by Step 1
$\hat a_{i-1} = \|v_i\|$, so $s(v_i,v_{i+1}) = D \le
2\hat a_{i-1} \le \hat s_{i-1}$ by~\eqref{eq:deep}, which
gives~\eqref{eq:perjunction}. If $0 < M_i < \sigma_i$ the pairwise
reduction halts at a depth strictly inside the region where Step 2
applies, so the global cascade halts there too and
$s(v_i,v_{i+1}) = \Sigma = \hat s_i$, which again
gives~\eqref{eq:perjunction}.

There remains the case $M_i \ge \sigma_i \ge 1$. The pairwise
reduction pops at depth $\sigma_i$, so $x_{\sigma_i}$ and
$y_{\sigma_i}$ are of the same type. By Step 1 the syllable the
global cascade meets at depth $\sigma_i$ is $x_{\sigma_i}$, with
its phase possibly altered; since a merge combines two $Z$
syllables into a $Z$ syllable and $H$ syllables never merge, the
type is unaltered. The global cascade therefore fires at depth
$\sigma_i$ as well, whether by annihilation or by merge, so
\begin{equation}
\label{eq:Bhat}
\hat B \;\ge\; 1 .
\end{equation}
Two subcases remain.

If $\hat\mu_{i-1} = 0$ then by Step 1 the whole remnant is
untouched, so the syllable at depth $\sigma_i$ is $x_{\sigma_i}$
itself and the argument of Step 2 extends one depth further: the
global pop at depth $\sigma_i$ is the pairwise one and
$\hat B = B$. With~\eqref{eq:deep},
\begin{equation}
B + D \;\le\; \hat B + \hat s_{i-1} .
\end{equation}

If $\hat\mu_{i-1} = 1$ then $\hat s_{i-1} = 2\hat a_{i-1} + 1$, so
$D \le 2\hat a_{i-1} = \hat s_{i-1} - 1$ by~\eqref{eq:deep}, while
$B \le 2$ because a single pop contributes at most two units.
With~\eqref{eq:Bhat},
\begin{equation}
B + D \;\le\; 2 + \hat s_{i-1} - 1 \;=\; \hat s_{i-1} + 1
\;\le\; \hat B + \hat s_{i-1} .
\end{equation}

In both subcases, substituting into~\eqref{eq:split},
\begin{equation}
s(v_i,v_{i+1}) = \Sigma + B + D
\;\le\; \Sigma + \hat B + \hat s_{i-1}
\;\le\; \hat s_i + \hat s_{i-1} . \qedhere
\end{equation}
\end{proof}

\begin{corollary}\label{cor:audit}
$S \le 2\,\CM - \hat s_{N_k-1}$, and consequently
$|R| \le \CM$.
\end{corollary}

\begin{proof}
The junctions are indexed $i = 1, \ldots, N_k - 1$ and
$\CM = \sum_{i=1}^{N_k-1}\hat s_i$.
Summing~\eqref{eq:perjunction} over them,
\begin{align}
S &\;\le\; \sum_{i=1}^{N_k-1}\hat s_i
\;+\; \sum_{i=1}^{N_k-1}\hat s_{i-1}
\;=\; \CM \;+\; \sum_{j=0}^{N_k-2}\hat s_j \notag\\
&\;=\; \CM + \big(\CM - \hat s_{N_k-1}\big),
\end{align}
using $\hat s_0 = 0$ in the reindexed sum. This is the first
assertion. For the second, Lemma~\ref{lem:densweight}(a) gives
$R = \CM - S$, so the bound just proved yields
\begin{align}
R = \CM - S &\;\ge\; \CM - \big(2\,\CM - \hat s_{N_k-1}\big)
\notag\\
&\;=\; -\CM + \hat s_{N_k-1} \;\ge\; -\CM ,
\end{align}
the last step because $\hat s_{N_k-1} \ge 0$. In the opposite
direction $S \ge 0$ gives $R = \CM - S \le \CM$. The two bounds
together read $-\CM \le R \le \CM$, which since $\CM \ge 0$ is
precisely $|R| \le \CM$.
\end{proof}

Neither direction of Corollary~\ref{cor:audit} admits a smaller
constant, and it is the negative one that binds. Take $2n$ blocks
alternating $Z_1$ and $Z_7$. Globally the blocks annihilate in
consecutive pairs, so $\CM = 2n$ and $\hat s_{N_k-1} = 2$, while
each of the $2n-1$ junctions has pairwise saving $2$ and hence
$S = 4n-2$. This meets the corollary with equality and gives
$|R|/\CM = 1 - 1/n$, which approaches $1$.

The positive side is not vacuous either, although it is never
attained. Values $R > 0$ arise whenever a block is consumed
entirely and the global cascade cancels straight across the gap
it leaves, an event that no pairwise term can register, so the
two-sided form of the statement is doing real work.

Two degenerate configurations are worth following through the
argument rather than excluding by hypothesis. If $\CM = 0$ then
nothing was removed anywhere, so no block was consumed and no
phase altered, and the stack holds every block pristine. Every
pairwise pop would then have been replicated by Step 2 of the
proof, and a replicated pop removes letters, so there can have
been no pops at all: $S = 0$ and $R = 0$. If instead some reduced
block is empty, it generates no pairwise pops on either side,
because $s(v,w)$ vanishes when either argument does, while the
global process may cancel straight across the gap. That inflates
$\hat s$ relative to the pairwise sum, and the easy direction
$R \le \CM$ absorbs it.

An earlier version of this audit charged pairwise savings to
individual letters and bounded the load a single letter could
carry, which yielded only $|R| \le 3\,\CM$. Two changes produce
the sharp constant. The first is to count per junction rather
than per letter, so that the deep pops at a junction are paid for
in bulk by the cascade at the junction before it instead of one
letter at a time. The second is the accounting recorded
in~\eqref{eq:sigma}: a merge leaves the remnant of the incoming
block one syllable longer than a per-letter count suggests, and
that single syllable is exactly the margin the estimate needs.

The numerical record sits inside the proven constant. Across the
$3065$ instances of the verification suite described in
Sec.~\ref{sec:tier1}, real SK runs on both nets at $k = 1$--$5$
together with $3000$ adversarial synthetic block sequences
including the family that falsified the earlier block-level
bound, the worst observed ratio is $|R|/(\CM + 1) = 0.909$. The
alternating family above shows what a sequence approaching the
ceiling has to look like, and no SK instance we generated
resembles it. (The $R = 0$ induction for Theorem~\ref{thm:A},
formerly deferred to this appendix, is proven in full in the main
text.)

\section{ZX conventions and the rule set}
\label{app:zx}

This appendix fixes the conventions used in the main text and
records the rewrite rules invoked in the proofs, in the form in
which we use them. It is a reference rather than an exposition;
for the calculus itself we refer to the standard
sources~\cite{coecke-duncan2011,coecke-kissinger-book,vandewetering-review},
and for the completeness of the Clifford+$T$ fragment to
Ref.~\cite{jeandel-perdrix-vilmart}.

\emph{Generators.} A ZX diagram is built from $Z$ spiders (green),
$X$ spiders (red), Hadamard boxes (yellow), and wires. The $Z$
spider with $m$ inputs, $n$ outputs, and phase $\alpha$ denotes
$|0\rangle^{\otimes n}\langle 0|^{\otimes m} + e^{i\alpha}
|1\rangle^{\otimes n}\langle 1|^{\otimes m}$, and the $X$ spider is
its Hadamard conjugate. A single-qubit Clifford+$T$ word translates
gate by gate, as in Fig.~\ref{fig:zxrules}(a): $T$, $S$, and $Z$
become one-input one-output $Z$ spiders of phases $\pi/4$, $\pi/2$,
and $\pi$, written $Z_1$, $Z_2$, $Z_4$ in the $\Zb_8$ labelling of
Sec.~\ref{sec:prelim}, and $H$ becomes a Hadamard box. Diagrams are
read left to right and all equalities are up to a nonzero scalar
and a global phase, which is why $\doteq$ rather than $=$ appears
in Sec.~\ref{sec:linkage}.

Two further identities are used in the main text, and both are
easier to read as pictures than as formulas
(Fig.~\ref{fig:zxid}). The colour change rule conjugates a spider
by Hadamard boxes on all of its legs and exchanges its colour,
which is what allows an arbitrary diagram to be brought to the
all-$Z$ form defined below. The Euler decomposition writes a
Hadamard box as three spiders of phase $\pi/2$ alternating in
colour; read as an operator identity it expresses an arbitrary
$SU(2)$ element as $R_z(\alpha) H R_z(\beta) H R_z(\gamma)$ up to
global phase, and the proof of Theorem~\ref{thm:B} uses it in
that form to reduce general single-qubit synthesis to
$z$-rotation synthesis.

\begin{figure}[t]
\centering
\begin{tikzpicture}[scale=0.90,every node/.style={transform shape}]
\node[anchor=west,font=\footnotesize\itshape] at (-0.15,1.15)
  {(a) colour change};
\node[zxX] (x) at (0.85,0.35) {$\alpha$};
\draw[zxwire] (0.15,0.65)--(x); \draw[zxwire] (0.15,0.05)--(x);
\draw[zxwire] (x)--(1.55,0.35);
\node[font=\small] at (2.0,0.35) {$=$};
\node[zxH] (h1) at (2.6,0.65) {};
\node[zxH] (h2) at (2.6,0.05) {};
\node[zxZ] (z) at (3.45,0.35) {$\alpha$};
\node[zxH] (h3) at (4.3,0.35) {};
\draw[zxwire] (2.15,0.65)--(h1); \draw[zxwire] (2.15,0.05)--(h2);
\draw[zxwire] (h1)--(z); \draw[zxwire] (h2)--(z);
\draw[zxwire] (z)--(h3); \draw[zxwire] (h3)--(4.85,0.35);

\node[anchor=west,font=\footnotesize\itshape] at (-0.15,-0.85)
  {(b) Euler decomposition of the Hadamard box};
\node[zxH] (hh) at (0.85,-1.65) {};
\draw[zxwire] (0.25,-1.65)--(hh); \draw[zxwire] (hh)--(1.45,-1.65);
\node[font=\small] at (1.9,-1.65) {$=$};
\node[zxZ] (e1) at (2.6,-1.65) {$\tfrac{\pi}{2}$};
\node[zxX] (e2) at (3.7,-1.65) {$\tfrac{\pi}{2}$};
\node[zxZ] (e3) at (4.8,-1.65) {$\tfrac{\pi}{2}$};
\draw[zxwire] (2.05,-1.65)--(e1); \draw[zxwire] (e1)--(e2);
\draw[zxwire] (e2)--(e3); \draw[zxwire] (e3)--(5.35,-1.65);
\node[font=\footnotesize,anchor=west] at (5.7,-1.65) {(up to a scalar)};
\end{tikzpicture}
\caption{The two auxiliary identities used in the main text.
(a) Colour change: conjugating a spider by Hadamard boxes on every
leg exchanges red for green and leaves the phase unaltered, which
is how a diagram is brought to the all-$Z$ graph-like form.
(b) Euler decomposition: a Hadamard box equals three alternating
spiders of phase $\pi/2$; read as an operator identity this is the
statement that any $SU(2)$ element factors into three
$z$-rotations conjugated by Hadamards, which the proof of
Theorem~\ref{thm:B} uses to pass from $z$-rotation synthesis to
general single-qubit synthesis.}
\label{fig:zxid}
\end{figure}
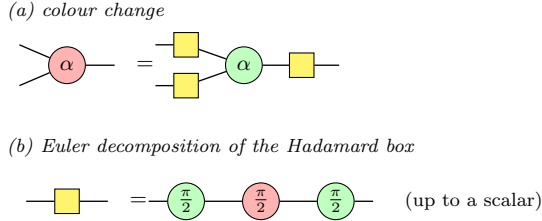

\emph{Tier 1.} Two adjacent spiders of the same colour joined by
one or more wires fuse into a single spider whose phase is the sum
of theirs. A spider of phase $0$ and degree two is the identity and
may be deleted, its two wires joined. These are the only rules used
in Sec.~\ref{sec:tier1}, and Lemma~\ref{lem:freeprod} identifies
their joint action on single-qubit words with the free-product
normal form.

\emph{Graph-like form.} A diagram is graph-like when every interior
spider is a $Z$ spider, all interior connections are Hadamard
edges, the underlying graph is simple, and each boundary vertex
meets exactly one spider, which meets at most one boundary
vertex. Every diagram can be brought to this form by colour change
and by the Hadamard-edge conventions, and the strategies of Tier 2
and Tier 3 operate on diagrams in this form.

\emph{Tier 2.} Local complementation at an interior spider $u$
carrying phase $\pm\pi/2$ deletes $u$, complements the subgraph
induced on its neighbourhood, and shifts the phases of the
neighbours by $\mp\pi/2$. Pivoting along an edge $uv$ between two
interior spiders of phase in $\{0,\pi\}$ is the composite
$\star u \star v \star u$ of three local complementations, and
deletes both endpoints. Both rules strictly decrease the number of
interior spiders, which is what drives the termination argument
used in the proof of Theorem~\ref{thm:L1}, and the general
simplification strategy assembled from them, together with the
statement that no interior proper Clifford spider survives, is
Theorem~5.4 of Ref.~\cite{duncan-kissinger2020}.

One feature of these rules matters for Theorem~\ref{thm:L1} and
is worth isolating. Neither rule inspects the value of a
non-Clifford phase: local complementation shifts the phases of the
neighbours of $u$ by $\mp\pi/2$ and pivoting shifts phases by
multiples of $\pi$, so a non-Clifford phase is only ever
translated by a Clifford amount or negated. A phase carried as a
formal variable therefore passes through the entire Tier-2
simplification with its identity intact, altered at most by a
sign and by added Clifford constants, which is what makes the
parametrization argument of Sec.~\ref{sec:linkage} legitimate.

\emph{Tier 3.} A phase gadget carries a non-Clifford phase on a
spider attached to the diagram through a pair of Hadamard edges
rather than sitting on a wire. Phase teleportation, implemented as
\texttt{teleport\_reduce} in \texttt{PyZX}~\cite{pyzx,kissinger-tcount},
parametrizes the non-Clifford phases, simplifies the Clifford
skeleton with the Tier-2 rules while the phases ride along as
gadgets, and records in a phase table which parameters have been
brought into a common slot and with which relative signs. It is
this table that Theorem~\ref{thm:L1} identifies with the linkage
procedure of Theorem~\ref{thm:linkage}.

\emph{Extraction.} Turning a simplified diagram back into a
circuit is the one step of the pipeline that is not automatic. A
graph-like diagram carries no notion of time order, so recovering
a gate sequence requires finding one, and this is possible when
the diagram admits a \emph{generalized flow}: an assignment to
each interior spider $u$ of a set $g(u)$ of spiders lying later in
some partial order, such that correcting $u$ propagates only
forward. Figure~\ref{fig:extract} shows the passage from circuit
to graph-like form and the flow condition. The simplification
strategies of Tiers 2 and 3 are designed to preserve the existence
of such a flow, which is why their output can be extracted at all;
the general analysis, including the cases where extraction becomes
hard, is that of Refs.~\cite{duncan-kissinger2020,backens-extraction},
and the same structure is what Kelly and Kissinger's phase
squashing requires~\cite{kelly-kissinger2025}.

Our results use none of this machinery directly. All that the
proofs need is the conclusion: the extraction step returns a
circuit denoting the diagram's linear map, so that the composite
of simplification and extraction preserves the implemented element
and is semantics-exact in the sense of Sec.~\ref{sec:barrier}. In
particular the barrier of Theorem~\ref{thm:B} applies to any
pipeline with this property, whether or not it proceeds through
graph-like diagrams and flows at all.

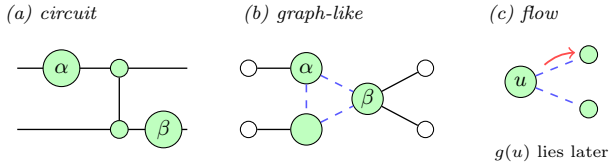
\begin{figure}[t]
\centering
\begin{tikzpicture}[scale=0.90,every node/.style={transform shape}]
\node[anchor=west,font=\footnotesize\itshape] at (-0.2,1.35) {(a) circuit};
\draw[zxwire] (0.1,0.55)--(2.6,0.55);
\draw[zxwire] (0.1,-0.35)--(2.6,-0.35);
\node[zxZ] (c1) at (0.75,0.55) {$\alpha$};
\node[zxdot] (c2) at (1.6,0.55) {};
\node[zxdot] (c3) at (1.6,-0.35) {};
\draw[zxwire] (c2)--(c3);
\node[zxZ] (c4) at (2.25,-0.35) {$\beta$};

\node[anchor=west,font=\footnotesize\itshape] at (3.3,1.35) {(b) graph-like};
\node[bnd] (i1) at (3.5,0.55) {};
\node[bnd] (i2) at (3.5,-0.35) {};
\node[zxZ,minimum size=4.6mm] (g1) at (4.35,0.55) {$\alpha$};
\node[zxZ,minimum size=4.6mm] (g2) at (4.35,-0.35) {};
\node[zxZ,minimum size=4.6mm] (g3) at (5.25,0.1) {$\beta$};
\node[bnd] (o1) at (6.1,0.55) {};
\node[bnd] (o2) at (6.1,-0.35) {};
\draw[zxwire] (i1)--(g1); \draw[zxwire] (i2)--(g2);
\draw[hedge] (g1)--(g3); \draw[hedge] (g2)--(g3);
\draw[hedge] (g1)--(g2);
\draw[zxwire] (g3)--(o1); \draw[zxwire] (g3)--(o2);

\node[anchor=west,font=\footnotesize\itshape] at (6.9,1.35) {(c) flow};
\node[zxZ,minimum size=4.6mm] (u) at (7.5,0.35) {$u$};
\node[zxdot] (w1) at (8.5,0.75) {};
\node[zxdot] (w2) at (8.5,-0.05) {};
\draw[hedge] (u)--(w1); \draw[hedge] (u)--(w2);
\draw[->,line width=0.7pt,draw=red!70] (7.85,0.6) to[bend left=18] (8.3,0.78);
\node[font=\scriptsize,anchor=west] at (7.0,-0.7) {$g(u)$ lies later};
\end{tikzpicture}
\caption{From circuit to graph-like diagram and back. (a) A
two-qubit Clifford+$T$ circuit as a ZX diagram, with time running
left to right. (b) The same diagram in graph-like form: every
interior spider is a $Z$ spider, interior connections are Hadamard
edges (dashed), and the temporal layout has been discarded, so a
gate order must be recovered before a circuit can be read off.
(c) A generalized flow supplies that order by assigning to each
interior spider $u$ a correction set $g(u)$ of spiders lying later
in a partial order; the simplification strategies preserve the
existence of such a flow, which is what makes their output
extractable.}
\label{fig:extract}
\end{figure}

\section{Numerical methodology and validation}
\label{app:num}

Every number in this paper comes from code written for it, and
this appendix records what that code is and how each piece was
checked before being trusted. The organising principle is that no
component was used in a reported result until it had been
validated against an independent computation of the same
quantity, either a floating-point evaluation, an exhaustive
enumeration, or a second implementation.

\emph{Software.} \texttt{numpy}~\cite{numpy},
\texttt{PyZX}~0.10.4~\cite{pyzx}, \texttt{qiskit}~2.5.0~\cite{qiskit},
\texttt{pygridsynth}~2.0.0, and exact-arithmetic modules
implemented for this work.

\emph{Exact arithmetic.} Four components carry the exact
arithmetic and each was validated separately. The free-product
reducer was checked for semantic correctness against direct matrix
multiplication, agreement up to global phase with squared-trace
tolerance below $10^{-15}$, and for normal-form correctness,
meaning strict alternation of syllable types and no zero phases,
on $2000$ random words. The $\Zb[\omega]$ arithmetic was compared
with floating-point evaluation on $300$ random words, with
worst-case discrepancy $3\times10^{-15}$, and was confirmed
invariant under Tier-1 reduction on $100$ further words, which
tests the reducer and the arithmetic against each other. The
$SO(3)$ Bloch representation was validated against floating point
on $100$ random words, and the projective Clifford group it
produces was confirmed to have order exactly $24$. The
$n$-qubit Pauli channel of Sec.~\ref{sec:multiqubit} was validated
against the single-qubit Bloch representation on $120$ random
words, agreeing in every entry.

\emph{Synthesis.} The balanced group-commutator solver, which
inverts $\sin(\theta/2) = 2q^2\sqrt{1-q^4}$ by bisection on the
branch isolated in Lemma~\ref{lem:scale}, was validated to
worst-case error $1.5\times10^{-8}$ on $200$ random targets. For
the multi-qubit pipeline we confirmed that every
\texttt{gridsynth} output arising in our experiments is in
Matsumoto--Amano normal form, the hypothesis under which
Theorem~\ref{thm:C}(a) applies, by parsing each returned word
against the grammar $(T|\varepsilon)(HT|SHT)^{\ell}\,c$; no
exception occurred. Assembled circuits were checked against
direct unitary evaluation at $n = 2$ and $n = 3$, with accumulated
distances consistent with the per-rotation $\epsilon$ budget,
which is what detects the gate-order and endianness mismatches
that silently corrupt cross-library pipelines.

\emph{Ground truth for Lemma~\ref{lem:bridge}.} Matsumoto--Amano
normal forms $(T|\varepsilon)(HT|SHT)^{\ell}\,c$ with
$c \in \mathcal{C}_1$ were enumerated exhaustively for
$t \le 8$, yielding $18{,}384$ distinct projective elements. The
histogram of elements by $T$-count is
$24\cdot[1,3,6,12,24,48,\ldots]$, matching the count predicted by
the normal-form characterization, which is itself a check on the
enumerator. The identity $\tmin = \mathrm{lde}$ was then tested
on every one of these elements, with no mismatch.

\emph{Net regularity.} On both nets the regularity assumption
$\ell_{\min} \ge 2D^*+1$ of Theorem~\ref{thm:A} fails for
$k \ge 2$. Empty and near-empty blocks occur at a stationary
density of $7$--$20\%$ at every depth $k = 1$--$6$, and the number
of degenerate junctions, where adjacent blocks are mutually
inverse, grows to about $300$ at $k = 6$ on the coarse net. This
is the measurement that motivates Theorem~\ref{thm:Aprime} and
Lemma~\ref{lem:densweight}: the clean theorem does not apply to
the instances one actually generates.

\emph{Multiplicity audit.} The per-junction
inequality~\eqref{eq:perjunction} underlying
Corollary~\ref{cor:audit} was checked on $3\times10^5$ random
sequences of internally reduced blocks, with up to nine blocks of
up to six syllables each, and exhaustively on all block triples
over the $22$ reduced words of length at most two and on all
quintuples of single-syllable blocks. No violation was found. The
alternating family exhibited in Appendix~\ref{app:proofs} meets
the bound with equality at every length tested, which confirms
that the search was reaching the extremal region rather than
sampling a slack interior.

\emph{Valuation lemma.} The mod-$2$ residue automaton of
Lemma~\ref{lem:valuation} was explored by breadth-first search.
States are pairs $(\bar P, \bar Q)$ of $3\times3$ matrices over
$\mathbb{F}_2$, seeded from the residues of the $24$ Clifford
prefixes, which collapse to six distinct states because the signs
of a signed permutation vanish modulo $2$. One transition is
counted for each triple of state, admissible interior Clifford,
and odd phase, although the residue action does not depend on the
phase; the counts reported in Sec.~\ref{sec:linkage} follow this
convention.

Because that computation verifies an intermediate representation
rather than the conclusion, we also checked the conclusion
directly on words. Using exact $\Zb[\sqrt2]$ arithmetic in the
form~\eqref{eq:ringform}, reducing after every product, we
generated $3200$ separated words, $400$ at each length
$n = 1, \ldots, 8$, with uniformly random odd phases, uniformly
random interior Cliffords outside $N$, and uniformly random outer
Cliffords. The lde equalled $n$ in every case. Lifting the
normalizer restriction, so that interior Cliffords range over all
of $\mathcal{C}_1$, gives $\mathrm{lde} < n$ for $69\%$ of $2400$
words, and the deficit $n - \mathrm{lde}$ is always even,
consistent with phases annihilating in pairs as the transport
picture of Theorem~\ref{thm:linkage} predicts.

\emph{Collapse criterion and optimality.} The refinement in
Lemma~\ref{lem:mqval}(b), that a criterion-positive step lowers
the exponent by exactly one rather than by more, was checked on
$400$ random steps at $n = 1$ and $150$ at $n = 2$, with no
exception; as noted in Sec.~\ref{sec:multiqubit}, no downstream
result depends on it. Separately, the claim of
Theorem~\ref{thm:L1} that automated ZX simplification attains
$\tmin$ was tested on the planted family
$\{T\,c\,T, T\,c\,T^\dagger : c \in \mathcal{C}_1\}$, on the SK
instances of Sec.~\ref{sec:barrier} with $\tmin$ up to $568$, on
adversarial $\pi$-blocker suites, and on random words, giving
$279$ exact matches in $279$ trials.

\emph{Certificates at scale.} The structured accumulator
described in Sec.~\ref{sec:numerics}, which reduces the cost of a
channel product from $O(4^{3n})$ to $O(4^{2n})$ per gate, was
validated bit for bit against the dense implementation on random
$n = 2$ circuits: for each circuit the two routines were required
to return identical integer numerators and identical exponents at
every gate, not merely equal lde at the end, so that a
discrepancy anywhere along the accumulation would be caught.

\end{document}